\documentclass[journal,onecolumn]{IEEEtran}

\usepackage{amsmath,amssymb,amsthm,mathtools}
\usepackage{cite}
\usepackage{url}
\usepackage{booktabs}
\usepackage{graphicx} 
\usepackage{placeins}
\usepackage{pgfplots}
\usepgfplotslibrary{groupplots}
\pgfplotsset{compat=1.17}

\makeatletter
\newcommand{\includefigure}[2][]{%
  \IfFileExists{#2}{\includegraphics[#1]{#2}}{%
    \filename@parse{#2}%
    \IfFileExists{\filename@base.\filename@ext}{\includegraphics[#1]{\filename@base.\filename@ext}}{%
      \fbox{\parbox[c][0.22\textheight][c]{0.9\linewidth}{\centering Missing figure file:\\[0.3em]\texttt{\detokenize{#2}}}}%
    }%
  }%
}
\makeatother

\usepackage{tikz}
\usepackage{xcolor}
\usetikzlibrary{positioning,arrows.meta}

\theoremstyle{plain}
\newtheorem{theorem}{Theorem}[section]
\newtheorem{lemma}[theorem]{Lemma}
\newtheorem{proposition}[theorem]{Proposition}
\newtheorem{corollary}[theorem]{Corollary}

\theoremstyle{definition}
\newtheorem{definition}[theorem]{Definition}
\newtheorem{example}[theorem]{Example}

\newtheorem{assumption}[theorem]{Assumption}
\theoremstyle{remark}
\newtheorem{remark}[theorem]{Remark}

\newenvironment{discussion}[1]{\par\medskip\noindent\textit{#1.}\ }{\par\medskip}

\newcommand{\hbin}[1]{h_2\!\left(#1\right)}
\newcommand{\Cch}{C_{\mathrm{ch}}}
\newcommand{\Cgate}{C_{\mathrm{gate}}}
\newcommand{\Cgateeff}{C_{\mathrm{gate,eff}}}
\newcommand{\Rsup}{R_{\mathrm{sup}}}
\newcommand{\Rsupb}{R_{\mathrm{sup}}^{(b)}}
\newcommand{\Rsuphard}{R_{\mathrm{sup}}^{(\mathrm{hard})}}
\newcommand{\Lif}{L_{\mathrm{if}}}

\newcommand{\condch}{n\Cch \le m\Cgate}
\newcommand{\condcomp}{m\Cgate \le n\Cch}
\newcommand{\regch}{channel-limited regime ($\condch$)}
\newcommand{\regcomp}{compute-limited regime ($\condcomp$)}

\tikzset{
  vtx/.style={circle,draw,thick,minimum size=5.2mm,inner sep=0pt,font=\small},
  elbl/.style={draw=none,fill=white,inner sep=1pt,font=\scriptsize},
}

\title{Reliable Remote Inference\\ from Unreliable Components:\\ Joint Communication and Computation Limits}

\author{Zhenyu~Liu, Yi~Ma, and Rahim~Tafazolli%
\thanks{Z. Liu, Y. Ma, and R. Tafazolli are with the 6GIC, Institute for Communication Systems, University of Surrey, Guildford, United Kingdom, GU2 7XH (e-mail: \{zhenyu.liu, y.ma, r.tafazolli\}@surrey.ac.uk).}}

\begin{document}
\maketitle


\begin{abstract}
Classical information theory typically assumes reliable receiver-side processing. We study remote inference when communication is noisy and the receiver itself is built from unreliable components under a finite redundancy budget. Under a committed/no-bypass receiver closure, task-relevant information can affect the final estimate only by passing through a budgeted collection of vulnerable primitives unless an explicit protected bypass is modeled. Modeling each vulnerable primitive as a memoryless noisy channel yields a baseline supply--demand converse: the task-relevant information needed to attain a target distortion cannot exceed the smaller of the total information supplied by the communication channel and the total information supplied by the vulnerable compute budget. Our main converse shows that committed intermediate interfaces create additional first-order serial cuts and receiver-internal computation-graph cuts, captured in general by a receiver-internal compute min-cut converse. In particular, the twofold loss in the symmetric two-stage hard-separation special case is not inherent to unreliable receiver computation but induced by hard-separation under the committed/no-bypass closure. This extra first-order tax is therefore closure-dependent rather than universal. On the converse side, if downstream modules retain soft visibility to the raw channel output, the converse reduces to the single-bottleneck supply, up to any explicitly reserved soft-path budget. Under a separate stronger protected-support closure with reliable decoder and control support, we establish achievability results for task-direct and serial hard-separation constructions. For the fully noisy-logic regime, we obtain only a conservative depth-dependent converse, and matched achievability remains open.
\end{abstract}


\section{Introduction}
\label{sec:intro}

\IEEEPARstart{C}{lassical} information theory~\cite{Shannon48} typically treats receiver-side decoding and downstream processing as reliable.
This paper studies when remote inference remains possible with noisy communication and an unreliable receiver under a finite redundancy budget.
That classical abstraction is increasingly strained in two distinct but converging regimes.
First, energy-driven hardware trends such as near-threshold voltage operation and approximate
computing can introduce timing violations, transient faults, and stochastic bit flips within the
receiver pipeline~\cite{Dreslinski2010NTC,HanOrshansky2013Approx,Ernst2003Razor,Salami2018Underscaling,Agiakatsikas2023SoftErrors}.
Second, harsh-environment deployments such as \emph{in-orbit} inference can induce radiation-driven soft errors that corrupt accelerator state in addition to channel noise.
Recent measurement and emulation studies report SEU-style single-bit flips and spatially correlated multi-bit upsets (MCUs) in accelerator memories and radiation-prone inference settings~\cite{XuREMUDAC25,WangRedNetArxiv24}.
Complementary systems work develops software protection mechanisms that convert many silent corruptions into detected failures or erasures on practical accelerators~\cite{ZhengSAVEATC25,WangRadshieldASPLOS26}.
In such regimes, the reliability of remote inference is no longer governed by the channel alone, but by a three-way tradeoff among
(i) channel resources,
(ii) compute resources, and
(iii) the end task (e.g., estimation or inference) that consumes the received information.

\paragraph{The missing benchmark}
Information-theoretic treatments of unreliable communication and unreliable computation have largely evolved along separate tracks.
The noisy computation literature asks how to compute reliably using faulty components~\cite{vonNeumann56,GacsGal94,EvansSchulman99,Simon11}, while the communication literature asks how to convey information reliably over noisy channels.
Neither directly answers the \emph{task-centric} question here:
\emph{what \textbf{joint} communication and compute budget (equivalently, redundancy budget) is required to achieve a target task level (e.g., a distortion threshold or a distortion tail) despite receiver-side unreliability?}
We address this through a \emph{supply--demand} lens:
the demand is task-relevant information (captured by indirect/remote rate--distortion), and the supply is jointly bottlenecked by channel capacity and a compute information budget.

\paragraph{Problem setup and questions}
We consider remote inference where an encoder observes $Y$ correlated with a latent source $X$ and the receiver produces $\hat X$ under a distortion constraint.
Communication uses a discrete memoryless channel of capacity $\Cch$ for $n$ channel uses per source sample.
Receiver-side processing is performed on an unreliable digital substrate with a budget of $m$ noisy primitives per sample.
We model each primitive as a discrete memoryless channel with per-use capacity $\Cgate$ (a canonical baseline is a binary symmetric channel (BSC) with crossover probability $\varepsilon$, so $\Cgate = 1-\hbin{\varepsilon}$),
so that the receiver has an aggregate compute information supply of $m\Cgate$ bits per sample (Definition~\ref{def:gpu_bitflip_model}).
Within this model we ask:
(i) what task distortion levels are fundamentally achievable under joint communication and computation constraints?
(ii) when decoding itself is noisy, does enforcing a hard digital interface between ``decoding'' and ``task computation'' impose additional unavoidable penalties compared to task-direct processing?

\paragraph{Scope clarification}
The additional first-order penalty studied here is not a universal consequence of receiver noise.
It appears when task-relevant information must cross committed vulnerable interfaces that remove downstream access to the raw channel output.
When downstream modules retain access to $R^{nT}$ within the committed/no-bypass closure, the converse reduces to the single-bottleneck supply.
An explicit protected side path relaxes the active cut, but if a committed vulnerable interface remains, residual serial cuts may still persist; see Definition~\ref{def:hard_separation}, Propositions~\ref{prop:no_tax_soft_interface} and~\ref{prop:soft_interface_tradeoff}, and Corollaries~\ref{cor:supply_demand_bypass} and~\ref{cor:interface_tax_bypass}.

\subsection{Related Work and Positioning}
\paragraph{Noisy computation: from circuit complexity to information supply}
Noisy computation has a long history, from von Neumann's seminal work on reliable computation from unreliable components~\cite{vonNeumann56} to information-decay and redundancy lower bounds for noisy circuits~\cite{GacsGal94,EvansSchulman94,EvansSchulman99}.
Complementing circuit-level viewpoints, Simon~\cite{Simon11} introduced information-theoretic notions of \emph{computation capacity} for noisy function evaluation.
Those works primarily study reliable computation of a specified function.
Here the objective is instead \emph{task-centric}: approximate reconstruction under a distortion criterion, with receiver computation coupled to communication through an explicit per-instance budget~$m$.

\paragraph{Faulty decoding: algorithm-specific analyses vs.\ algorithm-agnostic converses}
Computation noise inside iterative decoders has been studied extensively, for example for low-density parity-check (LDPC) decoding under faulty updates~\cite{VarshneyTIT11} and belief propagation on noisy hardware~\cite{HuangLiDolecek15}.
Those are \emph{algorithm-specific} bit-recovery analyses.
Our results are complementary: they give \emph{algorithm-agnostic}, \emph{task-centric} converses and corresponding achievability statements under an explicit compute budget, with the converse side upper bounding the task-relevant mutual information deliverable by \emph{any} receiver algorithm.
In particular, Theorem~\ref{thm:interface_tax} shows that when decoding and task execution share the same unreliable substrate, a committed hard interface creates an additional first-order bottleneck.
Unlike classical source--channel separation, where an intermediate digital description does not create a new first-order bottleneck below channel capacity~\cite{Shannon48,CoverThomas}, a committed interface here must itself be materialized and re-consumed through noisy primitives, so task-relevant information crosses two serial compute cuts.

\paragraph{Indirect rate--distortion and task-oriented communication}
The indirect (remote) rate--distortion problem~\cite{BergerBook71,Witsenhausen80} characterizes the mutual-information demand needed to reconstruct $X$ from an observation $Y$ under a distortion constraint.
More broadly, task-oriented or semantic communication formulations optimize task performance rather than signal reconstruction~\cite{Deniz23,Jialong23}.
These models typically assume reliable receiver-side processing.
Our contribution is to make receiver computation noise and a compute budget explicit first-order information bottlenecks and to show that they can change both the task limit and the validity of hard-separation.

\paragraph{Additional-noise and noisy-receiver models}
Receiver-side impairments have also been studied directly in classical information theory, including additional-noise models~\cite{DobrushinTsybakov62}, source encoding from a randomly disturbed observation~\cite{Sakrison68}, and minimum-distortion transmission to a noisy receiver~\cite{WolfZiv70}.
Those formulations typically append a noisy observation or reproduction map to an otherwise reliable receiver.
Our bottleneck is different in kind: it is induced by receiver-internal committed materialization and re-consumption under an explicit compute budget.
At a formal level, Theorem~\ref{thm:dag_cutset} is cut-set-like, but it should not be read as a classical communication-network cut-set with exogenous link capacities: the cut edges are receiver-internal committed interfaces whose capacities are induced compute-supply limits, and the demand side is indirect rate--distortion rather than reliable message throughput; see Remark~\ref{rem:classical_noisy_receiver_relation}.

\paragraph{System motivation and finite-blocklength context}
Practical fault mechanisms motivate the abstraction, but the main theorems depend on the substrate only through information-supply parameters such as $m\Cgate$ or $\sum_i C_i$ and, for tail benchmarks, coarse undetected-error/detected-erasure (UE--ER) parameters; see Subsection~\ref{subsec:practical_noise_models}.
Classical finite-blocklength lossy joint source--channel coding (JSCC) shows that with reliable receiver-side processing, separation incurs only a dispersion-level penalty at short blocklengths/latencies~\cite{KostinaVerdu13_JSCC}.
By contrast, the committed-interface penalty studied here is a first-order compute-supply loss caused by serial vulnerable cuts, and it persists in Shannon asymptotics.

\subsection{Contributions}
\label{subsec:contributions}
Task demand is quantified by the indirect (remote) rate--distortion function $R_{X|Y}(D)$. Unless explicitly stated otherwise, the baseline and serial-interface converse results below are stated under the committed/no-bypass closure. All achievability statements in this subsection are instead under the stronger protected-support closure of Assumption~\ref{assump:achievability_benchmark}; they should not be read as a matched characterization under a single common closure.

\paragraph{Baseline supply--demand converse}
For task-direct processing, Theorem~\ref{thm:supply_demand_converse} proves
\begin{equation}
  R_{X|Y}(D)\le \min\{n\Cch,\; m\Cgate\},
  \label{eq:glance_task_direct}
\end{equation}
which serves as the paper's baseline supply--demand converse.
It isolates the two basic first-order supplies: the physical-channel supply $n\Cch$ and the vulnerable-compute supply $m\Cgate$.
Corollary~\ref{cor:supply_demand_bypass} shows how an explicit protected side path of rate $b$ bits per task instance relaxes the task-direct compute cut to $b+m\Cgate$.
Under Assumption~\ref{assump:achievability_benchmark}, Theorem~\ref{thm:achievability_general} gives the task-direct achievability result, while Theorem~\ref{thm:interface_tax_bypass_achievability} gives the protected-bypass achievability theorem for the two-stage architecture.

\paragraph{Architecture theorem chain}
The main architectural point is that, once remote inference is implemented through committed vulnerable interfaces, the first-order converse depends on receiver architecture.
Theorem~\ref{thm:dag_cutset} gives the master receiver-internal compute min-cut converse
\begin{equation}
  R_{X|Y}(D)
  \le \min\Big\{n\Cch,\; \min_{\Omega:\,s_{\mathcal G}\in\Omega,\,t_{\mathcal G}\notin\Omega}
      \sum_{e\in\partial^+(\Omega)} m_e\Cgate\Big\},
  \label{eq:glance_dag_cut}
\end{equation}
Equation~\eqref{eq:glance_dag_cut} should therefore be interpreted under the committed/no-bypass closure: its cut values are induced compute-supply limits, not exogenous communication-link capacities.
Theorems~\ref{thm:interface_tax} and~\ref{thm:k_stage_tax} are its two-stage and $K$-stage serial specializations.
In particular, the familiar twofold hard-separation loss is not inherent to unreliable receiver computation; it arises only in the symmetric two-stage hard-separation special case of this more general receiver-internal compute-cut phenomenon.
Proposition~\ref{prop:no_tax_soft_interface} shows, on the converse side, that if downstream modules retain soft visibility to the raw channel output within that closure, the converse reduces to the single-bottleneck supply, up to any explicitly reserved soft-path budget; see also Proposition~\ref{prop:soft_interface_tradeoff}.
Theorem~\ref{thm:interface_tax_achievability} and Corollary~\ref{cor:k_stage_tax_achievability} give the corresponding serial achievability results.

\paragraph{Variants and separate noisy-logic closure}
Protected bypasses, reliable islands, soft interfaces, and heterogeneous primitive capacities enter by modifying the active cut or its effective capacity, yielding a unified picture of when the hard-separation penalty is relaxed, eliminated, or replaced by another bottleneck; the protected-bypass and reliable-island cases also receive achievability results.
Section~\ref{sec:noisy_logic_closure} then studies a stricter closure in which the entire receiver logic, including control and bookkeeping, is noisy.
In that regime, Theorem~\ref{thm:supply_demand_noisy_gate} gives a conservative depth-dependent converse through the effective per-gate supply $C_{\mathrm{logic}}(d_{\mathrm{logic}},\delta,K_{\mathrm{fan}})$ defined in~\eqref{eq:C_logic_def}; a matched achievability characterization remains open.

\subsection{Paper Outline}
Sections~\ref{sec:model}--\ref{sec:converse} formalize the committed/no-bypass model, the task-demand notion, and the reference supply plane.
Section~\ref{sec:interface_tax_section} develops the receiver-internal cut converses, their serial specializations, and the main architectural variants.
Section~\ref{sec:noisy_logic_closure} studies the separate noisy-logic closure.
Section~\ref{sec:examples} gives compact discrete and scalar-Gaussian examples.
The appendices collect finite-blocklength and tail benchmarks, UE--ER refinements, throughput mappings, and supplementary variants and extensions.
\section{Model and Architectural Taxonomy}
\label{sec:model}

\indent This section fixes the baseline committed/no-bypass receiver model, the task-demand notion, and the architectural taxonomy used by the converse results. The emphasis is on which information paths are permitted to reach the task output, while finite-$T$ tail criteria and benchmark refinements are deferred to Appendix~\ref{sec:finite_blocklength}.

\subsection{Notation and Conventions}
We index \emph{task instances} (samples) within a coding block by $t\in\{1,\dots,T\}$, where $T$ is the \emph{task-blocklength} (the number of task instances coded jointly).
We denote length-$T$ task blocks by superscripts, e.g., $X^T=(X_1,\dots,X_T)$ and $Y^T=(Y_1,\dots,Y_T)$.
For each task instance $t$, the transmitter uses $n$ physical channel uses, producing channel input/output vectors $S_t^n$ and $R_t^n$.
Over a task block of length $T$, the physical channel is therefore used $N_{\mathrm{ch}}\triangleq nT$ times (we often write sequences as $S^{nT}$ and $R^{nT}$ to keep the dependence on $(n,T)$ explicit).
Likewise, the receiver is allowed at most $m$ uses of an unreliable binary primitive per task instance (Definition~\ref{def:gpu_bitflip_model}), i.e., at most $L\triangleq mT$ primitive uses per task block; in hard-separation this budget is allocated across decoding and downstream task inference with $m_{\mathrm{dec}}+m_{\mathrm{task}}\le m$.

Here $t$ indexes task instances rather than physical seconds; individual physical channel uses within a block are indexed by $\ell\in\{1,\dots,nT\}$ (equivalently $\ell\in\{1,\dots,N_{\mathrm{ch}}\}$), unreliable primitive uses by $i$, and serial stages by $k$. Physical time re-enters only when mapping per-second budgets or hardware error rates into per-instance parameters, e.g., via the throughput $\lambda$ in Appendix~\ref{sec:throughput} or via a dwell time $\tau$ in a soft-error-rate mapping.

We reserve the superscript $^T$ for length-$T$ task blocks (sequences). If matrix/vector transposes appear (e.g., in linear-Gaussian extensions), we denote transpose by $(\cdot)^{\top}$ (or $(\cdot)^{\mathsf T}$) to avoid overloading $^T$.

We use $\varepsilon\in[0,\tfrac12)$ for the BSC crossover probability of the unreliable compute primitive, while $\epsilon$ (with subscripts) denotes reliability targets such as excess-distortion probabilities and block error probabilities in finite-blocklength benchmarks.
The physical channel capacity is denoted by $\Cch$ (bits/use) and the primitive ``compute capacity'' by $\Cgate$ (bits/primitive use); the canonical BSC$(\varepsilon)$ specialization is introduced in Subsection~\ref{subsec:compute_model}.
Unless otherwise noted, mutual information and rate quantities are measured in bits (base-2 logarithms), and ``bits/sample'' means per task instance.
Table~\ref{tab:symbols} summarizes the frequently used symbols and conventions for reference.

\begin{table*}[t]
\centering
\caption{Frequently used symbols and conventions.}
\label{tab:symbols}
{\renewcommand{\arraystretch}{1.1}%
\begin{tabular}{p{0.18\textwidth} p{0.78\textwidth}}
\hline
Symbol & Meaning \\
\hline
$t$;\ $T$ & Task-instance index; task-blocklength (number of task instances coded jointly). \\
$X_t,\,Y_t,\,\hat X_t$ & Latent task variable, observation, and reconstruction for instance $t$. \\
$d(x,\hat x)$ & Distortion function. \\
$X^T,\,Y^T,\,\hat X^T$ & Length-$T$ blocks. \\
$n$ & Physical channel uses per task instance. \\
$N_{\mathrm{ch}}$ & Total physical channel uses per task block, $N_{\mathrm{ch}}\triangleq nT$. \\
$S_t^n,\,R_t^n$ & Channel input/output vectors for instance $t$ (length $n$). \\
$\Cch$ & Physical channel capacity (bits/channel use). \\
$m$ & Unreliable primitive uses per task instance (bit materializations). \\
$w$ & Word length (bits) for word-level primitives. \\
$L$ & Total unreliable primitive uses per task block, $L\triangleq mT$. \\
$m_{\mathrm{dec}},\,m_{\mathrm{task}}$ & Compute budget split across decoding and task stages under hard-separation. \\
$\varepsilon$ & Crossover probability of the unreliable primitive (BSC$(\varepsilon)$). \\
$\delta$ & Output-flip probability of a noisy logic gate (BSC$(\delta)$) in the noisy-gate circuit model (Subsection~\ref{subsec:noisy_gate_connections}). \\
$K_{\mathrm{fan}}$ & Maximum fan-in of a noisy logic gate in the noisy-gate circuit model (Subsection~\ref{subsec:noisy_gate_connections}). \\
$d_{\mathrm{logic}}$ & Minimum noisy-gate depth from channel-dependent inputs to the task output in the noisy-logic converse (Theorem~\ref{thm:supply_demand_noisy_gate}). \\
$\beta$ & Evans--Schulman propagation factor $\beta\triangleq K_{\mathrm{fan}}(1-2\delta)^2$. \\
$C_{\mathrm{logic}}(d_{\mathrm{logic}},\delta,K_{\mathrm{fan}})$ & Effective per-gate information supply under noisy receiver logic, defined in \eqref{eq:C_logic_def}. \\

$\Cgate$ & Compute primitive capacity (bits/primitive use); in the BSC$(\varepsilon)$ baseline, $\Cgate=1-\hbin{\varepsilon}$. \\
$\Cgateeff$ & Per-bit effective compute capacity for word-level additive upsets, $\Cgateeff=1-\frac{H(E)}{w}$. \\
$\Rsup$ & Per-instance information supply $\min\{n\Cch,\;m\Cgate\}$. \\
$\Rsuphard$ & Per-instance hard-separation supply in the symmetric two-stage special case, $\Rsuphard\triangleq \min\{n\Cch,\;\tfrac{m}{2}\Cgate\}$. \\
$b$ & Reliable bypass bits per task instance in the $b$-bit bypass extension (Remark~\ref{rem:bypass_model}). \\
$\Rsupb$ & Per-instance supply with bypass, $\Rsupb\triangleq\min\{n\Cch,\;b+m\Cgate\}$ (Corollary~\ref{cor:supply_demand_bypass}). \\
$D$ & Target average distortion (mean task loss). \\
$\epsilon_D^{\mathrm{blk}}(D)$ & Block excess-distortion probability in \eqref{eq:excess_dist_prob_def}. \\
$\epsilon_D^{(1)}(D)$ & Per-instance excess-distortion probability in \eqref{eq:excess_dist_prob_one}. \\
$\epsilon$ & Overall tail-reliability target (typically an upper bound on $\epsilon_D^{\mathrm{blk}}(D)$). \\
$\epsilon_{\mathrm{src}},\,\epsilon_{\mathrm{ch}},\,\epsilon_{\mathrm{comp}}$ & Error budgets for the finite-$T$ benchmark: source/task coding, physical channel coding, and compute-side redundancy. \\
$\epsilon_{\mathrm{dec}},\,\epsilon_{\mathrm{task}}$ & Error budgets for decoding-stage and task-stage computations under hard-separation. \\
$p_{\mathrm{ue}},\,p_{\mathrm{er}}$ & Undetected-error and detected-erasure probabilities in the correct-output/detected-erasure/undetected-error ($\mathsf{OK}/\mathsf{ER}/\mathsf{UE}$) outcome model (Appendix~\ref{sec:erasures}). \\
$\mathcal{B},\,\mathcal{G}$ & Per-second physical-channel-use and compute-primitive budgets (Appendix~\ref{sec:throughput}). \\
$\lambda$ & Throughput (task instances per second). \\
$\Lif$ & Length of an explicit digital interface message (bits per task instance) when such an interface is modeled. \\
$\ell$ & Physical channel-use index within a task block ($\ell\in\{1,\dots,N_{\mathrm{ch}}\}$). \\
$k$ & Stage index in a $K$-stage serial receiver pipeline ($k\in\{1,\dots,K\}$). \\
$K$ & Number of serial stages in a receiver pipeline (Subsection~\ref{subsec:k_stage_tax}). \\
$G=(\mathcal{V},\mathcal{E})$ & Receiver computation graph (DAG) used for receiver-internal compute min-cut converses (Subsection~\ref{subsec:dag_cutset}). \\
$m_e$ & Primitive budget assigned to interface edge $e\in\mathcal{E}$ in the computation graph (Subsection~\ref{subsec:dag_cutset}). \\
$\partial^+(\Omega)$ & Outgoing edge boundary of a vertex set $\Omega$ in the computation graph, defined in \eqref{eq:delta_plus_def}. \\

$i$ & Unreliable primitive-use (materialization) index; within a task block $i\in\{1,\dots,L\}$. \\
\hline
\end{tabular}}
\end{table*}

\subsection{Source, Observation, and Task Distortion}
At each task instance indexed by $t$, a latent variable $X_t$ is observed through $Y_t$ with a known joint law $P_{X,Y}$.
The receiver produces an estimate $\hat X_t$ and incurs a single-letter distortion $d(X_t,\hat X_t)$.
We focus on average distortion
\begin{equation}
  D \triangleq \mathbb{E}\big[d(X_t,\hat X_t)\big].
\end{equation}

The mean-distortion criterion above is the primary first-order task metric for the core converse chain. Tail/excess-distortion criteria and finite-$T$ refinements are deferred to Appendix~\ref{sec:finite_blocklength}.

\subsection{Channel Model and Per-Instance Communication Budget}
Per task instance, the communication budget is $n$ channel uses.
For the formal information-theoretic statements below, the encoder may be stochastic and may map an entire observation block $Y^T$ to $S^{nT}$; the single-instance picture $Y_t\mapsto S_t^n$ is recovered when $T=1$.
The channel is memoryless with transition law $P_{R|S}$ (e.g., additive white Gaussian noise (AWGN)) and a standard input constraint (e.g., average power).
Let $\mathcal{P}$ denote the set of admissible input distributions under this constraint.
The per-use capacity is
\begin{equation}
  \Cch \triangleq \sup_{P_S\in\mathcal{P}} I(S;R).
  \label{eq:Cch_def}
\end{equation}
Thus, the physical channel can supply at most $n\Cch$ bits of mutual information per task instance, or equivalently $nT\Cch$ bits over a task block of length $T$.

\subsection{Blocklength-$T$ Formulation and Standing Assumptions}
\label{subsec:blocklength}
For information-theoretic statements (converse and achievability), it is useful to consider a block of $T$ task instances.
Let $(X^T,Y^T)=\{(X_t,Y_t)\}_{t=1}^T$ denote the source--observation pairs.
Unless otherwise stated, we adopt the following standing assumptions \textbf{(S1)--(S3)}.

\begin{itemize}
  \item \textbf{(S1) Memorylessness across task instances:} $(X_t,Y_t)$ are i.i.d.\ according to $P_{X,Y}$.
  \item \textbf{(S2) Memoryless channel:} conditioned on inputs $S^{nT}$, the channel outputs satisfy
  $P_{R^{nT}|S^{nT}}=\prod_{\ell=1}^{nT} P_{R|S}$.
  \item \textbf{(S3) Per-instance resource budgets:} the encoder uses $n$ channel uses per task instance, so the block uses $N_{\mathrm{ch}}\triangleq nT$ channel uses in total.
  The receiver is allowed at most $m$ unreliable primitives per task instance, i.e., at most $L\triangleq mT$ unreliable primitives in the block.
\end{itemize}

The receiver may be stochastic and interacts with the noisy-computation substrate (Definition~\ref{def:gpu_bitflip_model}) for at most $L=mT$ unreliable primitives before outputting an estimate $\hat X^T$.
The average distortion constraint is
\begin{equation}
  \frac{1}{T}\sum_{t=1}^T \mathbb{E}\big[d(X_t,\hat X_t)\big] \le D.
\end{equation}

The per-instance budgets $(n,m)$ are the natural objects once communication bandwidth and compute energy/latency scale linearly with throughput; Appendix~\ref{sec:throughput} makes this explicit by translating per-second budgets into per-instance budgets via the sampling rate $\lambda$.

\subsection{Unreliable Computation Substrate and Per-Instance Compute Budget}
\label{subsec:compute_model}
We model receiver-side processing as an interactive algorithm that, in order to persist intermediate digital state, must repeatedly \emph{commit} values to vulnerable storage (or other unreliable primitives) and later retrieve/use them.
Each such vulnerable materialization is abstracted as one use of a discrete memoryless channel (DMC) $W_{Z|U}$ from an intended input $U$ to an available output $Z$.
The corresponding per-use \emph{compute information supply} is the Shannon capacity
\begin{equation}
  \Cgate \triangleq \max_{P_U} I(U;Z) \quad \text{bits/primitive}.
\end{equation}

\paragraph{Canonical baseline: BSC$(\varepsilon)$ (independent bit flips)}
A canonical and analytically convenient specialization is BSC$(\varepsilon)$, which models independent bit flips during a store--retrieve (materialization) operation.\footnote{We adopt an i.i.d.\ BSC$(\varepsilon)$ bit-flip model as an analytically convenient baseline.
Architecture-level resilience studies commonly abstract transient faults as random bit corruptions, often implemented as randomly selected single-bit flips in architectural state; see, e.g.,~\cite{FangPRG16,FangPRG14,HariTSKE17,TsaiNVBitFI21}.}

In this baseline, each vulnerable bit materialization is modeled as
\begin{equation}
  Z = U \oplus N,\qquad N\sim\mathrm{Bern}(\varepsilon),\quad 0\le \varepsilon < \tfrac12,
\end{equation}
where $U$ is the intended stored bit and $Z$ is the bit value subsequently available to the algorithm after potential flips.
The per-bit information capacity of this bit-flip mechanism is
\begin{equation}
  \Cgate \triangleq 1-\hbin{\varepsilon}\quad \text{bits/primitive},
\end{equation}
where $h_2(p)\triangleq -p\log_2 p-(1-p)\log_2(1-p)$ denotes the binary entropy function (base-2).

All converses in this paper depend on the compute substrate only through an upper bound on the per-use mutual information that can cross each vulnerable primitive use.
Therefore, the results extend directly to any DMC $W_{Z|U}$ by plugging in an appropriate effective $\Cgate$ (or, under heterogeneity, a sum of class-dependent supplies as in Lemma~\ref{lem:compute_info_bound_hetero}); Subsection~\ref{subsec:practical_noise_models} records the corresponding plug-in and calibration view.

We measure receiver compute/energy/latency constraints via a \emph{bit-materialization budget}: per task instance, the receiver may cause at most $m$ such vulnerable bit materializations.
If the graphics processing unit (GPU) operates on $w$-bit words, then one word-level store corresponds to $w$ parallel bit materializations; throughout, $m$ counts bits rather than words.

To make the compute bottleneck information-theoretically explicit, we adopt the following interactive model.

\begin{definition}[Committed/no-bypass noisy-materialization receiver model]
\label{def:gpu_bitflip_model}
Given the channel output $R_t^n$, a (possibly randomized) receiver algorithm interacts with an unreliable digital substrate (e.g., a GPU/accelerator) whose intermediate \emph{stored bits}---including both task-dependent intermediates (e.g., activations/messages) and task-independent parameters (e.g., quantized deep neural network (DNN) weights) when stored in vulnerable memory---are subject to independent flips.
At materialization $i\in\{1,\dots,m\}$, the algorithm selects an intermediate bit $U_i$ as a (measurable) function of $(R_t^n,Z^{i-1})$ and stores it in a vulnerable state element.
The value that is later available to subsequent computation is
$Z_i=U_i\oplus N_i$ with $N_i\stackrel{\text{i.i.d.}}{\sim}\mathrm{Bern}(\varepsilon)$.
After $m$ materializations, the final task output satisfies $\hat X_t = g(Z^m)$ for some measurable function $g$.
Thus any task-relevant information from $R_t^n$ that influences $\hat X_t$ must be \emph{committed} to vulnerable state at least once and traverse a noisy primitive.
Adaptive access to $R_t^n$ when choosing the materializations is allowed, but no additional reliable path from $R_t^n$ to $\hat X_t$ is present unless it is modeled explicitly through a protected bypass or related extension.
\end{definition}

\begin{remark}[Why commitment matters]
\label{rem:no_bypass}
This theorem-bearing information-structure assumption still permits fully adaptive read--compute--store strategies; it is precisely what turns Lemma~\ref{lem:compute_info_bound} into a genuine compute cut.
\end{remark}

\begin{remark}[$b$-bit reliable bypass / accounting]
\label{rem:bypass_model}
Some receiver implementations retain a protected path from the raw channel output to the final task output~\cite{ZhengSAVEATC25,WangRadshieldASPLOS26}. We model this by an auxiliary variable $V_t$ satisfying $H(V_t)\le b$ bits and $\hat X_t=g(Z^m,V_t)$, so the bypass contributes a reliable side path of rate $b$ bits/sample.
Protected output-relevant artifacts are charged to $b$, whereas vulnerable materializations count toward $m$.
Resident or reliably stored weights contribute little or no per-instance budget, whereas weights re-materialized through vulnerable memory on each inference must be charged to $m$.
Detailed accelerator-specific mappings, including weight loading and the learned-inference proxy in \eqref{eq:m_mapping_nn}, are deferred to Appendix~\ref{app:learned_inference_supp}.
\end{remark}

Figure~\ref{fig:bypass_accounting} summarizes this protected-bypass versus vulnerable-materialization accounting.

\begin{figure}[t]
\centering
\scalebox{0.81}{%
\begin{tikzpicture}[node distance=2.0cm,>=Latex]
  \node[draw,rounded corners,align=center,inner sep=3pt] (R) {$R_t^n$\\(channel output)};
  \node[draw,rounded corners,align=center,inner sep=4pt,right=2.2cm of R] (A) {receiver\\computation};
  \node[draw,rounded corners,align=center,inner sep=4pt,above right=1.2cm and 1.6cm of A] (V) {protected\\bypass $V_t$\\($H(V_t)\le b$)};
  \node[draw,rounded corners,align=center,inner sep=4pt,below right=1.2cm and 1.6cm of A] (Z) {vulnerable\\materializations $Z^m$\\($m$ uses)};
  \node[draw,rounded corners,align=center,inner sep=3pt,right=3.0cm of A] (X) {$\hat X_t$\\(task output)};

  \draw[->] (R) -- (A);
  \draw[->] (A) -- (V);
  \draw[->] (A) -- (Z);
  \draw[->] (V) -- (X);
  \draw[->] (Z) -- (X);
\end{tikzpicture}%
}
\caption{Protected-side-information extension of the committed/no-bypass receiver model. The protected side variable $V_t$ is available directly at the final task output in parallel with the vulnerable materialization path $Z^m$, so the compute cut relaxes from $m\Cgate$ to $b+m\Cgate$ when $H(V_t)\le b$; setting $b=0$ recovers the baseline no-bypass model.}
\label{fig:bypass_accounting}
\end{figure}
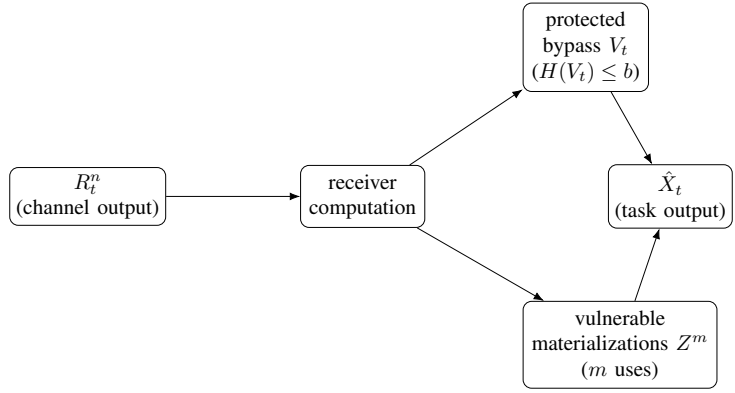

\subsection{Task Demand: Indirect (Remote) Rate--Distortion}\label{subsec:task_demand}
The \emph{task demand} for achieving distortion $D$ is quantified by the indirect rate--distortion function~\cite{BergerBook71,Witsenhausen80}
\begin{equation}
  R_{X|Y}(D)
  \triangleq
  \inf_{P_{\hat X|Y}:\ \mathbb{E}[d(X,\hat X)]\le D} I(Y;\hat X).
  \label{eq:indirect_rd_def}
\end{equation}

\begin{remark}[Induced-distortion view]
Define the induced distortion $\bar d(y,\hat x)\triangleq \mathbb{E}[d(X,\hat x)\mid Y=y]$.
Then $\mathbb{E}[d(X,\hat X)]\le D$ is equivalent to $\mathbb{E}[\bar d(Y,\hat X)]\le D$, and hence $R_{X|Y}(D)$ is the (direct) rate--distortion function of source $Y$ under distortion measure $\bar d$; see Lemma~\ref{lem:induced_distortion_equivalence}.
This viewpoint is reused for the finite-blocklength/dispersion refinements in Appendix~\ref{sec:finite_blocklength}.
\end{remark}

\begin{proposition}[Operational indirect rate--distortion theorem]
\label{prop:operational_indirect_rd}
For the memoryless block model introduced in Subsection~\ref{subsec:blocklength}, and under the standard regularity conditions for rate--distortion theory, $R_{X|Y}(D)$ equals the minimum asymptotically achievable fixed-length rate for reproducing $X^T$ from $Y^T$ under average distortion $D$.
Equivalently, for every $R>R_{X|Y}(D)$ there exist fixed-length codes $(f_T,\psi_T)$ with $|\mathcal W_T|\le 2^{TR}$ and
\[
  \hat X^T=\psi_T(f_T(Y^T)),
  \qquad
  \limsup_{T\to\infty}\frac{1}{T}\sum_{t=1}^T \mathbb{E}[d(X_t,\hat X_t)] \le D,
\]
while any such sequence must satisfy $R\ge R_{X|Y}(D)$.
\end{proposition}

\begin{proof}
By Lemma~\ref{lem:induced_distortion_equivalence}, $R_{X|Y}(D)$ is exactly the ordinary rate--distortion function of the source $Y$ under the induced distortion measure $\bar d$.
The claim therefore follows from Shannon's classical fixed-length lossy source coding theorem applied to the memoryless source $Y^T$ under distortion $\bar d$; see, e.g.,~\cite{BergerBook71,CoverThomas}.
\end{proof}

\subsection{Two Receiver Architectures}
\label{subsec:architectures}
We distinguish two baseline receiver organizations: Architecture~(A) is a \emph{task-direct} receiver and Architecture~(B) is a \emph{hard-separation} receiver. Soft-interface relaxations are formalized after the baseline hard-separation model.

\paragraph{Task-direct processing}
The receiver maps the channel output directly to $\hat X_t$ using the full compute budget of $m$ noisy bit materializations:
\begin{equation}
  Y_t \xrightarrow{\mathrm{Enc}} S_t^n \xrightarrow{\text{channel}} R_t^n
  \xrightarrow[\text{$m$ noisy bit materializations}]{\text{task-direct}} \hat X_t.
  \label{eq:task_direct_arch}
\end{equation}

\paragraph{Hard-separation with an intermediate digital interface}
The receiver first forms an intermediate digital representation $\hat B_t$ and then performs task computation using only $\hat B_t$ as input:
\begin{equation}
  Y_t \xrightarrow{\mathrm{Enc}} S_t^n \xrightarrow{\text{channel}} R_t^n
  \xrightarrow[\text{$m_{\mathrm{dec}}$ noisy bit materializations}]{\text{noisy decode}} \hat B_t
  \xrightarrow[\text{$m_{\mathrm{task}}$ noisy bit materializations}]{\text{noisy task}} \hat X_t,
  \qquad m_{\mathrm{dec}}+m_{\mathrm{task}}\le m.
  \label{eq:separation_arch}
\end{equation}
This downstream-visibility constraint is formalized in Definition~\ref{def:hard_separation}.

\begin{definition}[Hard-separation with a committed digital interface]
\label{def:hard_separation}
A $T$-block receiver architecture satisfies \emph{hard-separation} if there exists an intermediate \emph{digital interface} random variable $\hat B^T$ (taking values in a finite alphabet, e.g., $\{0,1\}^{\Lif T}$ for some interface length $\Lif$) such that the induced joint law obeys the Markov chain
$Y^T \to S^{nT} \to R^{nT} \to \hat B^T \to \hat X^T$,
and such that (i) $\hat B^T$ is computed from $R^{nT}$ on the noisy-computation substrate using at most $m_{\mathrm{dec}}T$ vulnerable primitive uses, (ii) $\hat X^T$ is computed from $\hat B^T$ (and not from $R^{nT}$) on the same substrate using at most $m_{\mathrm{task}}T$ primitive uses, with $m_{\mathrm{dec}}+m_{\mathrm{task}}\le m$.
If the downstream stage retains access to $R^{nT}$ or protected side information derived from it in addition to $\hat B^T$, the architecture becomes a \emph{soft-interface} relaxation and leaves pure hard-separation.
In that case the extra serial tax studied below is no longer mandatory; see Proposition~\ref{prop:no_tax_soft_interface}, Proposition~\ref{prop:soft_interface_tradeoff}, and Corollary~\ref{cor:supply_demand_bypass}.
The hard/soft qualifier refers to downstream information visibility, not to whether $\hat B^T$ itself contains multi-bit ``soft'' statistics.
If the task stage observes only $\hat B^T$, the Markov chain still holds and the serial hard-separation converse continues to apply.
\end{definition}

Figure~\ref{fig:arch} summarizes both organizations, while Table~\ref{tab:arch_taxonomy} places them in the broader architecture/closure taxonomy used later.

\begin{figure*}[t]
\centering
\begin{tikzpicture}[
  node distance=9mm,
  box/.style={draw, rounded corners, align=center, minimum height=7mm, minimum width=18mm},
  arrow/.style={-Latex, thick}
]
\node[box] (y) {$Y_t$};
\node[box, right=of y] (enc) {Enc};
\node[box, right=of enc] (ch) {Channel};
\node[box, right=of ch] (td) {Noisy\\task-direct\\{\small $m$}};
\node[box, right=of td] (xhat) {$\hat X_t$};

\draw[arrow] (y) -- (enc);
\draw[arrow] (enc) -- node[above, font=\small] {$S_t^n$} (ch);
\draw[arrow] (ch) -- node[above, font=\small] {$R_t^n$} (td);
\draw[arrow] (td) -- (xhat);

\node[below=5mm of ch, font=\small] {(a) Task-direct};

\node[box, below=16mm of y] (y2) {$Y_t$};
\node[box, right=of y2] (enc2) {Enc};
\node[box, right=of enc2] (ch2) {Channel};
\node[box, right=of ch2] (dec) {Noisy\\decode\\{\small $m_{\mathrm{dec}}$}};
\node[box, right=of dec] (task) {Noisy\\task\\{\small $m_{\mathrm{task}}$}};
\node[box, right=of task] (xhat2) {$\hat X_t$};

\draw[arrow] (y2) -- (enc2);
\draw[arrow] (enc2) -- node[above, font=\small] {$S_t^n$} (ch2);
\draw[arrow] (ch2) -- node[above, font=\small] {$R_t^n$} (dec);
\draw[arrow] (dec) -- node[above, font=\small] {$\hat B_t$} (task);
\draw[arrow] (task) -- (xhat2);

\node[below=6mm of ch2, align=center, font=\small] {{(b) Hard-separation via committed interface}};
\end{tikzpicture}
\caption{Two baseline receiver organizations under the same vulnerable-compute budget $m$. In task-direct processing, the task output is formed directly from the channel output using at most $m$ vulnerable primitive uses. In hard-separation, the receiver first materializes a committed interface $\hat B_t$ using $m_{\mathrm{dec}}$ vulnerable uses and then performs the downstream task using only $\hat B_t$ with $m_{\mathrm{task}}$ vulnerable uses, where $m_{\mathrm{dec}}+m_{\mathrm{task}}\le m$.}
\label{fig:arch}
\end{figure*}
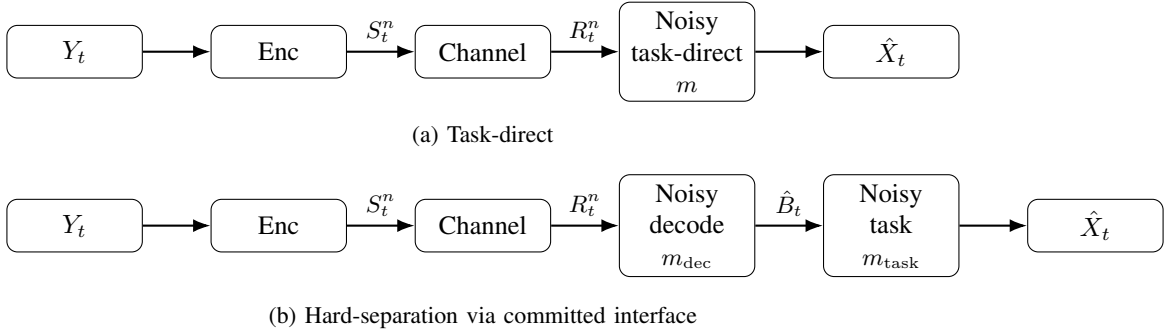

\begin{table*}[t]
\centering
\caption{Architecture taxonomy, active converses, and closure-dependent tax status.}
\label{tab:arch_taxonomy}
{\renewcommand{\arraystretch}{1.12}%
\scriptsize
\begin{tabular}{@{}p{0.15\textwidth}p{0.10\textwidth}p{0.11\textwidth}p{0.22\textwidth}p{0.14\textwidth}p{0.16\textwidth}@{}}
\toprule
Architecture / closure & Raw $R^{nT}$ downstream? & Committed / no-bypass? & Main converse & Interface-tax status & Achievability status \\
\midrule
Task-direct & Yes (direct use) & Yes & Theorem~\ref{thm:supply_demand_converse} & No extra serial cut & Achievability theorem under Assumption~\ref{assump:achievability_benchmark} \\
Hard-separation & No & Yes & Theorem~\ref{thm:interface_tax}; general form: Theorem~\ref{thm:dag_cutset} & Serial / receiver-internal compute cut present & Converse under the committed/no-bypass closure; achievability result under Assumption~\ref{assump:achievability_benchmark} \\
Soft-interface & Yes & Broken at the downstream interface & Task-direct cut; see Propositions~\ref{prop:no_tax_soft_interface}, \ref{prop:soft_interface_tradeoff} & Single-bottleneck form, up to reserved soft-path budget & No separate forward theorem; treated through the active cut structure \\
Protected bypass / reliable side path & Partially, via protected side information & Explicitly relaxed by a protected side path & Corollary~\ref{cor:supply_demand_bypass} for task-direct processing; Corollary~\ref{cor:interface_tax_bypass} when a committed interface remains & Mitigated / removable when $b$ is large & Converse + achievability theorem (Theorem~\ref{thm:interface_tax_bypass_achievability}) under Assumption~\ref{assump:achievability_benchmark} \\
Reliable island & Not necessarily & Commitment remains on vulnerable state & Corollary~\ref{cor:interface_tax_reliable_island} & Mitigated by enlarged task-side cut & Converse + achievability corollary (Corollary~\ref{cor:interface_tax_reliable_island_achievability}, homogeneous case) under Assumption~\ref{assump:achievability_benchmark} \\
Heterogeneous primitive capacities & Not necessarily & Commitment remains on vulnerable state & Lemma~\ref{lem:compute_info_bound_hetero}; Corollary~\ref{cor:interface_tax_unequal} & Rescaled stage capacities / mitigated cut & Converse / cut-capacity plug-in only \\
Noisy-logic closure & Architecture-dependent & Separate closure & Theorem~\ref{thm:supply_demand_noisy_gate} & Not a committed-interface tax & Conservative converse only \\
\bottomrule
\end{tabular}}
\end{table*}

\section{Baseline Reference Bound and Achievability Under a Protected-Support Closure}
\label{sec:converse}

\indent This section records the baseline reference bound for the committed/no-bypass closure and the corresponding achievability result under Assumption~\ref{assump:achievability_benchmark}.

\subsection{A Per-Instance Compute Information Bound}
Our converses rely on a single abstract ``compute cut'': each time the receiver commits information to a vulnerable primitive, at most $\Cgate$ bits of mutual information can cross that primitive use, even under adaptive operation. Thus any module using at most $m$ such primitives can convey at most $m\Cgate$ bits from its input to its final output. We first state this principle in a general conditional-capacity form and then record the baseline bit-flip specialization together with the reliable-bypass plug-in used later. Additional heterogeneous and word-level plug-ins are deferred to Appendix~\ref{app:section3_supp}.

\begin{lemma}[General compute information bound under a conditional-capacity constraint]
\label{lem:compute_info_bound}
Consider an interactive receiver module that observes an input $U$ and uses $m$ unreliable primitives sequentially.
At use $i$, it selects an intended primitive input $U_i$ as a (possibly randomized) function of $(U,Z^{i-1})$ and then obtains an output $Z_i$ according to some conditional law
$P_{Z_i|U_i,Z^{i-1}}(\cdot\mid \cdot, z^{i-1})$.
Assume that, for every history $z^{i-1}$, the induced conditional channel has capacity at most $\Cgate$, i.e.,
\begin{equation}
  \max_{P_{U_i}} I(U_i;Z_i \mid Z^{i-1}=z^{i-1}) \le \Cgate.
  \label{eq:cond_capacity_bound}
\end{equation}
Then the module output $\hat U=g(Z^m)$ satisfies
\begin{equation}
  I(U;\hat U)\le I(U;Z^m)\le m\,\Cgate.
\end{equation}
\end{lemma}

\begin{proof}
By data processing, $I(U;\hat U)\le I(U;Z^m)$ since $\hat U=g(Z^m)$.
For the second inequality, apply the chain rule:
\begin{equation}
  I(U;Z^m)=\sum_{i=1}^m I(U;Z_i\mid Z^{i-1}).
\end{equation}
Fix $i$ and condition on $Z^{i-1}=z^{i-1}$.
By construction, the intended primitive input $U_i$ is generated as a (possibly randomized) function of $(U,z^{i-1})$.
Moreover, the primitive output $Z_i$ is drawn according to the conditional law $P_{Z_i|U_i,Z^{i-1}}(\cdot\mid\cdot,z^{i-1})$, which depends on $U$ only through $U_i$.
Equivalently,
\begin{equation*}
  P_{Z_i|U,U_i,Z^{i-1}}(\cdot\mid u,u_i,z^{i-1})
  = P_{Z_i|U_i,Z^{i-1}}(\cdot\mid u_i,z^{i-1}),
\end{equation*}
so $U\to U_i\to Z_i$ forms a Markov chain conditional on $Z^{i-1}=z^{i-1}$.
Hence
\begin{equation}
  I(U;Z_i\mid Z^{i-1}=z^{i-1})
  \le I(U_i;Z_i\mid Z^{i-1}=z^{i-1})
  \le \Cgate
\end{equation}
by \eqref{eq:cond_capacity_bound}.
Averaging over $Z^{i-1}$ yields $I(U;Z_i\mid Z^{i-1})\le \Cgate$, and summing over $i=1,\dots,m$ gives $I(U;Z^m)\le m\Cgate$.
\end{proof}

\begin{corollary}[BSC$(\varepsilon)$ bit-flip specialization]
Under Definition~\ref{def:gpu_bitflip_model}, the conditional-capacity bound \eqref{eq:cond_capacity_bound} holds with $\Cgate=1-\hbin{\varepsilon}$. Hence
\begin{equation}
  I(R_t^n;\hat X_t)\ \le\ I(R_t^n;Z^m)\ \le\ m\Cgate.
  \label{eq:compute_info_bound}
\end{equation}
\end{corollary}

\begin{proof}
Immediate from Lemma~\ref{lem:compute_info_bound} with $U=R_t^n$.
\end{proof}

\begin{lemma}[Compute information bound with a $b$-bit reliable bypass]
\label{lem:compute_info_bound_bypass}
Consider the GPU substrate model in Definition~\ref{def:gpu_bitflip_model}, but allow the receiver to additionally generate an auxiliary variable $V_t$ through a reliable path as an arbitrary (possibly randomized) function of $R_t^n$, where $H(V_t)\le b$ bits.
If the receiver output has the form $\hat X_t=g(Z^m,V_t)$ for some measurable function $g$, then
\begin{equation}
  I(R_t^n;\hat X_t)\le b+m\Cgate.
  \label{eq:compute_info_bound_bypass}
\end{equation}
\end{lemma}

\begin{proof}
By data processing, $I(R_t^n;\hat X_t)\le I(R_t^n;Z^m,V_t)$.
Applying the chain rule,
\begin{equation}
\begin{aligned}
  I(R_t^n;Z^m,V_t)
  &= I(R_t^n;Z^m) + I(R_t^n;V_t\mid Z^m) 
  \le m\Cgate + H(V_t) \le m\Cgate + b,
\end{aligned}
\end{equation}
where $I(R_t^n;Z^m)\le m\Cgate$ follows from Lemma~\ref{lem:compute_info_bound} and $I(R_t^n;V_t\mid Z^m)\le H(V_t)\le b$.
\end{proof}

\subsection{Supply--Demand Converse for Task Distortion}
Our main converses compare a \emph{task demand} function, $R_{X|Y}(D)$, against a \emph{supply} limited by communication and computation.
To do so rigorously for block codes (which may couple many task instances), we will use the following standard ``rate--distortion lower bound'':
any reconstruction meeting an average distortion constraint must induce at least $R_{X|Y}(D)$ bits/sample of mutual information between $Y^T$ and $\hat X^T$.

\begin{lemma}[Indirect rate--distortion lower bound on block mutual information]
\label{lem:ird_block_lower_bound}
Under the standing assumptions in Subsection~\ref{subsec:blocklength}, let $\hat X^T$ be \emph{any} (possibly stochastic) reconstruction of $X^T$ based on $Y^T$.
If
\begin{equation}
  \frac{1}{T}\sum_{t=1}^T \mathbb{E}\big[d(X_t,\hat X_t)\big] \le D,
\end{equation}
then
\begin{equation}
  \frac{1}{T} I(Y^T;\hat X^T)\ \ge\ R_{X|Y}(D).
  \label{eq:ird_block_lower_bound}
\end{equation}
\end{lemma}

\begin{proof}
Define the per-instance distortions $D_t\triangleq \mathbb{E}[d(X_t,\hat X_t)]$.
Using the chain rule and noting that $\hat X_t$ is a component of $\hat X^T$,
\begin{equation}
\begin{aligned}
  I(Y^T;\hat X^T)
  &= \sum_{t=1}^T I(Y_t;\hat X^T \mid Y^{t-1}) 
  \ge \sum_{t=1}^T I(Y_t;\hat X_t \mid Y^{t-1}) \\
  &= \sum_{t=1}^T \big(I(Y_t;\hat X_t, Y^{t-1}) - I(Y_t;Y^{t-1})\big)
   = \sum_{t=1}^T I(Y_t;\hat X_t, Y^{t-1}) 
  \ge \sum_{t=1}^T I(Y_t;\hat X_t),
\end{aligned}
\label{eq:ird_chain_lower}
\end{equation}
where $I(Y_t;Y^{t-1})=0$ since $\{Y_t\}_{t=1}^T$ are i.i.d., and adding variables cannot decrease mutual information.

For each $t$, the joint law of $(Y_t,\hat X_t)$ induces a conditional distribution $P_{\hat X_t\mid Y_t}$ that attains distortion $D_t$.
By the definition of $R_{X|Y}(\cdot)$ in \eqref{eq:indirect_rd_def}, this implies
\begin{equation}
  I(Y_t;\hat X_t) \ge R_{X|Y}(D_t)\qquad \text{for every }t.
\end{equation}
Finally, $R_{X|Y}(D)$ is convex and nonincreasing in $D$ (as a rate--distortion function; see~\cite{BergerBook71,CoverThomas}), so Jensen's inequality gives
\begin{equation}
  \frac{1}{T}\sum_{t=1}^T R_{X|Y}(D_t)
  \ge R_{X|Y}\Big(\frac{1}{T}\sum_{t=1}^T D_t\Big)
  \ge R_{X|Y}(D).
\end{equation}
Combining with \eqref{eq:ird_chain_lower} and dividing by $T$ yields \eqref{eq:ird_block_lower_bound}.
\end{proof}

Define the per-instance information supply for task-direct processing as
\begin{equation}
  \Rsup \triangleq \min\{n\Cch,\ m\Cgate\}\quad \text{bits/sample}.
  \label{eq:Rsupply_def}
\end{equation}
We will repeatedly invoke Lemma~\ref{lem:compute_info_bound} by choosing the output variable to be an intermediate interface variable, thereby bounding the mutual information crossing that interface by the corresponding primitive budget times $\Cgate$.

\begin{theorem}[Supply--demand converse]
\label{thm:supply_demand_converse}
Assume \textbf{(S1)--(S3)} in Subsection~\ref{subsec:blocklength}, and consider the task-direct architecture in \eqref{eq:task_direct_arch}.
If there exists an encoding and noisy-compute decoding scheme such that the (block) average distortion satisfies
\[
  \frac{1}{T}\sum_{t=1}^T \mathbb{E}\big[d(X_t,\hat X_t)\big]\le D,
\]
then
\begin{equation}
  R_{X|Y}(D)\ \le\ \Rsup \;=\; \min\{n\Cch,\ m\Cgate\}.
  \label{eq:main_converse}
\end{equation}
\end{theorem}

\begin{proof}
 Let the overall $T$-block scheme induce the joint law of $(X^T,Y^T,\hat X^T)$.
 By Lemma~\ref{lem:ird_block_lower_bound}, achieving average distortion at most $D$ implies
 \begin{equation}
  R_{X|Y}(D)\ \le\ \frac{1}{T} I(Y^T;\hat X^T).
  \label{eq:ird_le_block_mutual}
\end{equation}

We now upper bound $\frac{1}{T}I(Y^T;\hat X^T)$ by two cuts.

 \emph{(Communication cut)} Since $Y^T\to S^{nT}\to R^{nT}\to \hat X^T$, data processing gives
 \[
   I(Y^T;\hat X^T)\le I(S^{nT};R^{nT}).
 \]
 By the standard memoryless-channel converse argument under the imposed input constraint,
 \[
   I(S^{nT};R^{nT})\le nT\,\Cch,
 \]
 and hence $\frac{1}{T}I(Y^T;\hat X^T)\le n\Cch$.

\emph{(Computation cut)} Let $U^{L}$ and $Z^{L}$ denote the intended and retrieved bits over the $L=mT$ vulnerable materializations.
Under the no-bypass (commitment) assumption in Definition~\ref{def:gpu_bitflip_model}, $\hat X^T$ is a (measurable) function of $Z^{L}$, so $R^{nT}\to Z^{L}\to \hat X^T$ is a Markov chain and therefore
\[
  I(Y^T;\hat X^T)\le I(R^{nT};\hat X^T)\le I(R^{nT};Z^{L}).
\]
Applying Lemma~\ref{lem:compute_info_bound} to the $L=mT$ primitive uses over the block gives
$I(R^{nT};Z^{L})\le mT\,\Cgate = L\,\Cgate$, and hence $\frac{1}{T}I(Y^T;\hat X^T)\le m\Cgate$.

Combining these two cuts with \eqref{eq:ird_le_block_mutual} yields \eqref{eq:main_converse}.
\end{proof}

\begin{discussion}{Interpretation}
Theorem~\ref{thm:supply_demand_converse} says that task performance is governed by the smaller of the physical channel supply $n\Cch$ and the receiver compute supply $m\Cgate$: when $n\Cch\le m\Cgate$, the benchmark reduces to the classical channel-limited remote rate--distortion constraint, whereas when $m\Cgate\le n\Cch$, unreliable computation becomes the active bottleneck and hard-separation can further reduce the usable compute supply to $\tfrac{m}{2}\Cgate$ (Corollary~\ref{cor:hard_sep_supply_demand}), yielding a strict separation loss under Corollary~\ref{cor:strict_sep_loss}.
\end{discussion}

As a plug-in variation of the same task-direct converse, an explicit reliable side path simply adds its protected entropy budget to the compute cut.

\begin{corollary}[Supply--demand converse with a $b$-bit reliable bypass]
\label{cor:supply_demand_bypass}
Assume the setting of Theorem~\ref{thm:supply_demand_converse}, but replace the no-bypass compute model in Definition~\ref{def:gpu_bitflip_model} by the $b$-bit bypass extension in Remark~\ref{rem:bypass_model}.
Then any scheme achieving average distortion at most $D$ must satisfy
\begin{equation}
  R_{X|Y}(D)
  \le \Rsupb
  \triangleq \min\{n\Cch,\ b+m\Cgate\}.
  \label{eq:supply_demand_bypass}
\end{equation}
In particular, the baseline converse bound in Theorem~\ref{thm:supply_demand_converse} is recovered by setting $b=0$.
\end{corollary}

\begin{proof}
The indirect rate--distortion lower bound and the communication cut are unchanged from Theorem~\ref{thm:supply_demand_converse}. Under the $b$-bit bypass extension,
\[
  I(R^{nT};\hat X^T)\le I(R^{nT};Z^{mT},V^T)\le mT\Cgate+bT,
\]
so combining with Lemma~\ref{lem:ird_block_lower_bound} yields \eqref{eq:supply_demand_bypass}.
\end{proof}

\paragraph{Supplementary plug-ins and task variants} Additional heterogeneous/word-level compute-supply replacements and a Fano-style classification benchmark are recorded in Appendix~\ref{app:section3_supp}.

\subsection{Achievability Under a Protected-Support Closure}
\label{subsec:achievability_general}
Theorem~\ref{thm:supply_demand_converse} also admits a Shannon-style achievability construction built from remote source coding, capacity-approaching channel coding, and internal coding across the memoryless compute primitive.
The closure used by these forward statements is formalized next, and the internal-coding ingredient is recorded abstractly here and explicitly in Subsection~\ref{subsec:benchmark_internal_coding}.

\paragraph{Scope note for achievability statements}
The converse results in Sections~\ref{sec:converse}--\ref{sec:interface_tax_section} use the committed/no-bypass closure, whereas the achievability statements below adopt the stronger protected-support closure of Assumption~\ref{assump:achievability_benchmark}; the forward and converse results therefore apply to different closures.

\paragraph{Standing regularity assumptions} Throughout the paper, source and reconstruction alphabets are understood as standard Borel spaces, the distortion measure is measurable and integrable under the relevant test channels and code-induced laws, and the indirect rate--distortion function is taken to be well defined under these standing conditions. We do not dwell further on these standard regularity assumptions.

\begin{assumption}[Protected-support closure for achievability]
\label{assump:achievability_benchmark}
The achievability statements are stated under the following stronger model-closure assumptions.
\begin{enumerate}
  \item[(A1)] \textbf{Budgeted memoryless primitive accounting.}
  Receiver-side vulnerability is captured by at most $L=mT$ uses of a DMC compute primitive $W_{Z|U}$ with capacity $\Cgate$ (Definition~\ref{def:gpu_bitflip_model}),
  independent across uses given their inputs.

  \item[(A2)] \textbf{Protected decoder/control support.}
  The control, addressing, stagewise recoding, codebook lookup, and decoder/post-processing support machinery invoked in the achievability proofs is assumed reliable,
  or equivalently is implemented within a protected ``reliable island'' whose resource cost is not charged to the vulnerable-materialization budget $m$.
  This stronger closure supports the Shannon-style achievability theorems stated here; if the same bookkeeping and recoding support must also be implemented on the vulnerable substrate and charged to $m$, then Section~\ref{sec:noisy_logic_closure} provides the stricter conservative model.

  \item[(A3)] \textbf{Shannon-theoretic coding availability.}
  The physical channel admits capacity-approaching codes of rate up to $n\Cch$ bits/sample with vanishing block error probability,
  and, more generally, for any allocated sub-budget $m'\ge 0$, the compute primitive admits internal redundancy codes of rate up to $m'\Cgate$ bits/sample with vanishing block error probability (Proposition~\ref{prop:compute_code_bsc}).
  Likewise, by Proposition~\ref{prop:operational_indirect_rd}, for any $D$ and any $R>R_{X|Y}(D)$, fixed-length indirect rate--distortion codes of rate $R$ exist in the Shannon asymptotic.
\end{enumerate}
\end{assumption}

Table~\ref{tab:assumption_scope} summarizes the assumption and closure scope of representative results in the paper, making explicit the converse/achievability split across closures.

\begin{table}[htbp]
\centering
\caption{Scope of representative results. Converse results use only the committed/no-bypass closure; achievability results use Assumption~\ref{assump:achievability_benchmark}; Theorem~\ref{thm:supply_demand_noisy_gate} belongs to a separate noisy-logic closure.}
\label{tab:assumption_scope}
\small
\begin{tabular}{p{0.33\linewidth}p{0.16\linewidth}p{0.08\linewidth}p{0.08\linewidth}p{0.21\linewidth}}
\toprule
Result (representative) & A1-style accounting & Needs (A2) & Needs (A3) & Achievability status \\
\midrule
Supply--demand / cut-set converses & Yes & No & No & Converse \\
Bypass / reliable-side converses & Yes+$b$ & No & No & Converse \\
Task-direct achievability & Yes & Yes & Yes & Achievability \\
Serial achievability & Yes & Yes & Yes & Achievability \\
Bypass / island achievability results & Yes + protected bits & Yes & Yes & Achievability \\
Noisy-logic depth-dependent converse & Circuit & No & No & Separate converse \\
\bottomrule
\end{tabular}
\end{table}

\begin{theorem}[Task-direct achievability under the protected-support closure]
\label{thm:achievability_general}
Assume Assumption~\ref{assump:achievability_benchmark}.
Then, for any distortion level $D$ such that
\begin{equation}
  R_{X|Y}(D) \;<\; \Rsup,
  \label{eq:achievability_condition}
\end{equation}
there exists a task-direct scheme (Architecture~(A)) whose average distortion satisfies
$\frac{1}{T}\sum_{t=1}^T\mathbb{E}[d(X_t,\hat X_t)]\le D$ in the Shannon asymptotic regime ($T\to\infty$).
\end{theorem}

\paragraph{Mean-squared error (MSE)/unbounded-distortion note} Under squared-error distortion and finite second moment, Lemma~\ref{lem:clip_closes_mse_achievability} closes the unbounded-distortion case; see Appendix~\ref{app:clipping_achievability_lemma}.

\begin{proof}
Fix any rate $R$ such that
$R_{X|Y}(D) < R < \Rsup$.
For each task-blocklength $T$, let $M_T\triangleq 2^{\lfloor TR\rfloor}$.

\emph{(Remote source coding)}
By Proposition~\ref{prop:operational_indirect_rd},
there exists a sequence of fixed-length indirect rate--distortion codes
\[(f_T,\psi_T):\ \mathcal{Y}^T\to \{1,\dots,M_T\}\to \hat{\mathcal X}^T\]
whose induced reconstruction $\hat X_{\mathrm{RD}}^T\triangleq \psi_T(f_T(Y^T))$ satisfies
\begin{equation}
  \frac{1}{T}\sum_{t=1}^T\mathbb{E}\big[d(X_t,\hat X_{\mathrm{RD},t})\big]
  \le D+\xi_T,
  \qquad \xi_T\to 0\ \text{as }T\to\infty.
  \label{eq:achievability_remote_code}
\end{equation}
Let $W\triangleq f_T(Y^T)\in\{1,\dots,M_T\}$ denote the remote description index.

\emph{(Physical channel coding)}
Since $\log_2 M_T\le TR$ and $R<n\Cch$, Shannon's channel coding theorem implies that there exists a sequence of length-$nT$ channel codes
with a (reliable) decoder (e.g., implemented on a protected control plane as in Assumption~\ref{assump:achievability_benchmark}-(A2)) producing an estimate $\tilde W$ from $R^{nT}$ such that
\begin{equation}
  \epsilon_{\mathrm{ch}}^{(T)}\triangleq \mathbb{P}(\tilde W\neq W)\ \longrightarrow\ 0\qquad\text{as }T\to\infty.
  \label{eq:achievability_channel_error}
\end{equation}

\emph{(Internal compute-side redundancy)}
Since $\log_2 M_T\le TR$ and $R<m\Cgate$, Proposition~\ref{prop:compute_code_bsc} guarantees the existence of a sequence of internal codes
(encoder $\alpha_T$ and decoder $\beta_T$) over the $L=mT$ primitive uses such that, if the receiver materializes
$U^{L}=\alpha_T(\tilde W)$ through the primitive to obtain $Z^{L}$ and forms $\hat W\triangleq \beta_T(Z^{L})$, then
\begin{equation}
  \epsilon_{\mathrm{comp}}^{(T)}\triangleq \mathbb{P}(\hat W\neq \tilde W)\ \longrightarrow\ 0\qquad\text{as }T\to\infty.
  \label{eq:achievability_compute_error}
\end{equation}
(See Appendix~\ref{app:internal_redundancy_instantiation} for a concrete correcting-network instantiation.)

\emph{(Overall task-direct receiver)}
The receiver algorithm is: decode $\tilde W$ from $R^{nT}$; materialize $U^L=\alpha_T(\tilde W)$ through the $L=mT$ primitive uses; decode $\hat W=\beta_T(Z^L)$;
and form the preliminary reconstruction
$\tilde X^T\triangleq \psi_T(\hat W)$.
We set $\hat X^T\triangleq \tilde X^T$ as the final receiver output.
This is a valid task-direct scheme of the form \eqref{eq:task_direct_arch}, since the final output depends on $R^{nT}$ only through the chosen materializations and the resulting $Z^L$.

Let $\mathsf{E}_T\triangleq\{\hat W\neq W\}$.
By construction, $\mathsf{E}_T\subseteq \{\tilde W\neq W\}\cup\{\hat W\neq \tilde W\}$, hence by a union bound,
\begin{equation}
  \mathbb{P}(\mathsf{E}_T)\ \le\ \epsilon_{\mathrm{ch}}^{(T)}+\epsilon_{\mathrm{comp}}^{(T)}.
  \label{eq:achievability_union_error}
\end{equation}
On the event $\mathsf{E}_T^c$, we have $\hat W=W$ and therefore $\tilde X^T=\psi_T(W)=\hat X_{\mathrm{RD}}^T$.

If the distortion measure is bounded, $0\le d(X,\hat X)\le d_{\max}$ almost surely, then combining \eqref{eq:achievability_remote_code} and \eqref{eq:achievability_union_error} yields
\begin{align}
  \frac{1}{T}\sum_{t=1}^T\mathbb{E}[d(X_t,\hat X_t)]
  &\le (D+\xi_T)\,\mathbb{P}(\mathsf{E}_T^c) + d_{\max}\,\mathbb{P}(\mathsf{E}_T)\\
  &\le D+\xi_T + d_{\max}\big(\epsilon_{\mathrm{ch}}^{(T)}+\epsilon_{\mathrm{comp}}^{(T)}\big).
\end{align}
Letting $T\to\infty$ gives the claimed achievability of $D$.

For squared-error distortion with $\mathbb{E}[X^2]<\infty$, apply Lemma~\ref{lem:clip_closes_mse_achievability} with error event $\mathsf{E}_T$ and \eqref{eq:achievability_union_error}; this closes mean-MSE achievability of $D$.
\end{proof}

\section{Architecture-Induced First-Order Compute Cuts Under Committed Interfaces}
\label{sec:interface_tax_section}

\indent We first develop the two-stage and serial hard-separation specializations before stating the general receiver-internal compute min-cut converse in Subsection~\ref{subsec:dag_cutset}; Proposition~\ref{prop:no_tax_soft_interface} then records the no-tax soft-interface relaxation.

\subsection{Two-Stage Hard-Separation Specialization: Interface Tax}
Recall the hard-separation constraint in Definition~\ref{def:hard_separation}.
In hard-separation, channel decoding is itself computed on the noisy substrate and its output is forced through an intermediate \emph{committed digital interface} before task inference, creating the two serial mutual-information cuts formalized below.
For a $T$-block scheme, this induces the Markov chain
\begin{equation}
  Y^T \to S^{nT} \to R^{nT} \to \hat B^T \to \hat X^T,
  \label{eq:block_markov_hard_sep}
\end{equation}
where $\hat B^T$ is the decoded (digital) representation and $\hat X^T$ is the task output.

We model the decoder stage as using at most $m_{\mathrm{dec}}T$ unreliable primitives to map $R^{nT}\mapsto \hat B^T$, and the downstream task stage as using at most $m_{\mathrm{task}}T$ unreliable primitives to map $\hat B^T\mapsto \hat X^T$, with
$m_{\mathrm{dec}}+m_{\mathrm{task}}\le m$.

\begin{lemma}[Information bound for each stage]
\label{lem:two_stage_bounds}
Under the $T$-block hard-separation architecture \eqref{eq:block_markov_hard_sep},
\begin{equation}
  I(R^{nT};\hat B^T)\le m_{\mathrm{dec}}T\,\Cgate,
  \qquad
  I(\hat B^T;\hat X^T)\le m_{\mathrm{task}}T\,\Cgate.
  \label{eq:stage_bounds}
\end{equation}
\end{lemma}

\begin{proof}
Each inequality follows by applying Lemma~\ref{lem:compute_info_bound} to the corresponding stage-specific bit-materialization process over the full block of length $T$.
\end{proof}

\begin{theorem}[Two-stage hard-separation converse with interface-tax specialization]
\label{thm:interface_tax}
Under the $T$-block hard-separation architecture \eqref{eq:block_markov_hard_sep}, any scheme satisfies
\begin{equation}
  \frac{1}{T}I(Y^T;\hat X^T)\ \le\ \min\big\{\, n\Cch,\ m_{\mathrm{dec}}\Cgate,\ m_{\mathrm{task}}\Cgate \big\}.
  \label{eq:interface_tax_bound}
\end{equation}
Consequently,
\begin{equation}
  \frac{1}{T}I(Y^T;\hat X^T)\ \le\ \min\Big\{\, n\Cch,\ \frac{m}{2}\Cgate \Big\}
  \label{eq:interface_tax_m_over_2}
\end{equation}
under the optimal budget split $m_{\mathrm{dec}}=m_{\mathrm{task}}=\frac{m}{2}$.
\end{theorem}

\begin{proof}
The communication cut gives $I(Y^T;\hat X^T)\le I(S^{nT};R^{nT})\le nT\Cch$.
Under hard-separation \eqref{eq:block_markov_hard_sep}, we have the Markov chain
$Y^T \to R^{nT} \to \hat B^T \to \hat X^T$,
so $I(Y^T;\hat X^T)\le I(Y^T;\hat B^T)\le I(R^{nT};\hat B^T)$ and also $I(Y^T;\hat X^T)\le I(\hat B^T;\hat X^T)$ by data processing.
Applying Lemma~\ref{lem:two_stage_bounds} yields \eqref{eq:interface_tax_bound}.
Maximizing the minimum of $(m_{\mathrm{dec}},m_{\mathrm{task}})$ subject to $m_{\mathrm{dec}}+m_{\mathrm{task}}\le m$
gives $m/2$, yielding \eqref{eq:interface_tax_m_over_2}.
\end{proof}

\begin{remark}[Committed-interface tax versus classical receiver models]
\label{rem:interface_tax_not_sep}
\label{rem:classical_noisy_receiver_relation}
Theorem~\ref{thm:interface_tax} is not a failure of Shannon source--channel separation: with reliable receiver processing, an intermediate digital description can be made asymptotically error-free and does not create an additional first-order bottleneck~\cite{Shannon48,CoverThomas}.
The additional penalty studied here arises only when task-relevant information must cross a receiver-internal committed vulnerable interface under an explicit compute budget; classical appended noisy-receiver/additional-noise models do not create such an internal cut~\cite{DobrushinTsybakov62,Sakrison68,WolfZiv70}.
Letting $m\to\infty$ only makes the compute-side cut nonbinding; it does not by itself recover the classical noisy-receiver/additional-noise closure.
In the special case with no communication bottleneck and exact-function recovery, the compute-supply term $m\Cgate$ may also be read as a converse-side ceiling analogous to Simon's computation capacity~\cite{Simon11}, but our setting keeps an explicit communication cut and indirect-distortion demand.
Equation~\eqref{eq:interface_tax_m_over_2} is therefore a structural hard-separation effect.
\end{remark}

\begin{theorem}[Two-stage achievability under the protected-support closure]
\label{thm:interface_tax_achievability}
Assume the conditions of Theorem~\ref{thm:achievability_general} and fix any split $m_{\mathrm{dec}},m_{\mathrm{task}}\ge 0$ with $m_{\mathrm{dec}}+m_{\mathrm{task}}\le m$. Then, for any distortion level $D$ such that
\begin{equation}
  R_{X|Y}(D)\ <\ \min\{n\Cch,\ m_{\mathrm{dec}}\Cgate,\ m_{\mathrm{task}}\Cgate\},
  \label{eq:interface_tax_achievability_condition}
\end{equation}
there exists a two-stage construction under Assumption~\ref{assump:achievability_benchmark} whose average distortion satisfies $\frac{1}{T}\sum_{t=1}^T\mathbb{E}[d(X_t,\hat X_t)]\le D$ in the Shannon asymptotic regime ($T\to\infty$).
\end{theorem}

\begin{proof}
Use the achievability template of Theorem~\ref{thm:achievability_general}. Choose
\[
  R_{X|Y}(D) < R < \min\{n\Cch,\ m_{\mathrm{dec}}\Cgate,\ m_{\mathrm{task}}\Cgate\},
\]
and let $W$ be the common remote-description index, reliably decoded from the physical channel as $\tilde W$.
Since $R<m_{\mathrm{dec}}\Cgate$ and $R<m_{\mathrm{task}}\Cgate$, Proposition~\ref{prop:compute_code_bsc} yields two internal codes over the stage-$1$ and stage-$2$ primitive budgets, producing a committed interface $\hat B^T\equiv \hat W_{\mathrm{dec}}$ and final index $\hat W_{\mathrm{task}}$ with vanishing stage errors.
A union bound gives $\mathbb{P}(\hat W_{\mathrm{task}}\neq W)\to0$, so on the success event $\hat X^T=\psi_T(\hat W_{\mathrm{task}})=\hat X_{\mathrm{RD}}^T$. The bounded-distortion case then follows exactly as in Theorem~\ref{thm:achievability_general}, and the squared-error case follows from Lemma~\ref{lem:clip_closes_mse_achievability}.
\end{proof}

\subsection{Serial $K$-Stage Extensions}\label{subsec:k_stage_tax}
We now extend the two-stage hard-separation converse to serial $K$-stage pipelines, where each additional committed interface creates another compute cut.

Practical receivers often contain \emph{multiple} serial digital interfaces beyond the basic ``decode--then--task'' split.
Examples include pipelines of the form
\emph{decode $\rightarrow$ decrypt $\rightarrow$ decompress $\rightarrow$ infer},
or multi-module baseband stacks in which intermediate representations are stored and reloaded across the memory hierarchy.
On an unreliable substrate, each such serial interface can become a mandatory information bottleneck.

We capture this phenomenon by considering a $K$-stage serial organization with intermediate interfaces $U_1,\dots,U_{K-1}$.
For a $T$-block scheme, let $U_0\triangleq R^{nT}$ denote the channel output block and $U_K\triangleq \hat X^T$ denote the final task-output block.
Stage $k\in\{1,\dots,K\}$ maps $U_{k-1}$ to $U_k$ using at most $m_kT$ unreliable primitives (bit materializations),
with a total budget $\sum_{k=1}^K m_k\le m$.

\begin{theorem}[Serial $K$-stage compute-cut converse]
\label{thm:k_stage_tax}
Consider any receiver architecture that induces the Markov chain
\begin{equation}
  Y^T \to S^{nT} \to U_0(=R^{nT}) \to U_1 \to \cdots \to U_{K-1} \to U_K(=\hat X^T),
\end{equation}
and suppose each stage $U_{k-1}\to U_k$ is implemented on the noisy-computation substrate using at most $m_kT$ unreliable primitives.
Then
\begin{equation}
  \frac{1}{T}I(Y^T;\hat X^T)\ \le\ \min\Big\{\, n\Cch,\ m_1\Cgate,\ m_2\Cgate,\ \dots,\ m_K\Cgate \Big\}.
  \label{eq:k_stage_tax_bound}
\end{equation}
In particular, under a fixed total compute budget $\sum_{k=1}^K m_k\le m$, the right-hand side is maximized by equal splitting
$m_k=m/K$ and yields the simplified bound
\begin{equation}
  \frac{1}{T}I(Y^T;\hat X^T)\ \le\ \min\Big\{\, n\Cch,\ \frac{m}{K}\Cgate \Big\}.
  \label{eq:k_stage_tax_m_over_K}
\end{equation}
\end{theorem}

\begin{proof}
The communication cut gives $I(Y^T;\hat X^T)\le I(S^{nT};R^{nT})\le nT\Cch$.
Moreover, for each $k=1,\dots,K$, the Markov chains $Y^T\to U_k\to \hat X^T$ and $Y^T\to U_{k-1}\to U_k$ imply, by data processing, that
\begin{equation}
  I(Y^T;\hat X^T)\ \le\ I(Y^T;U_k)\ \le\ I(U_{k-1};U_k).
\end{equation}
Applying Lemma~\ref{lem:compute_info_bound}
to each stage over the full block (with input $U_{k-1}$ and output $U_k$) yields
$I(U_{k-1};U_k)\le m_kT\Cgate$.
Taking the minimum over $k=1,\dots,K$ yields \eqref{eq:k_stage_tax_bound}.
Finally, maximizing $\min_k m_k$ subject to $\sum_k m_k\le m$ is achieved by equal splitting, yielding \eqref{eq:k_stage_tax_m_over_K}.
\end{proof}

\begin{corollary}[Serial $K$-stage achievability under the protected-support closure]
\label{cor:k_stage_tax_achievability}
Assume the conditions of Theorem~\ref{thm:achievability_general}. Fix stage budgets $m_1,\dots,m_K\ge 0$ with
\[
  \sum_{k=1}^K m_k \le m.
\]
Then, for any distortion level $D$ such that
\begin{equation}
  R_{X|Y}(D)\ <\ \min\{n\Cch,\ m_1\Cgate,\ \dots,\ m_K\Cgate\},
  \label{eq:k_stage_tax_achievability_condition}
\end{equation}
there exists a serial $K$-stage construction under Assumption~\ref{assump:achievability_benchmark} whose average distortion satisfies $\frac{1}{T}\sum_{t=1}^T\mathbb{E}[d(X_t,\hat X_t)]\le D$ in the Shannon asymptotic regime ($T\to\infty$). In particular, under equal splitting $m_k=m/K$, any $D$ satisfying
\[
  R_{X|Y}(D) < \min\{n\Cch,\ (m/K)\Cgate\}
\]
is achievable under Assumption~\ref{assump:achievability_benchmark}.
\end{corollary}

\begin{proof}
Repeat the common-index construction of Theorem~\ref{thm:interface_tax_achievability}, replacing the two vulnerable rematerialization stages by $K$ stages with budgets $m_1T,\dots,m_KT$.
For any
\[
  R_{X|Y}(D) < R < \min\{n\Cch,\ m_1\Cgate,\ \dots,\ m_K\Cgate\},
\]
the physical channel decodes the common index reliably, and Proposition~\ref{prop:compute_code_bsc} yields stagewise internal codes for all $K$ vulnerable stages.
A union bound gives vanishing overall index error, so the bounded-distortion and squared-error conclusions follow exactly as in Theorem~\ref{thm:achievability_general}. The equal-splitting specialization is immediate.
\end{proof}

\subsection{General DAG Receiver-Internal Compute Min-Cut Theorem}\label{subsec:dag_cutset}
Theorem~\ref{thm:k_stage_tax} covers purely serial organizations. We now move to a specified DAG interface decomposition of the receiver computation, where each $s_{\mathcal G}$--$t_{\mathcal G}$ cut becomes a candidate compute bottleneck.

The serial model of Subsection~\ref{subsec:k_stage_tax} captures pipelines in which task-relevant information must traverse a \emph{single} path of committed digital interfaces. Many practical receiver stacks, however, contain parallel branches (e.g., multi-thread baseband processing, cached side information, or architectural ``skip'' connections) so that information can traverse \emph{multiple} interface paths. The next theorem extends the same interface-based bound to general \emph{receiver computation graphs}.

Consider a directed acyclic graph (DAG) $G=(\mathcal{V},\mathcal{E})$ with a designated \emph{source} node $s_{\mathcal G}\in\mathcal{V}$ and \emph{sink} node $t_{\mathcal G}\in\mathcal{V}$.
Over a task block of length $T$, node $s_{\mathcal G}$ holds the channel output block $U_{s_{\mathcal G}} \triangleq R^{nT}$ and node $t_{\mathcal G}$ outputs the final task block $U_{t_{\mathcal G}} \triangleq \hat X^T$.
Each edge $e=(u,v)\in\mathcal{E}$ represents a committed interface variable $W_e$ that is generated from the tail-node state $U_u$ on the noisy substrate using at most $m_eT$ unreliable primitives, and is then available to the head node $v$.\footnote{Feedback and iterative processing can be accommodated by unfolding iterations into a larger DAG over time; the same cut-set argument then applies to the unfolded graph.}
We assume that the induced joint law is consistent with this DAG in the sense that all downstream computations depend on $R^{nT}$ only through the collection of committed interface variables $\{W_e\}_{e\in\mathcal{E}}$.

\begin{definition}[DAG interface representation]
\label{def:dag_interface_rep}
A $T$-block receiver scheme is said to admit a \emph{DAG interface representation} on $G=(\mathcal{V},\mathcal{E})$ with source $s_{\mathcal G}$ and sink $t_{\mathcal G}$ if there exist node-state random variables $\{U_v\}_{v\in\mathcal{V}}$, committed interface variables $\{W_e\}_{e\in\mathcal{E}}$, node-private randomness variables $\{\zeta_v\}_{v\neq s_{\mathcal G}}$, and for each edge $e=(u,v)$ an edge-private randomness variable $\xi_e$, fresh primitive-noise variables $N_{e,1},\dots,N_{e,M_e}$, and a deterministic integer $M_e\le m_eT$, such that $U_{s_{\mathcal G}}=R^{nT}$ and $U_{t_{\mathcal G}}=\hat X^T$, and for every node $v\neq s_{\mathcal G}$,
\[
  U_v = g_v\!\big(W_{\mathrm{In}(v)},\,\zeta_v\big),
\]
where $W_{\mathrm{In}(v)}\triangleq\{W_e: e=(u,v)\in\mathcal{E}\}$ are the incoming interface variables.
Moreover, for each edge $e=(u,v)$ there is an edge-local noisy subroutine of length $M_e$ with transcript $Z_e^{M_e}\triangleq (Z_{e,1},\dots,Z_{e,M_e})$ such that
\[
  A_{e,i}=\alpha_{e,i}\!\big(U_u,Z_e^{i-1},\xi_e\big),
  \qquad
  Z_{e,i}\sim W_{Z|U}(\cdot\mid A_{e,i}),
  \qquad i=1,\dots,M_e,
\]
and
\[
  W_e=\beta_e\!\big(Z_e^{M_e},\xi_e\big).
\]
The variables $N_{e,1},\dots,N_{e,M_e}$ are the fresh primitive noises that realize the conditional draws $Z_{e,i}\sim W_{Z|U}(\cdot\mid A_{e,i})$.
Equivalently, there exists a measurable function $\phi$ such that
\[
  Z_{e,i}=\phi(A_{e,i},N_{e,i}),
\]
where the distribution of $N_{e,i}$ is chosen so that $Z_{e,i}\mid A_{e,i}=a \sim W_{Z|U}(\cdot\mid a)$.
The variable $W_e$ is the only information from $u$ made available to the head node $v$.
Moreover, the family
\[
  \{\zeta_v\}_{v\neq s_{\mathcal G}}
  \;\cup\;
  \{\xi_e\}_{e\in\mathcal{E}}
  \;\cup\;
  \{N_{e,i}: e\in\mathcal{E},\ 1\le i\le M_e\}
\]
is mutually independent and jointly independent of the upstream variables $(Y^T,R^{nT})$, reflecting the operational model in which these are fresh receiver-internal randomness variables generated independently of the source/channel realization.
\end{definition}

The graph $G$ represents a specified interface decomposition of the receiver architecture.
It does not assert that arbitrary node fusion or node splitting is free; such transformations are admissible only when the corresponding intermediate state is not treated as a committed vulnerable interface, or when its cost is accounted for elsewhere.

For a vertex subset $\Omega\subseteq\mathcal{V}$, define the outgoing edge boundary
\begin{equation}
  \partial^+(\Omega)\ \triangleq\ \big\{(u,v)\in\mathcal{E}: u\in\Omega,\ v\in\Omega^c\big\}.
  \label{eq:delta_plus_def}
\end{equation}

\begin{lemma}[Cut-interface separation]
\label{lem:cut_interface_sep}
Under Definition~\ref{def:dag_interface_rep}, fix any cut $\Omega\subseteq\mathcal{V}$ with $s_{\mathcal G}\in\Omega$ and $t_{\mathcal G}\notin\Omega$, and let
\[
  W_{\partial^+(\Omega)} \triangleq \{W_e : e\in\partial^+(\Omega)\}.
\]
Also define
\[
  E_{\Omega^c}^{\mathrm{int}}
  \triangleq
  \{(u,v)\in\mathcal{E}: u,v\in\Omega^c\},
\]
and
\[
  S_{\Omega^c}
  \triangleq
  \Bigl(
    \{\zeta_v:v\in\Omega^c\},
    \{\xi_e:e\in E_{\Omega^c}^{\mathrm{int}}\},
    \{N_{e,i}: e\in E_{\Omega^c}^{\mathrm{int}},\ 1\le i\le M_e\}
  \Bigr).
\]
Then the following two properties hold:
\emph{(i)} there exists a measurable function $\psi_\Omega$ such that
\[
  \hat X^T
  =
  U_{t_{\mathcal G}}
  =
  \psi_\Omega\!\big(W_{\partial^+(\Omega)},S_{\Omega^c}\big);
\]
\emph{(ii)} $S_{\Omega^c}\perp R^{nT}$.
Consequently,
\[
  I(R^{nT};\hat X^T)
  \;\le\;
  I\!\big(R^{nT};W_{\partial^+(\Omega)},S_{\Omega^c}\big)
  \;=\;
  I\!\big(R^{nT};W_{\partial^+(\Omega)}\,\big|\,S_{\Omega^c}\big).
\]
\end{lemma}

\begin{proof}
\emph{Proof of (i).} Because $G$ is acyclic, the nodes in $\Omega^c$ can be ordered topologically as $v_1,\dots,v_m$. We claim by induction over this order that, for every $r\in\{1,\dots,m\}$, the node state $U_{v_r}$ and every interface variable $W_e$ on an internal edge $e=(v_r,w)$ with $w\in\Omega^c$ are measurable functions of $\bigl(W_{\partial^+(\Omega)},S_{\Omega^c}\bigr)$.

For $r=1$, every incoming edge to $v_1$ either comes from $\Omega$ and therefore belongs to $\partial^+(\Omega)$, or does not exist. Hence
\[
  U_{v_1}
  =
  g_{v_1}\!\big(W_{\mathrm{In}(v_1)},\,\zeta_{v_1}\big)
\]
is a measurable function of $\bigl(W_{\partial^+(\Omega)},S_{\Omega^c}\bigr)$. If $e=(v_1,w)$ is any edge with $w\in\Omega^c$, then the edge-local recursion
\[
  A_{e,i}=\alpha_{e,i}\!\big(U_{v_1},Z_e^{i-1},\xi_e\big),
  \qquad
  Z_{e,i}=\phi(A_{e,i},N_{e,i}),
  \qquad
  W_e=\beta_e\!\big(Z_e^{M_e},\xi_e\big)
\]
shows that $W_e$ is also a measurable function of $\bigl(W_{\partial^+(\Omega)},S_{\Omega^c}\bigr)$.

Assume the claim holds up to $r-1$. Consider node $v_r$. Every incoming edge from $\Omega$ belongs to $\partial^+(\Omega)$, while every incoming edge from $\Omega^c$ originates at one of $v_1,\dots,v_{r-1}$. By the induction hypothesis, all such internal incoming interface variables are measurable functions of $\bigl(W_{\partial^+(\Omega)},S_{\Omega^c}\bigr)$. Therefore
\[
  U_{v_r}
  =
  g_{v_r}\!\big(W_{\mathrm{In}(v_r)},\,\zeta_{v_r}\big)
\]
is a measurable function of $\bigl(W_{\partial^+(\Omega)},S_{\Omega^c}\bigr)$. The same edge-local recursion then shows that every internal outgoing edge $e=(v_r,w)$ with $w\in\Omega^c$ has interface variable $W_e$ measurable with respect to $\bigl(W_{\partial^+(\Omega)},S_{\Omega^c}\bigr)$. This completes the induction.

Since $t_{\mathcal G}\in\Omega^c$, we obtain
\[
  \hat X^T
  =
  U_{t_{\mathcal G}}
  =
  \psi_\Omega\!\big(W_{\partial^+(\Omega)},S_{\Omega^c}\big)
\]
for some measurable function $\psi_\Omega$.

\emph{Proof of (ii).} $S_{\Omega^c}$ is a subcollection of the family $\{\zeta_v\}_{v\neq s_{\mathcal G}}\cup\{\xi_e\}_{e\in\mathcal{E}}\cup\{N_{e,i}\}_{e,i}$, which by Definition~\ref{def:dag_interface_rep} is jointly independent of $(Y^T,R^{nT})$; in particular, $S_{\Omega^c}\perp R^{nT}$.

\emph{Consequence.} From (i) and data processing, $I(R^{nT};\hat X^T)\le I\!\big(R^{nT};W_{\partial^+(\Omega)},S_{\Omega^c}\big)$. By the chain rule and (ii),
\[
  I\!\big(R^{nT};W_{\partial^+(\Omega)},S_{\Omega^c}\big)
  = I(R^{nT};S_{\Omega^c})
  + I\!\big(R^{nT};W_{\partial^+(\Omega)}\,\big|\,S_{\Omega^c}\big)
  = I\!\big(R^{nT};W_{\partial^+(\Omega)}\,\big|\,S_{\Omega^c}\big),
\]
which proves the displayed inequality.
\end{proof}

\begin{theorem}[Receiver-internal compute min-cut converse under committed interfaces]
\label{thm:dag_cutset}
Consider any receiver architecture that admits a DAG interface representation in the sense of Definition~\ref{def:dag_interface_rep}, and suppose that each interface edge $e\in\mathcal{E}$ is implemented using at most $m_eT$ unreliable primitives on the compute substrate (Definition~\ref{def:gpu_bitflip_model}).
Then, for every cut $\Omega\subseteq\mathcal{V}$ with $s_{\mathcal G}\in\Omega$ and $t_{\mathcal G}\notin\Omega$,
\begin{equation}
  \frac{1}{T}I(Y^T;\hat X^T)
  \ \le\
  \min\Big\{\, n\Cch,\ \sum_{e\in\partial^+(\Omega)} m_e\,\Cgate \Big\}.
  \label{eq:dag_cutset_bound_one}
\end{equation}
Equivalently,
\begin{equation}
  \frac{1}{T}I(Y^T;\hat X^T)
  \ \le\
  \min\Big\{\, n\Cch,\ \min_{\Omega:\,s_{\mathcal G}\in\Omega,\,t_{\mathcal G}\notin\Omega}\ \sum_{e\in\partial^+(\Omega)} m_e\,\Cgate \Big\}.
  \label{eq:dag_cutset_bound}
\end{equation}
\end{theorem}

\begin{proof}
The physical-channel cut gives $I(Y^T;\hat X^T)\le nT\Cch$ as in the proof of Theorem~\ref{thm:k_stage_tax}.
Fix a cut $\Omega$ with $s_{\mathcal G}\in\Omega$ and $t_{\mathcal G}\notin\Omega$, and collect the interface variables across the cut as
$W_{\partial^+(\Omega)}\triangleq \{W_e: e\in\partial^+(\Omega)\}$.
By Definition~\ref{def:dag_interface_rep}, $\hat X^T$ is a measurable function of $R^{nT}(=U_{s_{\mathcal G}})$ and of receiver-private randomness that is jointly independent of $(Y^T,R^{nT})$; hence
\[
  Y^T \to R^{nT} \to \hat X^T
\]
forms a Markov chain. Combining this with Lemma~\ref{lem:cut_interface_sep} and data processing yields
\begin{equation}
  I(Y^T;\hat X^T)
  \;\le\; I(R^{nT};\hat X^T)
  \;\le\; I\!\big(R^{nT};W_{\partial^+(\Omega)}\,\big|\,S_{\Omega^c}\big).
  \label{eq:dag_dp}
\end{equation}

Order the cut edges as $e_1,\dots,e_q$ so that their tail nodes appear in a topological order of $G$, breaking ties arbitrarily, and define
\[
  e_j=(u_j,v_j),
  \qquad
  W_{<j}\triangleq (W_{e_1},\dots,W_{e_{j-1}}).
\]
The chain rule then gives
\[
  I\!\big(R^{nT};W_{\partial^+(\Omega)}\,\big|\,S_{\Omega^c}\big)
  =\sum_{j=1}^q I\!\big(R^{nT};W_{e_j}\,\big|\,W_{<j},S_{\Omega^c}\big).
\]

Fix $j$. Under this ordering, for every $i\le j$ the vertex $u_i$ is not a descendant of $v_j$ in $G$, so $(U_{u_j},W_{<j})$ is a measurable function of $R^{nT}$ and of local randomness attached to nodes and edges other than $e_j$. Moreover, $S_{\Omega^c}$ contains no local randomness attached to $e_j$, since $e_j\in\partial^+(\Omega)$ is not an internal $\Omega^c$-edge. By Definition~\ref{def:dag_interface_rep}, the pair $(\xi_{e_j},N_{e_j,1},\dots,N_{e_j,M_{e_j}})$ is jointly independent of $R^{nT}$ and of every other member of the local-randomness family, and is therefore independent of $\bigl(R^{nT},U_{u_j},W_{<j},S_{\Omega^c}\bigr)$. Since $W_{e_j}=\beta_{e_j}(Z_{e_j}^{M_{e_j}},\xi_{e_j})$ depends on the conditioning variables only through $(U_{u_j},\xi_{e_j},N_{e_j,1:M_{e_j}})$, we have the Markov chain
\[
  R^{nT} \to (U_{u_j},W_{<j},S_{\Omega^c}) \to W_{e_j},
\]
and hence
\[
  I\!\big(R^{nT};W_{e_j}\,\big|\,W_{<j},S_{\Omega^c}\big)
  \le
  I\!\big(U_{u_j};W_{e_j}\,\big|\,W_{<j},S_{\Omega^c}\big).
\]
Because $\xi_{e_j}\perp\bigl(U_{u_j},W_{<j},S_{\Omega^c}\bigr)$,
\[
  I\!\big(U_{u_j};W_{e_j}\,\big|\,W_{<j},S_{\Omega^c}\big)
  \le
  I\!\big(U_{u_j};W_{e_j},\xi_{e_j}\,\big|\,W_{<j},S_{\Omega^c}\big)
  =
  I\!\big(U_{u_j};W_{e_j}\,\big|\,W_{<j},S_{\Omega^c},\xi_{e_j}\big).
\]

For each fixed realization $(W_{<j}=w_{<j},S_{\Omega^c}=s,\xi_{e_j}=\xi)$, the edge-local subroutine generating $W_{e_j}$ is an $M_{e_j}$-use interactive noisy-primitive module with input $U_{u_j}$; after fixing $\xi$, the output
\[
  W_{e_j}=\beta_{e_j}\!\big(Z_{e_j}^{M_{e_j}},\xi\big)
\]
is a measurable function of the transcript $Z_{e_j}^{M_{e_j}}$. We do \emph{not} condition on the primitive noises $N_{e_j,1},\dots,N_{e_j,M_{e_j}}$ (doing so would collapse each primitive use to a deterministic map and void the $\Cgate$ capacity bound); by Definition~\ref{def:dag_interface_rep}, they remain independent of $\bigl(U_{u_j},W_{<j},S_{\Omega^c},\xi_{e_j}\bigr)$, so the conditional-capacity hypothesis of Lemma~\ref{lem:compute_info_bound} is preserved and
\[
  I\!\big(U_{u_j};W_{e_j}\,\big|\,W_{<j}=w_{<j},S_{\Omega^c}=s,\xi_{e_j}=\xi\big)
  \le
  M_{e_j}\Cgate
  \le
  m_{e_j}T\,\Cgate.
\]
Averaging over $(W_{<j},S_{\Omega^c},\xi_{e_j})$ and combining with the preceding inequalities gives
\[
  I\!\big(R^{nT};W_{e_j}\,\big|\,W_{<j},S_{\Omega^c}\big)
  \le
  m_{e_j}T\,\Cgate.
\]

Summing over $j$ yields
\[
  I\!\big(R^{nT};W_{\partial^+(\Omega)}\,\big|\,S_{\Omega^c}\big)
  \le
  T\sum_{e\in\partial^+(\Omega)} m_e\Cgate.
\]
Combining this with \eqref{eq:dag_dp} and the physical-channel cut, and dividing by $T$, gives \eqref{eq:dag_cutset_bound_one}. Minimizing over cuts yields \eqref{eq:dag_cutset_bound}.
\end{proof}

\begin{example}[Worked toy example: serial depth tax vs.\ parallel/skip min-cuts]
\label{ex:serial_parallel_mincut}
Assume the \regcomp{} and fix a total interface budget $\sum_{e\in\mathcal{E}} m_e = m$ per task instance.
Figure~\ref{fig:toy_mincut_graphs} depicts the three corresponding receiver computation graphs.
The min-cut term in \eqref{eq:dag_cutset_bound} then yields the following one-line comparisons:
\begin{align}
  \text{(serial $K$-edge chain)}\qquad
  \min_{\Omega}\sum_{e\in\partial^+(\Omega)} m_e
  &= \min_{k=1,\dots,K} m_k \ \le\ \frac{m}{K},
  \label{eq:toy_serial_m_over_k}\\
  \text{(two parallel edges $s_{\mathcal G}\!\to\! t_{\mathcal G}$)}\qquad
  \min_{\Omega}\sum_{e\in\partial^+(\Omega)} m_e
  &= m_1+m_2\ =\ m,
  \label{eq:toy_parallel_m}\\
  \text{(serial chain + skip edge)}\qquad
  \min_{\Omega}\sum_{e\in\partial^+(\Omega)} m_e
  &= m_{\mathrm{skip}}+\min_{k} m_k
   = \Big(1-\rho_{\mathrm{skip}}+\frac{\rho_{\mathrm{skip}}}{K}\Big)m \xrightarrow[\rho_{\mathrm{skip}}\to 0]{} m,
  \label{eq:toy_skip_near_m}
\end{align}
where \eqref{eq:toy_serial_m_over_k} is maximized by equal splitting $m_k=m/K$,
and \eqref{eq:toy_skip_near_m} uses the allocation $m_{\mathrm{skip}}=(1-\rho_{\mathrm{skip}})m$ and $m_k=\rho_{\mathrm{skip}} m/K$.
Multiplying by $\Cgate$ gives the corresponding compute supply terms in \eqref{eq:dag_cutset_bound}.
\end{example}

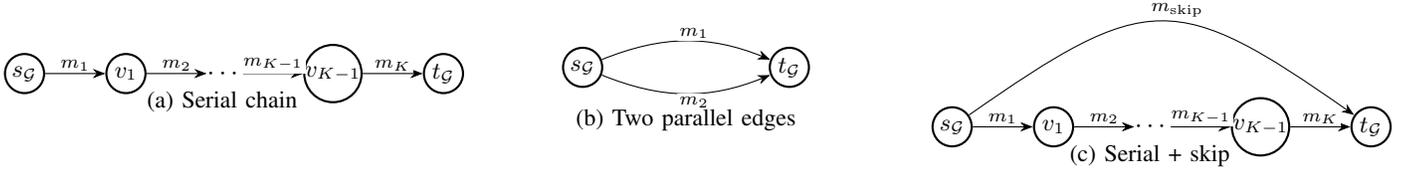
\begin{figure}[t]
  \centering

  \begin{minipage}{0.32\linewidth}
    \centering
    \begin{tikzpicture}[>=Stealth, node distance=0.8cm, baseline=(current bounding box.center)]
      \node[vtx] (s) {$s_{\mathcal G}$};
      \node[vtx] (v1) [right=of s] {$v_1$};
      \node[draw=none,inner sep=0pt] (dots) [right=of v1] {$\cdots$};
      \node[vtx] (vK) [right=of dots] {$v_{K-1}$};
      \node[vtx] (t)  [right=of vK] {$t_{\mathcal G}$};

      \draw[->] (s)  -- node[elbl,above] {$m_1$}     (v1);
      \draw[->] (v1) -- node[elbl,above] {$m_2$}     (dots);
      \draw[->] (dots)-- node[elbl,above] {$m_{K-1}$} (vK);
      \draw[->] (vK) -- node[elbl,above] {$m_K$}     (t);
    \end{tikzpicture}\\[-1ex]
    \small (a) Serial chain
  \end{minipage}
  \hfill
  \begin{minipage}{0.32\linewidth}
    \centering
    \begin{tikzpicture}[>=Stealth, node distance=2.2cm, baseline=(current bounding box.center)]
      \node[vtx] (s) {$s_{\mathcal G}$};
      \node[vtx] (t) [right=of s] {$t_{\mathcal G}$};

      \draw[->] (s) to[bend left=22]
        node[elbl,above,pos=0.55] {$m_1$} (t);
      \draw[->] (s) to[bend right=22]
        node[elbl,below,pos=0.55] {$m_2$} (t);
    \end{tikzpicture}\\[-0.03ex]
    \small (b) Two parallel edges
  \end{minipage}
  \hfill
  \begin{minipage}{0.32\linewidth}
    \centering
    \begin{tikzpicture}[>=Stealth, node distance=0.8cm, baseline=(current bounding box.center)]
      \node[vtx] (s) {$s_{\mathcal G}$};
      \node[vtx] (v1) [right=of s] {$v_1$};
      \node[draw=none,inner sep=0pt] (dots) [right=of v1] {$\cdots$};
      \node[vtx] (vK) [right=of dots] {$v_{K-1}$};
      \node[vtx] (t)  [right=of vK] {$t_{\mathcal G}$};

      \draw[->] (s)  -- node[elbl,above] {$m_1$}      (v1);
      \draw[->] (v1) -- node[elbl,above] {$m_2$}      (dots);
      \draw[->] (dots)-- node[elbl,above] {$m_{K-1}$} (vK);
      \draw[->] (vK) -- node[elbl,above] {$m_K$}      (t);

      \draw[->] (s) to[out=35,in=145,looseness=1.45]
        node[elbl,above,pos=0.55,yshift=1pt] {$m_{\mathrm{skip}}$} (t);
    \end{tikzpicture}\\[-1ex]
    \small (c) Serial + skip
  \end{minipage}

  \caption{Toy computation graphs for Example~\ref{ex:serial_parallel_mincut}.
  Under a fixed total interface budget $m$, the compute supply term in Theorem~\ref{thm:dag_cutset} is governed by the min-cut across the receiver computation graph:
  (a) a serial chain yields $\min_k m_k \le m/K$,
  (b) two parallel edges yield $m_1+m_2=m$, and
  (c) adding a skip edge yields $m_{\mathrm{skip}}+\min_k m_k$, which can approach $m$ when most budget is allocated to the skip.}
  \label{fig:toy_mincut_graphs}
\end{figure}

\begin{remark}[Serial chains and skip connections as special cases]
If $G$ is a serial chain with $K$ edges, every $s_{\mathcal G}$--$t_{\mathcal G}$ cut crosses exactly one edge and \eqref{eq:dag_cutset_bound} reduces to Theorem~\ref{thm:k_stage_tax}. Example~\ref{ex:serial_parallel_mincut} quantifies this depth-induced loss under a fixed total budget.
More generally, adding a ``skip'' edge that carries side information from an earlier node to a later node creates \emph{parallel} paths and can strictly increase the min-cut term in \eqref{eq:dag_cutset_bound}, thereby mitigating depth-induced bottlenecks.
This provides an information-theoretic rationale for hybrid interfaces and architectural bypasses beyond the purely serial model.
\end{remark}

\begin{remark}[Heterogeneous primitives and reliable bypasses]
The cut-set term in \eqref{eq:dag_cutset_bound} depends on the compute substrate only through per-use information bounds.
If an interface edge $e$ is implemented using heterogeneous primitives with per-class budgets $(m_{e,1},\dots,m_{e,J})$ per task instance, then the contribution of $e$ becomes
$\sum_{j=1}^J m_{e,j}\Cgate^{(j)}$ by Lemma~\ref{lem:compute_info_bound_hetero}.
Likewise, a $b$-bit reliable bypass (Remark~\ref{rem:bypass_model}) can be represented as an additional reliable edge of capacity $b$ (bits per task instance), which simply adds to the corresponding cut capacity and recovers the bypass relaxation in Corollary~\ref{cor:supply_demand_bypass}.
\end{remark}

\begin{discussion}{Compute supply as a min-cut of the receiver computation graph}
Theorem~\ref{thm:dag_cutset} also suggests a max--min allocation problem over interface budgets $m_e$: equal splitting is optimal in serial architectures (Theorem~\ref{thm:k_stage_tax}), whereas multi-branch architectures can exploit parallel paths to relax depth-induced bottlenecks.
\end{discussion}

\begin{remark}[Relation of Theorem~\ref{thm:dag_cutset} to classical network cut-sets]
Theorem~\ref{thm:dag_cutset} is formally cut-set-like, but its content differs from a classical communication-network cut-set in two respects. First, the cut edges are receiver-internal committed interfaces rather than external links with prescribed capacities, and each contribution $m_e\Cgate$ is derived from Lemma~\ref{lem:compute_info_bound} for the noisy materialization of $W_e$, not postulated as an exogenous link throughput. Second, together with Lemma~\ref{lem:ird_block_lower_bound}, the theorem limits the task-relevant mutual information needed to attain distortion $D$ through $R_{X|Y}(D)$, rather than the rate of reliably communicating an independent message. The theorem thus turns committed materialization into induced compute cut values and couples those values to indirect rate--distortion demand.
\end{remark}

\subsection{Plug-Ins and Architectural Variants}\label{subsec:hetero_interface_tax}
Theorems~\ref{thm:interface_tax}, \ref{thm:k_stage_tax}, and~\ref{thm:dag_cutset} remain the baseline references throughout this subsection: the variants below act only by changing the active cut or its effective capacity.

\begin{corollary}[Hard-separation interface tax with a $b$-bit reliable bypass]
\label{cor:interface_tax_bypass}
Consider the hard-separation architecture \eqref{eq:block_markov_hard_sep}, but suppose the downstream task stage also has access to a protected bypass variable $V^T$ satisfying $H(V^T)\le bT$ and produced from $R^{nT}$ through a reliable path.
That is, $\hat X^T=g(\hat B^T,V^T)$, where $\hat B^T$ is generated from $R^{nT}$ using $m_{\mathrm{dec}}T$ vulnerable primitive uses and the task stage uses $m_{\mathrm{task}}T$ vulnerable primitive uses.
Then any scheme satisfies
\begin{equation}
  \frac{1}{T}I(Y^T;\hat X^T)
  \le
  \min\big\{\, n\Cch,\ b+m_{\mathrm{dec}}\Cgate,\ b+m_{\mathrm{task}}\Cgate \big\}.
  \label{eq:interface_tax_bypass_bound}
\end{equation}
In particular, under the total vulnerable-budget constraint $m_{\mathrm{dec}}+m_{\mathrm{task}}\le m$,
\begin{equation}
  \frac{1}{T}I(Y^T;\hat X^T)
  \le
  \min\Big\{\, n\Cch,\ b+\frac{m}{2}\Cgate \Big\}.
  \label{eq:interface_tax_bypass_m_over_2}
\end{equation}
Thus even a modest reliable bypass increases the usable supply under hard-separation, and a sufficiently large bypass can substantially mitigate the symmetric two-stage serial tax.
\end{corollary}

\begin{proof}
The communication cut gives $I(Y^T;\hat X^T)\le I(S^{nT};R^{nT})\le nT\Cch$.
Moreover, $Y^T\to R^{nT}\to (\hat B^T,V^T)\to \hat X^T$ forms a Markov chain, so data processing gives
$I(Y^T;\hat X^T)\le I(R^{nT};\hat B^T,V^T)$ and $I(Y^T;\hat X^T)\le I(\hat B^T,V^T;\hat X^T)$.
For the first cut, apply the chain rule in the order $\hat B^T$ then $V^T$:
\begin{equation}
I(R^{nT};\hat B^T,V^T)
  = I(R^{nT};\hat B^T) + I(R^{nT};V^T\mid \hat B^T) 
  \le I(R^{nT};\hat B^T) + H(V^T) 
  \le m_{\mathrm{dec}}T\,\Cgate + bT,
\end{equation}
where $I(R^{nT};\hat B^T)\le m_{\mathrm{dec}}T\Cgate$ follows from Lemma~\ref{lem:two_stage_bounds}.

For the second cut, similarly,
\begin{equation}
  I(\hat B^T,V^T;\hat X^T) = I(V^T;\hat X^T) + I(\hat B^T;\hat X^T\mid V^T) 
  \le H(V^T) + I(\hat B^T;\hat X^T\mid V^T) 
  \le bT + m_{\mathrm{task}}T\,\Cgate,
\end{equation}
where the last inequality is obtained by applying Lemma~\ref{lem:two_stage_bounds} to the stage-$2$ computation for each fixed $V^T=v^T$ (so the computation remains a valid $m_{\mathrm{task}}T$-primitive process) and then averaging over $V^T$.
Combining these bounds yields \eqref{eq:interface_tax_bypass_bound}.
Maximizing $\min(m_{\mathrm{dec}},m_{\mathrm{task}})$ subject to $m_{\mathrm{dec}}+m_{\mathrm{task}}\le m$ gives $m/2$, yielding \eqref{eq:interface_tax_bypass_m_over_2}.
\end{proof}

\begin{theorem}[Two-stage achievability with a protected bypass]
\label{thm:interface_tax_bypass_achievability}
Assume the conditions of Theorem~\ref{thm:achievability_general}.
Consider the two-stage architecture with a protected bypass of Corollary~\ref{cor:interface_tax_bypass}, i.e., a committed interface $\hat B^T$ produced using at most $m_{\mathrm{dec}}T$ vulnerable primitive uses, a downstream task stage using at most $m_{\mathrm{task}}T$ vulnerable primitive uses, and a protected bypass variable $V^T$ satisfying $H(V^T)\le bT$.
Then, for any distortion level $D$ such that
\begin{equation}
  R_{X|Y}(D)
  <
  \min\{\, n\Cch,\; b+m_{\mathrm{dec}}\Cgate,\; b+m_{\mathrm{task}}\Cgate \,\},
  \label{eq:interface_tax_bypass_achievability_condition}
\end{equation}
there exists a two-stage construction with a protected bypass whose average distortion satisfies
\[
  \frac{1}{T}\sum_{t=1}^T \mathbb{E}[d(X_t,\hat X_t)] \le D
\]
in the Shannon asymptotic regime $(T\to\infty)$.
\end{theorem}

\begin{proof}
Choose nonnegative rates $R_{\mathrm{p}}<b$ and
$R_{\mathrm{u}}<\min\{m_{\mathrm{dec}}\Cgate,m_{\mathrm{task}}\Cgate\}$ (with the obvious zero-rate conventions when one of these budgets vanishes) such that
\[
  R_{X|Y}(D)<R_{\mathrm{p}}+R_{\mathrm{u}}<\min\{n\Cch,\ b+m_{\mathrm{dec}}\Cgate,\ b+m_{\mathrm{task}}\Cgate\}.
\]
By Proposition~\ref{prop:operational_indirect_rd}, for all sufficiently large $T$ there is a fixed-length indirect rate--distortion code whose index can be relabeled as
$W=(W_{\mathrm{p}},W_{\mathrm{u}})$ with per-block rates arbitrarily close to $(R_{\mathrm{p}},R_{\mathrm{u}})$ and whose decoder can be written as
$\hat X_{\mathrm{RD}}^T=\psi_T(W_{\mathrm{p}},W_{\mathrm{u}})$. No successive-refinement structure is required: this split is only a relabeling of the remote-description index.
Reliable channel coding over the $nT$ channel uses recovers $(W_{\mathrm{p}},W_{\mathrm{u}})$ with vanishing error probability; set the protected bypass variable to $V^T=W_{\mathrm{p}}$, so $H(V^T)\le bT$.

The vulnerable serial part carries only $W_{\mathrm{u}}$: Proposition~\ref{prop:compute_code_bsc} gives a stage-1 internal code over the $m_{\mathrm{dec}}T$ primitive uses and a stage-2 internal code over the $m_{\mathrm{task}}T$ primitive uses, each with vanishing block error, so the reconstruction has the form
\[
  \hat X^T=\psi_T(V^T,\hat W_{\mathrm{task}}).
\]
On the complement of the union of the channel-decoding, stage-1, and stage-2 error events, we have
$V^T=W_{\mathrm{p}}$ and $\hat W_{\mathrm{task}}=W_{\mathrm{u}}$, hence $\hat X^T=\hat X_{\mathrm{RD}}^T$.
The bounded-distortion conclusion then follows exactly as in Theorem~\ref{thm:achievability_general}; under squared error with $\mathbb{E}[X^2]<\infty$, apply Lemma~\ref{lem:clip_closes_mse_achievability} to the same union-bound error event.

Under the total vulnerable-budget constraint $m_{\mathrm{dec}}+m_{\mathrm{task}}\le m$, maximizing $\min\{m_{\mathrm{dec}},m_{\mathrm{task}}\}$ gives the symmetric split $m_{\mathrm{dec}}=m_{\mathrm{task}}=m/2$, so every distortion level satisfying
\begin{equation}
  R_{X|Y}(D)
  <
  \min\Big\{\, n\Cch,\; b+\frac{m}{2}\Cgate \Big\}
  \label{eq:interface_tax_bypass_achievability_m_over_2}
\end{equation}
is achievable under Assumption~\ref{assump:achievability_benchmark}.
\end{proof}

The next corollary allows different effective primitive capacities in the two vulnerable stages (e.g., due to different voltages, ECC protection, or numerical precision).

\begin{corollary}[Interface tax with unequal stage primitive capacities]
\label{cor:interface_tax_unequal}
Consider the hard-separation architecture in which the decoder stage uses primitives with per-use information capacity
$\Cgate^{(\mathrm{dec})}$ and the downstream task stage uses primitives with per-use information capacity $\Cgate^{(\mathrm{task})}$.
If the decoder and task stages use at most $m_{\mathrm{dec}}T$ and $m_{\mathrm{task}}T$ primitives, respectively, then any scheme satisfies
\begin{equation}
  \frac{1}{T}I(Y^T;\hat X^T)
  \le
  \min\big\{\, n\Cch,\; m_{\mathrm{dec}}\,\Cgate^{(\mathrm{dec})},\; m_{\mathrm{task}}\,\Cgate^{(\mathrm{task})} \big\}.
  \label{eq:interface_tax_unequal_bound}
\end{equation}
Under a total compute-budget constraint $m_{\mathrm{dec}}+m_{\mathrm{task}}\le m$, the best split equalizes the two stage bottlenecks,
$m_{\mathrm{dec}}\Cgate^{(\mathrm{dec})}=m_{\mathrm{task}}\Cgate^{(\mathrm{task})}$, and yields
\begin{equation}
  \frac{1}{T}I(Y^T;\hat X^T)
  \le
  \min\Big\{\, n\Cch,\; m\,\frac{\Cgate^{(\mathrm{dec})}\,\Cgate^{(\mathrm{task})}}{\Cgate^{(\mathrm{dec})}+\Cgate^{(\mathrm{task})}} \Big\}.
  \label{eq:interface_tax_harmonic}
\end{equation}
In the symmetric case $\Cgate^{(\mathrm{dec})}=\Cgate^{(\mathrm{task})}=\Cgate$, \eqref{eq:interface_tax_harmonic} reduces to \eqref{eq:interface_tax_m_over_2}.
\end{corollary}

\begin{proof}
The proof is identical to Theorem~\ref{thm:interface_tax} except that the stage bounds in Lemma~\ref{lem:two_stage_bounds}
become $I(R^{nT};\hat B^T)\le m_{\mathrm{dec}}T\Cgate^{(\mathrm{dec})}$ and $I(\hat B^T;\hat X^T)\le m_{\mathrm{task}}T\Cgate^{(\mathrm{task})}$,
which follow from Lemma~\ref{lem:compute_info_bound_hetero} (or by repeating its argument).
Optimizing the bottleneck
\[
  \min\big\{\, m_{\mathrm{dec}}\Cgate^{(\mathrm{dec})},\ (m-m_{\mathrm{dec}})\Cgate^{(\mathrm{task})} \big\}
\]
over $m_{\mathrm{dec}}\in[0,m]$ yields
the equalization rule and the harmonic-mean value in \eqref{eq:interface_tax_harmonic}.
\end{proof}

Figure~\ref{fig:interface_tax_supply_split} summarizes this equalization picture and also previews how a reliable island lifts the task-stage cut.

\begin{figure}[t]
\centering
\includefigure[width=0.66\linewidth]{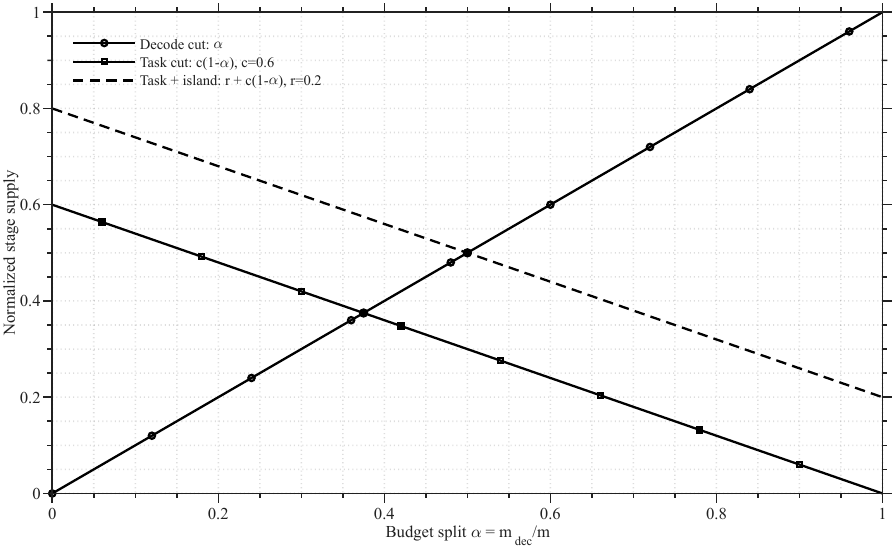}
\caption{Schematic ``equivalent supply'' picture for hard-separation with unequal stage reliabilities.
The optimal split in Corollary~\ref{cor:interface_tax_unequal} equalizes the active decode and task cuts (solid intersection), while a reliable island shifts the task-stage cut upward as in Corollary~\ref{cor:interface_tax_reliable_island} (dashed) and increases the bottleneck.}
\label{fig:interface_tax_supply_split}
\end{figure}

\begin{corollary}[Mitigating (and possibly eliminating) the interface tax using a reliable island]
\label{cor:interface_tax_reliable_island}
Consider hard-separation with stage capacities $\Cgate^{(\mathrm{dec})}$ and $\Cgate^{(\mathrm{task})}$ as in Corollary~\ref{cor:interface_tax_unequal}.
Assume the task stage can additionally use $m_{\mathrm{rel}}T$ \emph{reliable} primitive invocations (or protected memory slots) per block, each contributing at most $1$ bit/use,
and that these reliable invocations are not counted in the vulnerable budgets $(m_{\mathrm{dec}},m_{\mathrm{task}})$.
Then any scheme satisfies
\begin{equation}
  \frac{1}{T}I(Y^T;\hat X^T)
  \le
  \min\big\{\, n\Cch,\; m_{\mathrm{dec}}\,\Cgate^{(\mathrm{dec})},\; m_{\mathrm{rel}} + m_{\mathrm{task}}\,\Cgate^{(\mathrm{task})} \big\}.
  \label{eq:interface_tax_reliable_island_bound}
\end{equation}
Under the total vulnerable-budget constraint $m_{\mathrm{dec}}+m_{\mathrm{task}}\le m$, the optimal split maximizes the stage bottleneck and yields
\begin{equation}
  \frac{1}{T}I(Y^T;\hat X^T)
  \le
  \min\!\left\{\, n\Cch,\;
  \min\!\left(m\,\Cgate^{(\mathrm{dec})},\;
  \frac{\Cgate^{(\mathrm{dec})}\big(m_{\mathrm{rel}}+m\,\Cgate^{(\mathrm{task})}\big)}{\Cgate^{(\mathrm{dec})}+\Cgate^{(\mathrm{task})}}\right)\right\}.
  \label{eq:interface_tax_reliable_island_opt}
\end{equation}
In the symmetric case $\Cgate^{(\mathrm{dec})}=\Cgate^{(\mathrm{task})}=\Cgate$, this becomes
\begin{equation}
  \frac{1}{T}I(Y^T;\hat X^T)
  \le
  \min\!\left\{\, n\Cch,\; \min\!\left(m\Cgate,\; \frac{m\Cgate+m_{\mathrm{rel}}}{2}\right)\right\}
  \label{eq:interface_tax_reliable_island_symmetric}
\end{equation}
so even a small reliable island can partially or fully offset the serial tax.
\end{corollary}

\begin{proof}
The bound \eqref{eq:interface_tax_reliable_island_bound} follows from the same Markov-chain/data-processing argument as Theorem~\ref{thm:interface_tax} together with Lemma~\ref{lem:compute_info_bound_hetero} applied to the task stage with two primitive classes (reliable capacity $1$ and unreliable capacity $\Cgate^{(\mathrm{task})}$).
To optimize the bottleneck under $m_{\mathrm{task}}=m-m_{\mathrm{dec}}$, maximize
$\min\{m_{\mathrm{dec}}\Cgate^{(\mathrm{dec})},\ m_{\mathrm{rel}}+(m-m_{\mathrm{dec}})\Cgate^{(\mathrm{task})}\}$,
whose maximizer is either the intersection point or the boundary $m_{\mathrm{dec}}=m$.
Setting $\Cgate^{(\mathrm{dec})}=\Cgate^{(\mathrm{task})}=\Cgate$ yields \eqref{eq:interface_tax_reliable_island_symmetric}.
\end{proof}

\begin{corollary}[Reliable-island achievability]
\label{cor:interface_tax_reliable_island_achievability}
Assume the conditions of Theorem~\ref{thm:achievability_general}.
Consider the architecture of Corollary~\ref{cor:interface_tax_reliable_island} in the homogeneous case
\[
  \Cgate^{(\mathrm{dec})}=\Cgate^{(\mathrm{task})}=\Cgate.
\]
Suppose the task stage additionally has access to a reliable island of $m_{\mathrm{rel}}T$ protected bits per block (equivalently, $m_{\mathrm{rel}}T$ reliable primitive invocations of capacity $1$ bit/use), not charged to the vulnerable budgets $m_{\mathrm{dec}}T$ and $m_{\mathrm{task}}T$.
Fix any split $m_{\mathrm{dec}},m_{\mathrm{task}}\ge 0$ with $m_{\mathrm{dec}}+m_{\mathrm{task}}\le m$.
If
\begin{equation}
  R_{X|Y}(D)
  <
  \min\{\, n\Cch,\; m_{\mathrm{dec}}\Cgate,\; m_{\mathrm{rel}}+m_{\mathrm{task}}\Cgate \,\},
  \label{eq:interface_tax_reliable_island_achievability_condition}
\end{equation}
then there exists a two-stage construction with a reliable island whose average distortion satisfies
\[
  \frac{1}{T}\sum_{t=1}^T \mathbb{E}[d(X_t,\hat X_t)] \le D
\]
in the Shannon asymptotic regime $(T\to\infty)$.
\end{corollary}

\begin{proof}
Repeat the rate-splitting construction of Theorem~\ref{thm:interface_tax_bypass_achievability}, but replace the protected bypass part by a reliable-island part that is first rematerialized through the stage-1 vulnerable interface and then stored without error on the reliable island before the second stage. Choose rates $R_{\mathrm{rel}}<m_{\mathrm{rel}}$ and $R_{\mathrm{v}}<m_{\mathrm{task}}\Cgate$ (again with the obvious zero-rate conventions) such that
\[
  R_{X|Y}(D)<R_{\mathrm{rel}}+R_{\mathrm{v}}<\min\{n\Cch,\ m_{\mathrm{dec}}\Cgate,\ m_{\mathrm{rel}}+m_{\mathrm{task}}\Cgate\}.
\]
Relabel the indirect rate--distortion index as $W=(W_{\mathrm{rel}},W_{\mathrm{v}})$, decode this pair reliably from the channel, and use Proposition~\ref{prop:compute_code_bsc} to rematerialize it through the $m_{\mathrm{dec}}T$ stage-1 vulnerable primitives. The reliable-island portion stores $W_{\mathrm{rel}}$ without error, while only $W_{\mathrm{v}}$ is rematerialized through the $m_{\mathrm{task}}T$ stage-2 vulnerable primitives.
Let $\mathcal{E}_{\mathrm{ch}}$, $\mathcal{E}_{1}$, and $\mathcal{E}_{2}$ denote the channel-decoding, stage-1, and stage-2 error events. On the complement of $\mathcal{E}_{\mathrm{ch}}\cup \mathcal{E}_{1}\cup \mathcal{E}_{2}$, the reconstruction equals the nominal decoder output $\hat X_{\mathrm{RD}}^T$, so the bounded-distortion and squared-error conclusions follow exactly as in Theorems~\ref{thm:achievability_general}, \ref{thm:interface_tax_achievability}, and \ref{thm:interface_tax_bypass_achievability}.
Optimizing over $m_{\mathrm{dec}}+m_{\mathrm{task}}\le m$ gives
\begin{equation}
  R_{X|Y}(D)
  <
  \min\!\left\{\, n\Cch,\; \min\!\left(m\Cgate,\; \frac{m\Cgate+m_{\mathrm{rel}}}{2}\right)\right\}
  \label{eq:interface_tax_reliable_island_achievability_symmetric}
\end{equation}
as a symmetric achievability condition under Assumption~\ref{assump:achievability_benchmark}.
\end{proof}

\begin{proposition}[Single-bottleneck converse under a soft-interface relaxation]
\label{prop:no_tax_soft_interface}
Consider a soft-interface relaxation of the two-stage architecture in which the receiver first produces an intermediate representation $\hat B^T$ from $R^{nT}$ using at most $m_{\mathrm{dec}}T$ unreliable primitives,
and then outputs $\hat X^T$ using at most $m_{\mathrm{task}}T$ unreliable primitives, where the second stage has access to \emph{both} $\hat B^T$ and the raw channel output $R^{nT}$.
Here, access to $R^{nT}$ means adaptive consultation when choosing vulnerable primitive inputs; unless a protected bypass is modeled explicitly, no additional protected/direct path from $R^{nT}$ to $\hat X^T$ is present.
Then
\begin{equation}
  \frac{1}{T}I(Y^T;\hat X^T)\ \le\ \min\big\{\, n\Cch,\ (m_{\mathrm{dec}}+m_{\mathrm{task}})\Cgate \big\}.
  \label{eq:no_tax_soft_interface}
\end{equation}
In particular, no mandatory extra serial compute cut appears; when the full vulnerable budget is used ($m_{\mathrm{dec}}+m_{\mathrm{task}}=m$), the right-hand side coincides with the task-direct supply $\Rsup$ of Theorem~\ref{thm:supply_demand_converse}.
\end{proposition}
\begin{proof}
Let $L_1\triangleq m_{\mathrm{dec}}T$ and $L_2\triangleq m_{\mathrm{task}}T$.
The communication cut gives
\[
  I(Y^T;\hat X^T)\le I(S^{nT};R^{nT})\le nT\Cch.
\]
For the compute side, the first stage is a committed/no-bypass module with input $R^{nT}$ and output $\hat B^T$, so Lemma~\ref{lem:compute_info_bound} gives
\[
  I(R^{nT};\hat B^T)\le L_1\Cgate.
\]
Now condition on $\hat B^T=b$.
For each fixed $b$, the second stage remains a committed/no-bypass module: it may adaptively inspect $R^{nT}$ and the fixed side information $b$ when choosing at most $L_2$ vulnerable materializations, but any influence of $R^{nT}$ on $\hat X^T$ must still traverse those primitive uses.
Applying Lemma~\ref{lem:compute_info_bound} conditionally therefore yields
\[
  I(R^{nT};\hat X^T\mid \hat B^T=b)\le L_2\Cgate
\]
for every $b$, and averaging over $b$ gives
\[
  I(R^{nT};\hat X^T\mid \hat B^T)\le L_2\Cgate.
\]
Therefore
\[
  I(R^{nT};\hat X^T)
  \le I(R^{nT};\hat B^T,\hat X^T)
  = I(R^{nT};\hat B^T)+I(R^{nT};\hat X^T\mid \hat B^T)
  \le (L_1+L_2)\Cgate.
\]
Combining this with the communication cut and data processing proves \eqref{eq:no_tax_soft_interface}.
\end{proof}

\begin{proposition}[Soft-interface precision--budget tradeoff]
\label{prop:soft_interface_tradeoff}
Let $m_{\mathrm{int}}$ denote the per-instance budget reserved for the soft side-information path in a soft-interface design, including cached soft statistics and any implementation-dependent transport/routing support that must be charged against the vulnerable budget.
In the storage-dominated specialization where the receiver stores $L_{\mathrm{soft}}$ soft symbols using $w$ bits/symbol and each stored bit counts as one vulnerable primitive use, this reserved budget becomes
\begin{equation}
  m_{\mathrm{int}} = w\,L_{\mathrm{soft}}
  \qquad \text{(primitive-equivalent uses per task instance)}
  \label{eq:m_int_def}
\end{equation}
and more generally $m_{\mathrm{int}}$ is deducted from the nominal per-instance budget $m$ before applying Proposition~\ref{prop:no_tax_soft_interface}.
Thus any soft-interface design using budgets $(m_{\mathrm{dec}},m_{\mathrm{task}},m_{\mathrm{int}})$ must satisfy
\begin{equation}
  m_{\mathrm{int}} + m_{\mathrm{dec}} + m_{\mathrm{task}} \ \le\ m.
  \label{eq:soft_interface_feasible}
\end{equation}
Under \eqref{eq:soft_interface_feasible},
\begin{equation}
  \frac{1}{T}I(Y^T;\hat X^T)\ \le\ \min\big\{\, n\Cch,\ (m-m_{\mathrm{int}})\Cgate \big\}.
  \label{eq:soft_interface_tradeoff_bound}
\end{equation}
\end{proposition}

\begin{proof}
The feasibility constraint \eqref{eq:soft_interface_feasible} implies $m_{\mathrm{dec}}+m_{\mathrm{task}}\le m-m_{\mathrm{int}}$.
Applying Proposition~\ref{prop:no_tax_soft_interface} yields
\[
  \frac{1}{T}I(Y^T;\hat X^T)
  \le
  \min\big\{\, n\Cch,\ (m_{\mathrm{dec}}+m_{\mathrm{task}})\Cgate \big\}
  \le
  \min\big\{\, n\Cch,\ (m-m_{\mathrm{int}})\Cgate \big\},
\]
which is \eqref{eq:soft_interface_tradeoff_bound}.
\end{proof}

\begin{remark}[Streaming and low-dimensional soft state]
Streaming access or low-dimensional soft summaries can make the stored-state contribution to $m_{\mathrm{int}}$ much smaller than $n$.
This does not mean that the soft path is physically free: non-negligible transport, buffering, routing, or protection overhead should still be absorbed into $m_{\mathrm{int}}$.
\end{remark}

Writing the converse family in the generic form $\min\{n\Cch,S_{\mathrm{comp}}\}$ highlights the architecture dependence of the compute-side supply.
For task-direct processing and the ideal low-overhead soft-interface limit, $S_{\mathrm{comp}}=m\Cgate$; for homogeneous hard-separation, $S_{\mathrm{comp}}=\frac{m}{2}\Cgate$ under the optimal split.
With a reserved soft-support budget, the usable vulnerable supply becomes $S_{\mathrm{comp}}=(m-m_{\mathrm{int}})\Cgate$, while a reliable island lifts the task-side cut to $S_{\mathrm{comp}}=\min\!\left(m\Cgate,\frac{m\Cgate+m_{\mathrm{rel}}}{2}\right)$ in the symmetric case.
Figure~\ref{fig:arch_phase_supply_plane} collects these boundaries in a single supply-plane view; for a scalar Gaussian distortion witness of the same compute-limited tax, see Figure~\ref{fig:interface_tax_gaussian}.

\begin{figure}[t]
\centering
\includefigure[width=0.7\linewidth]{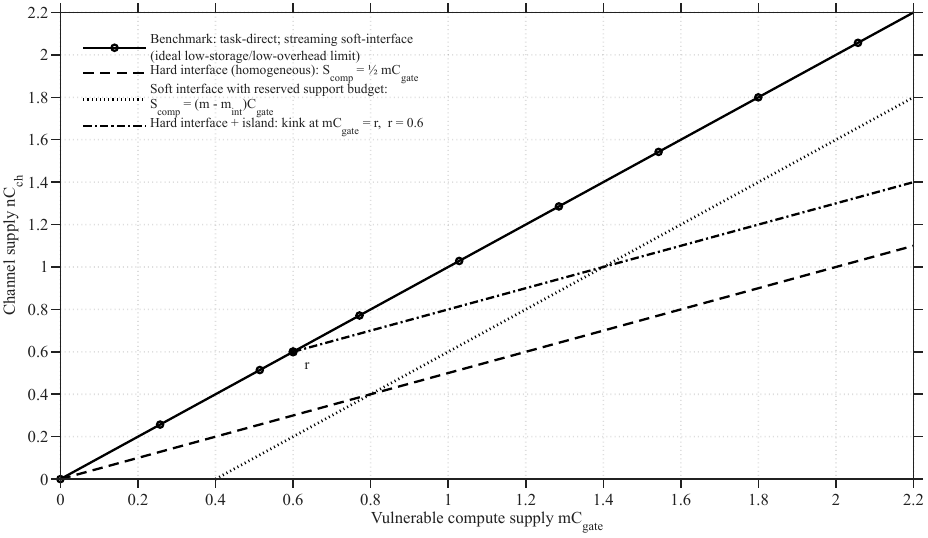}
\vspace{-3mm}
\caption{Unified ``architecture phase diagram'' in the supply plane.
Each curve separates the channel-limited and compute-limited regions for one interface organization.
The solid diagonal is the task-direct boundary (and the ideal low-overhead soft-interface limit), the dashed line is the homogeneous hard-separation boundary from Theorem~\ref{thm:interface_tax}, the dotted line is the reserved-budget soft-interface boundary from Proposition~\ref{prop:soft_interface_tradeoff}, and the dash-dotted kinked curve is the reliable-island boundary from Corollaries~\ref{cor:interface_tax_reliable_island} and~\ref{cor:interface_tax_reliable_island_achievability}.}
\label{fig:arch_phase_supply_plane}
\end{figure}

\subsection{Distortion Consequences of Serial Compute Cuts}\label{subsec:hard_sep_consequences}
The serial compute-cut converses translate immediately into distortion lower bounds by combining them with the indirect rate--distortion demand.

\begin{corollary}[Hard-separation supply--demand bound with interface-tax specialization]
\label{cor:hard_sep_supply_demand}
Under the $T$-block hard-separation architecture of Definition~\ref{def:hard_separation}, equivalently \eqref{eq:block_markov_hard_sep}, any scheme achieving average distortion at most $D$ must satisfy
\begin{equation}
  R_{X|Y}(D)\ \le\ \min\big\{\, n\Cch,\ m_{\mathrm{dec}}\Cgate,\ m_{\mathrm{task}}\Cgate \big\}.
  \label{eq:hard_sep_supply_demand}
\end{equation}
In particular, under the optimal budget split $m_{\mathrm{dec}}=m_{\mathrm{task}}=\frac{m}{2}$,
\begin{equation}
  R_{X|Y}(D)\ \le\ \min\Big\{\, n\Cch,\ \frac{m}{2}\Cgate \Big\}.
  \label{eq:hard_sep_supply_demand_m_over_2}
\end{equation}
\end{corollary}

\begin{proof}
Let the overall $T$-block hard-separation scheme induce the joint law of $(X^T,Y^T,\hat X^T)$.
By Lemma~\ref{lem:ird_block_lower_bound}, achieving average distortion at most $D$ implies
$R_{X|Y}(D)\le \frac{1}{T}I(Y^T;\hat X^T)$.
The mutual information upper bound in \eqref{eq:hard_sep_supply_demand} then follows from Theorem~\ref{thm:interface_tax}.
\end{proof}

\begin{corollary}[Strict gap: task-direct vs.\ hard-separation]
\label{cor:strict_sep_loss}
Assume the achievability conditions of Theorem~\ref{thm:achievability_general}.
Suppose we are in the \emph{\regcomp} (cf.~the Interpretation following Theorem~\ref{thm:supply_demand_converse}), i.e.,
\begin{equation}
  \condcomp.
\end{equation}
Then the following hold.
\begin{enumerate}
  \item \textbf{(Task-direct achievability)} Any distortion level $D$ satisfying
  $R_{X|Y}(D) < m\Cgate$ is achievable by task-direct processing (Architecture~(A)).
  \item \textbf{(Hard-separation achievability under the stronger closure)} Under Assumption~\ref{assump:achievability_benchmark} and the symmetric budget split
  $m_{\mathrm{dec}}=m_{\mathrm{task}}=\frac{m}{2}$, any distortion level $D$ satisfying
  $R_{X|Y}(D) < \frac{m}{2}\Cgate$ is achievable by the two-stage construction of Theorem~\ref{thm:interface_tax_achievability}.
  \item \textbf{(Strict gap region)} For any distortion level $D$ such that
  \begin{equation}
    \frac{m}{2}\Cgate \;<\; R_{X|Y}(D) \;<\; m\Cgate,
    \label{eq:strict_sep_gap_condition}
  \end{equation}
  the task-direct construction can achieve distortion $D$, while every hard-separation scheme (Architecture~(B)) is ruled out by the converse in Corollary~\ref{cor:hard_sep_supply_demand}. Thus the compute-limited regime contains a strict interval of demand levels that is feasible for task-direct processing but impossible under hard-separation, even though the serial achievability result attains every demand level below the symmetric threshold $\frac{m}{2}\Cgate$.
\end{enumerate}
\end{corollary}

\begin{proof}
Item~1 follows by applying Theorem~\ref{thm:achievability_general} with
$R_{X|Y}(D)<m\Cgate\le n\Cch$.
For Item~2, apply Theorem~\ref{thm:interface_tax_achievability} with the symmetric split
$m_{\mathrm{dec}}=m_{\mathrm{task}}=\frac{m}{2}$; in the compute-limited regime,
$R_{X|Y}(D)<\frac{m}{2}\Cgate\le n\Cch$, so the achievability condition
\eqref{eq:interface_tax_achievability_condition} is satisfied.
For Item~3, task-direct achievability follows from Item~1, whereas Corollary~\ref{cor:hard_sep_supply_demand}
(in particular, \eqref{eq:hard_sep_supply_demand_m_over_2}) implies that any hard-separation scheme must satisfy
$R_{X|Y}(D)\le \frac{m}{2}\Cgate$ under the optimal split. Hence no hard-separation scheme can attain any
distortion level whose demand lies in \eqref{eq:strict_sep_gap_condition}.
\end{proof}

\section{Alternative Noisy-Logic Closure}
\label{sec:noisy_logic_closure}
\indent This section studies the separate noisy-logic closure in which \emph{all} receiver computation, including control and bookkeeping, is noisy, and records a conservative depth-dependent converse for that closure.

For context only, we draw on Evans--Schulman's information-propagation converse~\cite{EvansSchulman99}; we do not attempt to turn this section into a full reliable-computation theory in the sense of von Neumann-style threshold constructions~\cite{vonNeumann56}.

\subsection{Noisy-Gate Propagation and a Closure-Specific Compute Cut}
\label{subsec:noisy_gate_connections}
Our baseline model (Definition~\ref{def:gpu_bitflip_model}) treats control as protected and exposes the first-order compute supply $m\Cgate$ under a committed/no-bypass receiver interface. Here we replace that baseline term by a depth-dependent noisy-logic supply for the separate closure of this section.

\paragraph{A canonical noisy-gate model}
Consider a Boolean circuit in which each gate computes a deterministic Boolean function of at most $K_{\mathrm{fan}}$ inputs, and then its output is flipped independently with probability $\delta\in(0,\tfrac12)$ (a BSC$(\delta)$ gate-noise model).
Let $X_1,\dots,X_N$ denote input bits and let $Y$ denote a circuit output random variable determined by the noisy gate outputs.

\begin{theorem}[Evans--Schulman signal-propagation bound (one form)]
\label{thm:evans_schulman_bound}
Consider the $K_{\mathrm{fan}}$-fan-in noisy-gate model above: each gate computes a deterministic Boolean function of at most $K_{\mathrm{fan}}$ inputs and then its output is flipped independently with probability $\delta\in(0,\tfrac12)$.
Let $X_1,\dots,X_N$ denote input bits and let $Y$ denote \emph{any} circuit output random variable that is a (measurable) function of the noisy gate outputs (thus $Y$ may be multi-bit or nonbinary).
For each input $X_i$, let $d_i$ denote the length of the shortest directed path from $X_i$ to any gate output on which $Y$ depends.
Then Evans and Schulman~\cite{EvansSchulman99} show that information propagates at most geometrically with depth:
\begin{equation}
  I(X_i;Y)\ \le\ \min\!\Big\{1,\ \big(K_{\mathrm{fan}}(1-2\delta)^2\big)^{d_i}\Big\}\,H(X_i).
  \label{eq:evans_schulman_decay}
\end{equation}
In particular, if $K_{\mathrm{fan}}(1-2\delta)^2<1$, then $I(X_i;Y)$ decays exponentially in $d_i$.
\end{theorem}

Viewed through our framework, bound~\eqref{eq:evans_schulman_decay} provides a closure-specific effective compute budget: in high-noise or large-depth regimes, maintaining $O(1)$ bits of task-relevant information at the output may require either bounded depth or exponential fan-out/redundancy, so a budget of $m$ noisy gates need not translate into $m\Cgate$ usable bits. We therefore interpret $m\Cgate$ in this section as a best-case first-order surrogate, with the propagation bound supplying the depth-dependent correction.

\paragraph{From single-input decay to a compute cut}
We next translate \eqref{eq:evans_schulman_decay} into an upper bound on the \emph{total} mutual information that can flow from the circuit input to its output.
This yields a closure-specific depth-dependent compute cut that can be inserted into the same supply--demand template as Theorem~\ref{thm:supply_demand_converse}.

\begin{lemma}[Closure-specific depth-dependent compute-information bound for noisy-gate circuits]
\label{lem:noisy_gate_compute_cut}
Let $U^q=(U_1,\dots,U_q)$ be a binary vector obtained from the channel output $R^{nT}$ (possibly using receiver-side private randomness).
Suppose the receiver feeds $U^q$ as the \emph{only} external input to a $K_{\mathrm{fan}}$-fan-in noisy-gate circuit with gate flip probability $\delta$, where all gate flips are independent of $(R^{nT},U^q)$, and outputs $\hat X^T$ as a measurable function of the noisy gate outputs.
For each input bit $U_i$, let $d_i$ denote the length of the shortest directed path (in number of noisy gates) from $U_i$ to any gate output on which $\hat X^T$ depends.
Then
\begin{align}
  I(R^{nT};\hat X^T)
  &\le I(U^q;\hat X^T) \label{eq:noisy_gate_cut_dp}\\
  &= \sum_{i=1}^q I(U_i;\hat X^T \mid U^{i-1}) \label{eq:noisy_gate_cut_sum}\\
  &\le \sum_{i=1}^q \min\!\Big\{1,\ \big(K_{\mathrm{fan}}(1-2\delta)^2\big)^{d_i}\Big\} H(U_i \mid U^{i-1})\\
  &\le \sum_{i=1}^q \min\!\Big\{1,\ \big(K_{\mathrm{fan}}(1-2\delta)^2\big)^{d_i}\Big\} H(U_i).
  \label{eq:noisy_gate_cut_final}
\end{align}
In particular, if $H(U_i)\le 1$ and $d_i\ge d_{\mathrm{logic}}$ for all $i$, then
\begin{equation}
  I(R^{nT};\hat X^T) \le q\,\min\!\Big\{1,\ \big(K_{\mathrm{fan}}(1-2\delta)^2\big)^{d_{\mathrm{logic}}}\Big\}.
  \label{eq:noisy_gate_cut_simplified}
\end{equation}
If moreover $q\le mT$ (e.g., each input bit must be latched/materialized by a noisy input register counted in the $mT$-gate budget), then
\begin{equation}
  \frac{1}{T}I(R^{nT};\hat X^T)\ \le\ m\,\min\!\Big\{1,\ \big(K_{\mathrm{fan}}(1-2\delta)^2\big)^{d_{\mathrm{logic}}}\Big\}.
  \label{eq:noisy_gate_cut_per_sample}
\end{equation}
\end{lemma}

\begin{proof}
By construction, the receiver's noisy logic uses the channel output only through the interface bits $U^q$ (and independent gate-flip noise), so $R^{nT}\to U^q\to \hat X^T$ forms a Markov chain; the data-processing inequality yields \eqref{eq:noisy_gate_cut_dp}.
The chain rule gives \eqref{eq:noisy_gate_cut_sum}.
Fix $i\in\{1,\dots,q\}$ and condition on $U^{i-1}=u^{i-1}$.
With these upstream bits fixed, the induced conditional law still corresponds to a $K_{\mathrm{fan}}$-fan-in noisy-gate circuit in which the shortest directed path from $U_i$ to any gate output on which $\hat X^T$ depends has length $d_i$.
Applying Theorem~\ref{thm:evans_schulman_bound} under this conditional law gives
\begin{equation*}
  I(U_i;\hat X^T\mid U^{i-1}=u^{i-1})
  \le
  \min\!\Big\{1,\ \big(K_{\mathrm{fan}}(1-2\delta)^2\big)^{d_i}\Big\}\,H(U_i\mid U^{i-1}=u^{i-1}).
\end{equation*}
Averaging over $U^{i-1}$ yields the third line in the display, and the final line uses $H(U_i\mid U^{i-1})\le H(U_i)$.
The simplified bounds \eqref{eq:noisy_gate_cut_simplified} and \eqref{eq:noisy_gate_cut_per_sample} are immediate.
\end{proof}

\begin{theorem}[Alternative depth-dependent converse under the noisy-logic closure]
\label{thm:supply_demand_noisy_gate}
Consider the task-direct architecture \eqref{eq:task_direct_arch}, but suppose that the \emph{entire} receiver computation (including control logic) is implemented as a $K_{\mathrm{fan}}$-fan-in noisy-gate circuit with gate flip probability $\delta$.
Let $C_{\mathrm{gate}}(\delta)\triangleq 1-\hbin{\delta}$ denote the Shannon capacity (bits/use) of a $\mathrm{BSC}(\delta)$ noisy gate.
Assume that over a task block of length $T$, the receiver uses at most $mT$ noisy gates in total, and injects at most $q\le mT$ channel-dependent input bits into the noisy circuit, where each such bit must traverse at least $d_{\mathrm{logic}}$ noisy gates to influence the final reconstruction $\hat X^T$.
If a scheme achieves average distortion at most $D$, then
\begin{equation}
  R_{X|Y}(D)\ \le\ \min\Big\{\, n\Cch,\ m\,C_{\mathrm{gate}}(\delta),\ m\,\min\!\big\{1,\ \big(K_{\mathrm{fan}}(1-2\delta)^2\big)^{d_{\mathrm{logic}}}\big\}\Big\}.
  \label{eq:main_converse_noisy_gate}
\end{equation}
\end{theorem}

\begin{remark}[Closure scope and interpretation]
Theorem~\ref{thm:supply_demand_noisy_gate} belongs to the separate noisy-logic closure; if the receiver has a protected control plane or other reliable islands, then the baseline cut $m\Cgate$ in Theorem~\ref{thm:supply_demand_converse} (or the heterogeneous/bypass plug-ins in Lemmas~\ref{lem:compute_info_bound_bypass} and~\ref{lem:compute_info_bound_hetero}) is the more appropriate model.
\end{remark}

\begin{proof}
The communication cut is identical to Theorem~\ref{thm:supply_demand_converse}.
For the compute cut, there are two independent upper bounds on $I(R^{nT};\hat X^T)$.
First, order the $mT$ noisy gates topologically and apply the same chain-rule argument as in Lemma~\ref{lem:compute_info_bound} to the $\mathrm{BSC}(\delta)$ gates to obtain
$I(R^{nT};\hat X^T)\le mT\,C_{\mathrm{gate}}(\delta)$.
Second, Lemma~\ref{lem:noisy_gate_compute_cut} yields
$I(R^{nT};\hat X^T)\le mT\,\min\{1,(K_{\mathrm{fan}}(1-2\delta)^2)^{d_{\mathrm{logic}}}\}$ under the stated depth assumption.
Combining these bounds with Lemma~\ref{lem:ird_block_lower_bound} gives \eqref{eq:main_converse_noisy_gate}.
\end{proof}

\paragraph{Interpretive note}
The converse \eqref{eq:main_converse_noisy_gate} is conservative rather than matched: fault-tolerant constructions in the spirit of von Neumann~\cite{vonNeumann56} can, below suitable noise thresholds, induce an architecture-dependent effective reliable supply, but a corresponding achievability characterization for this noisy-logic closure remains open.

\paragraph{Closure-specific hard-separation plug-in}
The previous theorem is task-direct. Under the same noisy-logic closure, imposing a committed hard-separation interface yields the analogous stagewise replacement on the two serial cuts.

\begin{corollary}[Closure-specific hard-separation plug-in under the noisy-logic closure]
Consider the $T$-block hard-separation architecture \eqref{eq:block_markov_hard_sep}.
Suppose there exist binary interface vectors $U_{\mathrm{dec}}^{q_{\mathrm{dec}}}$ and $U_{\mathrm{task}}^{q_{\mathrm{task}}}$ such that
\[
  R^{nT} \to U_{\mathrm{dec}}^{q_{\mathrm{dec}}} \to \hat B^T,
  \qquad
  \hat B^T \to U_{\mathrm{task}}^{q_{\mathrm{task}}} \to \hat X^T
\]
under the decoder stage and the downstream task stage, respectively.
Suppose further that these two stages are each implemented by $K_{\mathrm{fan}}$-fan-in noisy-gate circuits with gate flip probability $\delta$.
Assume that over the block the decoder uses at most $m_{\mathrm{dec}}T$ noisy gates, injects at most $q_{\mathrm{dec}}\le m_{\mathrm{dec}}T$ binary input bits into the noisy circuit, and has depth at least $d_{\mathrm{dec}}$ from those injected bits to $\hat B^T$.
Assume likewise that the task stage uses at most $m_{\mathrm{task}}T$ noisy gates, injects at most $q_{\mathrm{task}}\le m_{\mathrm{task}}T$ binary input bits into the noisy circuit, and has depth at least $d_{\mathrm{task}}$ from those injected bits to $\hat X^T$.
Then any scheme satisfies
\begin{equation}
  \frac{1}{T}I(Y^T;\hat X^T)
  \ \le\
  \min\big\{\, n\Cch,\ S_{\mathrm{dec}}^{(\mathrm{logic})},\ S_{\mathrm{task}}^{(\mathrm{logic})}\big\},
  \label{eq:interface_tax_noisy_gate}
\end{equation}
where $C_{\mathrm{gate}}(\delta)$ is as in Theorem~\ref{thm:supply_demand_noisy_gate}, and
\begin{align}
  S_{\mathrm{dec}}^{(\mathrm{logic})}
  &\triangleq
  \min\Big\{\, m_{\mathrm{dec}}\,C_{\mathrm{gate}}(\delta),\ m_{\mathrm{dec}}\,\min\!\big\{1,\big(K_{\mathrm{fan}}(1-2\delta)^2\big)^{d_{\mathrm{dec}}}\big\}\Big\},\\
  S_{\mathrm{task}}^{(\mathrm{logic})}
  &\triangleq
  \min\Big\{\, m_{\mathrm{task}}\,C_{\mathrm{gate}}(\delta),\ m_{\mathrm{task}}\,\min\!\big\{1,\big(K_{\mathrm{fan}}(1-2\delta)^2\big)^{d_{\mathrm{task}}}\big\}\Big\}.
\end{align}
\end{corollary}

\begin{proof}
As in Theorem~\ref{thm:interface_tax}, data processing under \eqref{eq:block_markov_hard_sep} yields
$I(Y^T;\hat X^T)\le I(R^{nT};\hat B^T)$ and $I(Y^T;\hat X^T)\le I(\hat B^T;\hat X^T)$, and the communication cut gives $I(Y^T;\hat X^T)\le nT\Cch$.
For each stage, order its noisy gates topologically to obtain the per-gate capacity bound, exactly as in the proof of Theorem~\ref{thm:supply_demand_noisy_gate}.
For the decoder stage, apply Lemma~\ref{lem:noisy_gate_compute_cut} with stage input $U_{\mathrm{dec}}^{q_{\mathrm{dec}}}$, $q=q_{\mathrm{dec}}$, and $d_{\mathrm{logic}}=d_{\mathrm{dec}}$.
For the task stage, apply Lemma~\ref{lem:noisy_gate_compute_cut} with stage input $U_{\mathrm{task}}^{q_{\mathrm{task}}}$, $q=q_{\mathrm{task}}$, and $d_{\mathrm{logic}}=d_{\mathrm{task}}$.
Dividing by $T$ and taking the minimum of the resulting bounds yields \eqref{eq:interface_tax_noisy_gate}.
\end{proof}

\subsection{Making Depth $d_{\mathrm{logic}}$ a Countable Resource: Latency, Iterations, and Size--Depth Tradeoffs}
\label{subsec:depth_mapping}
The depth-dependent compute cut in Theorem~\ref{thm:supply_demand_noisy_gate} introduces an additional parameter beyond the aggregate gate budget $m$: the minimum noisy-gate path length $d_{\mathrm{logic}}$ from channel-dependent input bits to the final task output.
To make this parameter operational, we record a compact reformulation and several common interpretations that map $d_{\mathrm{logic}}$ to familiar algorithmic or architectural resources.

\paragraph{An effective per-gate information supply under noisy logic}
Define the \emph{effective} per-gate information supply
\begin{equation}
  C_{\mathrm{logic}}(d_{\mathrm{logic}},\delta,K_{\mathrm{fan}})
  \triangleq
  \min\Big\{\, C_{\mathrm{gate}}(\delta),\ \min\!\big\{1,\ \big(K_{\mathrm{fan}}(1-2\delta)^2\big)^{d_{\mathrm{logic}}}\big\}\Big\}.
  \label{eq:C_logic_def}
\end{equation}
\paragraph*{Remark} In this section we reserve $d_{\mathrm{logic}}$ for circuit depth (number of noisy-gate layers), while $d(\cdot,\cdot)$ continues to denote the task distortion function.

\paragraph*{How to read $C_{\mathrm{logic}}$} Let $\beta\triangleq K_{\mathrm{fan}}(1-2\delta)^2$.
\begin{itemize}
\item \textbf{$\beta>1$:} information propagation does not decay with depth (it saturates), so $C_{\mathrm{logic}}(d_{\mathrm{logic}},\delta,K_{\mathrm{fan}})=C_{\mathrm{gate}}(\delta)$ and the per-gate capacity cut is typically the active constraint.
\item \textbf{$\beta<1$:} information propagation decays geometrically as $\beta^{d_{\mathrm{logic}}}$, so $C_{\mathrm{logic}}(d_{\mathrm{logic}},\delta,K_{\mathrm{fan}})=\min\{C_{\mathrm{gate}}(\delta),\beta^{d_{\mathrm{logic}}}\}$ and circuit depth becomes a key resource.
\end{itemize}

Then the noisy-logic converse \eqref{eq:main_converse_noisy_gate} can be written more compactly as
\begin{equation}
  R_{X|Y}(D)\ \le\ \min\big\{\, n\Cch,\ m\,C_{\mathrm{logic}}(d_{\mathrm{logic}},\delta,K_{\mathrm{fan}})\big\}.
  \label{eq:main_converse_noisy_gate_compact}
\end{equation}
In words: under noisy logic, the $m$-gate budget behaves as if each noisy gate contributed at most $C_{\mathrm{logic}}$ usable task bits (rather than $C_{\mathrm{gate}}(\delta)$).

\begin{corollary}[Closure-specific size--depth requirement under the noisy-logic closure]
\label{cor:size_depth_requirement}
Fix $r>0$ and consider any task-direct design under the noisy-gate model of Theorem~\ref{thm:supply_demand_noisy_gate}.
If achieving distortion $D$ requires (or is certified by) a task-information demand of at least $r$ bits/sample, i.e., if $R_{X|Y}(D)\ge r$, then necessarily
\begin{equation}
  m\ \ge\ \frac{r}{C_{\mathrm{logic}}(d_{\mathrm{logic}},\delta,K_{\mathrm{fan}})}.
  \label{eq:size_depth_requirement}
\end{equation}
In the information-propagation regime $\beta=K_{\mathrm{fan}}(1-2\delta)^2<1$ where the propagation term dominates,
$C_{\mathrm{logic}}(d_{\mathrm{logic}},\delta,K_{\mathrm{fan}})=\min\{C_{\mathrm{gate}}(\delta),\,\beta^{d_{\mathrm{logic}}}\}$ and sustaining a constant $r$ as $d_{\mathrm{logic}}$ grows requires
\begin{equation}
  m\ =\ \Omega(\beta^{-d_{\mathrm{logic}}}),
\end{equation}
i.e., an \emph{exponential} growth of the noisy-gate budget with depth.
\end{corollary}

\paragraph{Latency and synchronous pipelines}
In a synchronous receiver pipeline, the computation admits a directed-acyclic-graph (DAG) interpretation unrolled over clock cycles.
If a channel-dependent bit is latched at cycle $0$ and must influence the final output only after $J$ sequential cycles of noisy combinational logic and registers, then its path length in noisy primitives is at least proportional to $J$.
Thus, in many implementations, the depth parameter $d_{\mathrm{logic}}$ in Theorem~\ref{thm:supply_demand_noisy_gate} can be interpreted as (or lower bounded by) a \emph{latency} measure.
This makes the depth-dependent term in \eqref{eq:main_converse_noisy_gate} particularly relevant for low-latency receivers: even when the total gate budget $m$ is large, deep pipelines can suffer a sharp decay in usable task information when $\beta<1$.

\paragraph{Iterative algorithms and unrolled computation graphs}
Many decoding and inference algorithms are iterative (e.g., belief propagation, turbo/LDPC decoding, iterative estimation, and unrolled neural networks).
Unrolling a $J$-iteration algorithm into a feedforward computation graph typically produces a depth proportional to $J$ from the initial observations to the final decision.
Accordingly, when the logic itself is noisy, Theorem~\ref{thm:supply_demand_noisy_gate} and Corollary~\ref{cor:size_depth_requirement} predict a strong tradeoff:
in the propagation-limited regime $\beta<1$, increasing the iteration count $J$ can require exponentially increasing redundancy (or reliable islands) to avoid an exponential decay of task-relevant mutual information.

\paragraph{Architectural implication}
The depth-dependent cut helps interpret why \emph{skip connections}, \emph{side-information interfaces}, and \emph{hybrid designs} (which provide later modules access to earlier representations) can be valuable on unreliable substrates:
such designs create shorter effective paths for critical task-relevant information, reducing the effective depth $d_{\mathrm{logic}}$ and mitigating propagation-induced losses (cf.~Proposition~\ref{prop:no_tax_soft_interface}).

\section{Examples}
\label{sec:examples}

\indent We close the main text with one discrete example and one scalar-Gaussian benchmark. The binary example shows that the committed-interface penalty is not a Gaussian artifact, while the scalar-Gaussian benchmark turns the same converse structure into closed-form distortion formulas. Fuller Fano-style classification bounds, additional Gaussian variants, noisy-logic/word-level plug-ins, uncoded analog baselines, and vector extensions are deferred to Appendices~\ref{app:section3_supp} and~\ref{app:gaussian_supp}.

\subsection{Binary Discrete Example}
\label{subsec:binary_example}
Consider a uniformly distributed $q$-bit label $J=(J_1,\dots,J_q)\in\{0,1\}^q$ to be recovered from an observation $Y$, and let $\hat J$ denote the receiver decision. Proposition~\ref{prop:fano_classification} in Appendix~\ref{app:section3_supp} gives the formal Fano-style statement; here we retain only the resulting supply implications:
\begin{equation}
  P_e \ \ge\ \max\!\left\{0,\ 1-\frac{\Rsup+1}{q}\right\}
  \label{eq:binary_example_taskdirect}
\end{equation}
under task-direct processing, while under hard-separation,
\begin{equation}
  P_e \ \ge\ \max\!\left\{0,\ 1-\frac{\min\{n\Cch,\ \frac{m}{2}\Cgate\}+1}{q}\right\}.
  \label{eq:binary_example_hardsep}
\end{equation}
In the compute-limited regime, the same target reliability therefore requires roughly twice as many vulnerable primitive uses once a committed intermediate interface is imposed. This discrete witness shows that the committed-interface penalty is not peculiar to Gaussian rate--distortion formulas.

\subsection{Scalar Gaussian Benchmark}
\label{subsec:scalar_gaussian_example}
We specialize to the scalar linear--Gaussian model
\begin{equation}
  X\sim \mathcal{N}(0,\sigma_x^2),\qquad Y=X+V,\quad V\sim\mathcal{N}(0,\sigma_v^2),
\end{equation}
i.i.d.\ across task instances (indexed by $t$), with quadratic distortion $d(x,\hat x)=(x-\hat x)^2$.
Define the conditional variance
\begin{equation}
  \sigma_{x|y}^2 \triangleq \mathrm{Var}(X\mid Y)=\frac{\sigma_x^2\sigma_v^2}{\sigma_x^2+\sigma_v^2}.
\end{equation}

The indirect rate--distortion function for this model is classical~\cite{BergerBook71,Witsenhausen80}.

\begin{corollary}[Closed-form distortion bound for scalar Gaussian]
\label{cor:gaussian_closed_form}
Under the scalar Gaussian model, any task-direct scheme (Architecture~(A)) must satisfy
\begin{equation}
  \boxed{
  D \ \ge\ \sigma_{x|y}^2 + \big(\sigma_x^2-\sigma_{x|y}^2\big)\,2^{-2\Rsup}.
  }
  \label{eq:gaussian_distortion_bound}
\end{equation}
where $\Rsup$ is defined in \eqref{eq:Rsupply_def}.
\end{corollary}

\begin{proof}
For $\sigma_{x|y}^2 < D < \sigma_x^2$, the indirect rate--distortion function equals
$R_{X|Y}(D)=\frac12\log_2\!\left(\frac{\sigma_x^2-\sigma_{x|y}^2}{D-\sigma_{x|y}^2}\right)$~\cite{Witsenhausen80}.
Combining this expression with Theorem~\ref{thm:supply_demand_converse} and solving for $D$ yields \eqref{eq:gaussian_distortion_bound}.
For $D\le \sigma_{x|y}^2$ the target is impossible, and for $D\ge \sigma_x^2$ it is trivially achievable.
\end{proof}

\paragraph{Committed-interface specialization}
Combining Theorem~\ref{thm:interface_tax} with the Gaussian indirect rate--distortion demand yields an explicit closed-form hard-separation converse.

\begin{corollary}[Gaussian distortion bound under hard-separation]
\label{cor:gaussian_sep}
For scalar Gaussian remote estimation under hard-separation,
\begin{equation}
  \boxed{
  D \ \ge\ \sigma_{x|y}^2 + \big(\sigma_x^2-\sigma_{x|y}^2\big)\,
  2^{-2\min\{n\Cch,\ \frac{m}{2}\Cgate\}}.
  }
  \label{eq:gaussian_sep_bound}
\end{equation}
\end{corollary}

\begin{proof}
Use the scalar Gaussian indirect rate--distortion formula from the proof of Corollary~\ref{cor:gaussian_closed_form} with Corollary~\ref{cor:hard_sep_supply_demand} (equivalently, Theorem~\ref{thm:interface_tax}), which gives $R_{X|Y}(D)\le \min\{n\Cch,\frac{m}{2}\Cgate\}$ under the optimal split. Solving for $D$ yields \eqref{eq:gaussian_sep_bound}.
\end{proof}

\begin{remark}[Square-root interpretation]
In the \emph{compute-limited} regime $m\Cgate \le n\Cch$, Corollary~\ref{cor:gaussian_closed_form} gives a normalized excess distortion above the MMSE floor scaling as $2^{-2m\Cgate}$ under task-direct processing, whereas Corollary~\ref{cor:gaussian_sep} yields $2^{-m\Cgate}$ under the symmetric hard-separation split, i.e., a square-root loss. Under Assumption~\ref{assump:achievability_benchmark}, the task-direct and hard-separation scalar-Gaussian curves are direct specializations of Theorems~\ref{thm:achievability_general} and~\ref{thm:interface_tax_achievability}, respectively; Appendix~\ref{app:gaussian_supp} records the strict-gap witness and higher-dimensional variants.
\end{remark}

\section{Conclusion}
\label{sec:conclusion}
We developed an information-theoretic framework for remote inference when communication is noisy and receiver-side computation is implemented through vulnerable primitives under a finite per-instance compute budget. Under the committed/no-bypass noisy-materialization closure, the first-order converse picture is governed by a supply--demand comparison between the task information required by indirect rate--distortion and the information jointly supplied by the physical channel and the vulnerable compute substrate. This yields the baseline single-bottleneck converse (Theorem~\ref{thm:supply_demand_converse}) for task-direct processing and, more generally, a receiver-internal compute min-cut converse (Theorem~\ref{thm:dag_cutset}) showing that committed intermediate interfaces can create additional first-order bottlenecks inside the receiver. In particular, the hard-separation penalty is closure-dependent rather than universal: it arises when task-relevant information must cross serial vulnerable interfaces, while downstream soft visibility to the raw channel output within that closure can reduce the converse to the single-bottleneck supply, and protected bypasses or reliable islands can mitigate or eliminate it.

For the separate fully noisy-logic closure, we obtained only a conservative depth-dependent converse. Developing a matched achievability characterization for that closure, extending the theory to correlated or architecture-dependent fault models, and tightening the finite-blocklength and implementation-level connections for low-latency remote inference remain natural directions for future work. More broadly, once receiver-internal materialization is vulnerable, architecture itself becomes a first-order information-theoretic object.

\appendices

\section{Supplementary Finite-Blocklength and Tail Benchmarks}
\label{sec:finite_blocklength}

\indent We refine the asymptotic converses with finite-blocklength (dispersion-aware) and excess-distortion benchmarks tailored to low-latency design. The first-order limits of Sections~\ref{sec:converse} and~\ref{sec:interface_tax_section} compare a task-information \emph{demand} against a communication--computation \emph{supply}. They are sharp in the large-blocklength regime in which coding across many task instances is allowed. In low-latency regimes (small $T$), however, one needs nonasymptotic guarantees on \emph{tail} performance metrics such as the excess-distortion probability~\cite{KostinaVerdu13_JSCC}. This appendix therefore formalizes the block excess-distortion criteria, states a finite-$T$ concatenation benchmark, and records normal-approximation and reliability-function refinements. When specifying a finite-$T$ $(D,\epsilon)$ operating point, we treat $\epsilon$ as a total tail-reliability budget split across the active source, channel, and compute-side failure events of the chosen architecture (cf.~\eqref{eq:epsilon_budget_union} and \eqref{eq:normal_interface_tax}); examples use symmetric splits unless stated otherwise.

\subsection{Tail Reliability Metric: Excess Distortion Probability}
\label{subsec:excess_distortion_metric}
Average distortion captures mean task performance but can hide rare catastrophic failures.
Motivated by the finite-blocklength lossy JSCC work (e.g.,~\cite{KostinaVerdu13_JSCC}), we track \emph{tail reliability} via excess-distortion probabilities.
We use both a block-averaged criterion,
\begin{equation}
  \epsilon_D^{\mathrm{blk}}(D)
  \triangleq
  \mathbb{P}\!\left(\frac{1}{T}\sum_{t=1}^T d(X_t,\hat X_t) > D\right),
  \label{eq:excess_dist_prob_def}
\end{equation}
and the corresponding per-instance tail probability,
\begin{equation}
  \epsilon_D^{(1)}(D)
  \triangleq
  \mathbb{P}\!\left(d(X_t,\hat X_t)>D\right).
  \label{eq:excess_dist_prob_one}
\end{equation}
The block criterion controls the empirical-average task loss over a block of length $T$, while \eqref{eq:excess_dist_prob_one} emphasizes single-step safety.

Since $\frac{1}{T}\sum_{t=1}^T d(X_t,\hat X_t)>D$ implies $\max_{1\le t\le T} d(X_t,\hat X_t)>D$, we always have the union bound
\begin{equation}
  \epsilon_D^{\mathrm{blk}}(D)
  \le
  \mathbb{P}\!\left(\max_{1\le t\le T} d(X_t,\hat X_t) > D\right)
  \le
  \sum_{t=1}^T \mathbb{P}\!\left(d(X_t,\hat X_t) > D\right)
  =
  T\,\epsilon_D^{(1)}(D),
  \label{eq:block_to_single_tail_bound}
\end{equation}
where the last equality uses stationarity (and does not require independence).
In the i.i.d.\ memoryless setting the two criteria are closely related (and can be connected via concentration),
but we keep \eqref{eq:excess_dist_prob_def} as our main block-level tail metric.

Under noisy decoding/computation, undetected corruption can disproportionately affect $\epsilon_D^{\mathrm{blk}}(D)$ and $\epsilon_D^{(1)}(D)$, while detection and erasure-aware fallbacks can reduce them (Appendix~\ref{sec:erasures}).

\begin{definition}[$(D,\epsilon)$-achievability at blocklength $T$]
Fix a task-blocklength $T$ (number of task instances) and a receiver architecture (task-direct or hard-separation) under per-instance budgets $(n,m)$.
A pair $(D,\epsilon)$ is said to be \emph{$T$-achievable} if there exists a scheme such that the block excess-distortion probability satisfies
$\epsilon_D^{\mathrm{blk}}(D)\le \epsilon$ in~\eqref{eq:excess_dist_prob_def}.
When emphasizing single-step safety, we say $(D,\epsilon)$ is \emph{one-step achievable} if $\epsilon_D^{(1)}(D)\le \epsilon$ in~\eqref{eq:excess_dist_prob_one}.
\end{definition}

\begin{discussion}{Average vs.\ tail metrics}
For any nonnegative random variable $Z$, $\mathbb{E}[Z]=\int_0^\infty \mathbb{P}(Z>u)\,du$.
Therefore, scanning the threshold $D$ in $\epsilon_D^{(1)}(D)$ (or $\epsilon_D^{\mathrm{blk}}(D)$) reveals the full loss distribution, while the mean distortion reports only one number.

If the distortion is bounded, say $0\le d(X,\hat X)\le d_{\max}$ almost surely, then a block excess-distortion constraint implies a mean-distortion bound:
\begin{equation}
  \frac{1}{T}\sum_{t=1}^T \mathbb{E}\big[d(X_t,\hat X_t)\big]
  \le D + (d_{\max}-D)\,\epsilon_D^{\mathrm{blk}}(D)
  \le D + d_{\max}\,\epsilon_D^{\mathrm{blk}}(D).
  \label{eq:tail_to_mean_bounded}
\end{equation}
For unbounded distortions (e.g., squared error), tail and mean can decouple dramatically, and rare undetected errors can dominate mean and tail; Appendix~\ref{sec:erasures} discusses erasures and clipping-based ``damage limiters'' that restore quantitative control.

In finite-blocklength regimes, the mapping between tail and mean can also be loose because codes optimized for one distortion threshold need not be optimal for another~\cite{KostinaVerdu13_JSCC}.
Accordingly, we keep average distortion as our main first-order metric while using $(D,\epsilon)$ tail operating points to articulate reliability/safety tradeoffs and low-latency behavior.
\end{discussion}

\subsection{Design-Oriented Benchmark: Normal Approximations for Information Supply}
Let $N_{\mathrm{ch}}\triangleq nT$ and $L\triangleq mT$ denote, respectively, the total numbers of channel uses and compute-primitive uses in a task block of length $T$.
Let $M_{\mathrm{ch}}^\star(N_{\mathrm{ch}},\epsilon_{\mathrm{ch}})$ and $M_{\mathrm{comp}}^\star(L,\epsilon_{\mathrm{comp}})$ denote, respectively, the maximum numbers of messages achievable over $N_{\mathrm{ch}}$ channel uses and $L$ uses of the compute primitive with block error probabilities at most $\epsilon_{\mathrm{ch}}$ and $\epsilon_{\mathrm{comp}}$.
For memoryless channels, both admit the \emph{normal approximations}~\cite{Polyanskiy10}:
\begin{align}
  \log_2 M_{\mathrm{ch}}^\star(N_{\mathrm{ch}},\epsilon_{\mathrm{ch}})
  &\approx N_{\mathrm{ch}}\,\Cch
  - \sqrt{N_{\mathrm{ch}}\,V_{\mathrm{ch}}}\;Q^{-1}(\epsilon_{\mathrm{ch}}),
  \label{eq:normal_supply_channel}\\
  \log_2 M_{\mathrm{comp}}^\star(L,\epsilon_{\mathrm{comp}})
  &\approx L\,\Cgate
  - \sqrt{L\,V_{\mathrm{gate}}}\;Q^{-1}(\epsilon_{\mathrm{comp}}),
  \label{eq:normal_supply_compute}
\end{align}
where $Q(\cdot)$ is the Gaussian tail function, $Q^{-1}(\cdot)$ its inverse, and $V_{\mathrm{ch}}$ and $V_{\mathrm{gate}}$ denote the per-use channel and compute dispersions:
\begin{align}
  V_{\mathrm{ch}} &\triangleq \mathrm{Var}\big[\imath(S;R)\big]\quad \text{under a capacity-achieving input,}\\
  V_{\mathrm{gate}} &\triangleq \mathrm{Var}\big[\imath(U;Z)\big]\quad \text{under a capacity-achieving input.}
\end{align}
Here $\imath(S;R)\triangleq \log_2\frac{P_{R|S}(R|S)}{P_R(R)}$ denotes the information density under a capacity-achieving input distribution; similarly for $\imath(U;Z)$. The approximations \eqref{eq:normal_supply_channel}--\eqref{eq:normal_supply_compute} are accurate up to $O(\log T)$ terms for a broad class of memoryless channels~\cite{Polyanskiy10}.

Two common cases admit closed forms. For a BSC$(\varepsilon)$, a capacity-achieving input is uniform and the dispersion is
\begin{equation}
  V_{\mathrm{BSC}}(\varepsilon)
  = \varepsilon(1-\varepsilon)\Big(\log_2\frac{1-\varepsilon}{\varepsilon}\Big)^{\!2}\quad \text{bits$^2$/use}.
  \label{eq:bsc_dispersion}
\end{equation}
For a real AWGN channel with signal-to-noise ratio $\rho$, $\Cch=\frac12\log_2(1+\rho)$ and the dispersion is
\begin{equation}
  V_{\mathrm{AWGN}}(\rho)
  = \frac{\rho(\rho+2)}{2(\rho+1)^2}\,(\log_2 e)^2\quad \text{bits$^2$/use}.
  \label{eq:awgn_dispersion}
\end{equation}
See~\cite{Polyanskiy10} for details and extensions.

\begin{remark}[Word-level MCUs: effective capacity and dispersion plug-in]
\label{rem:mcu_dispersion_plugin}
For word-level additive upsets, replace the per-bit compute capacity $\Cgate$ in \eqref{eq:normal_supply_compute} by the effective value $\Cgateeff=1-\frac{H(E)}{w}$ and replace the compute dispersion $V_{\mathrm{gate}}$ by the corresponding per-bit value
$V_{\mathrm{gate}}^{\mathrm{eff}} \triangleq \mathrm{Var}\!\big[\log_2 P_E(E)\big]/w$.
Because the compute budget counts \emph{bits}, these substitutions yield the same normal-approximation form with $L=mT$:
\begin{equation}
  \log_2 M_{\mathrm{comp}}^\star(L,\epsilon_{\mathrm{comp}})
  \approx
  L\,\Cgateeff
  - \sqrt{L\,V_{\mathrm{gate}}^{\mathrm{eff}}}\;Q^{-1}(\epsilon_{\mathrm{comp}}).
  \label{eq:normal_supply_compute_mcu}
\end{equation}
\end{remark}

The compute substrate model in Definition~\ref{def:gpu_bitflip_model} is interactive: the algorithm may choose each stored bit $U_i$ as a function of past retrieved bits $Z^{i-1}$, much like channel coding with feedback. For memoryless channels, this leaves the first-order compute bottleneck $m\Cgate$ unchanged, but feedback can improve second-order terms; accordingly, the dispersion formulas in \eqref{eq:normal_supply_compute} (and its multi-stage variants) provide engineering benchmarks rather than sharp converses for all interactive strategies.

\begin{discussion}{Finite-$T$ supply benchmarks under noisy receiver logic}
The depth-dependent propagation cut in Theorem~\ref{thm:supply_demand_noisy_gate} is \emph{nonasymptotic}: it holds for every blocklength $T$ and does not rely on coding arguments.
Accordingly, a conservative dispersion-aware benchmark under noisy receiver logic is obtained by taking the minimum of three cuts:
\begin{equation}
  R_{\mathrm{sup,NA}}^{\mathrm{logic}}(T)
  \triangleq
  \min\Big\{
    \frac{1}{T}\log_2 M_{\mathrm{ch}}^\star(N_{\mathrm{ch}},\epsilon_{\mathrm{ch}}),\;
    \frac{1}{T}\log_2 M_{\mathrm{comp}}^\star(L,\epsilon_{\mathrm{comp}}),\;
    m\,\min\!\big\{1,\ \beta^{d_{\mathrm{logic}}}\big\}
  \Big\},
  \qquad
  \beta\triangleq K_{\mathrm{fan}}(1-2\delta)^2,
  \label{eq:normal_supply_noisy_logic}
\end{equation}
where the first two terms admit the normal approximations \eqref{eq:normal_supply_channel}--\eqref{eq:normal_supply_compute} and the third term is the Evans--Schulman propagation bottleneck.
In other words, finite-$T$ backoff can be incorporated for the communication cut and for the per-gate capacity cut, while the propagation cut acts as a separate deterministic cap that persists even as $T$ grows.
\end{discussion}

\subsection{Design-Oriented Benchmark: Reliability-Function Refinements (Error Exponents)}
\label{subsec:error_exponents}
The normal approximations in \eqref{eq:normal_supply_channel}--\eqref{eq:normal_supply_compute} capture the $\sqrt{T}$ backoff of achievable rates
at a fixed target block error probability.
In ultra-reliable regimes (very small $\epsilon$), it is often more convenient to work with \emph{error exponents} (reliability functions),
which quantify how fast the block error probability decays exponentially with blocklength when operating below capacity.

\paragraph{Error exponents for a generic memoryless channel}
Consider a DMC $W$ with capacity $C$.
For any code rate $\tilde R<C$ in bits/use, there exist sequences of blocklength-$N$ codes whose block error probability obeys
\begin{equation}
  P_e(N,\tilde R) \ \lesssim\ 2^{-N\,E_r(\tilde R)},
  \label{eq:error_exponent_generic}
\end{equation}
where $E_r(\tilde R)$ is the random-coding error exponent.
One convenient expression is Gallager's bound~\cite{Gallager68,CoverThomas}:
\begin{equation}
  E_r(\tilde R)
  \;=\;
  \max_{0\le \rho\le 1}\;
  \max_{P_X}\;
  \Big(
    E_0(\rho,P_X) - \rho\,\tilde R
  \Big),
  \label{eq:gallager_random_coding_exponent}
\end{equation}
with
\begin{equation}
  E_0(\rho,P_X)
  \triangleq
  -\log_2\!\sum_{z}
  \Big(\sum_x P_X(x)\,W(z|x)^{\frac{1}{1+\rho}}\Big)^{\!1+\rho}.
  \label{eq:gallager_E0}
\end{equation}
For symmetric channels, the maximizing input distribution in \eqref{eq:gallager_random_coding_exponent} is uniform.
Conversely, for any $\tilde R>C$, the strong converse for DMCs implies that the optimal block error probability tends to $1$ as $N\to\infty$.

\paragraph{Implication for the dual-noise budgets}
In our setting, a per-sample rate $R$ (bits/sample) corresponds to code rates $\tilde R_{\mathrm{ch}}=R/n$ (bits/channel-use) on the physical channel
and $\tilde R_{\mathrm{gate}}=R/m$ (bits/primitive-use) on the compute primitive.
Thus, operating at any 
$
  R < \Rsup
$
not only is first-order feasible but also admits exponentially decaying stage error probabilities:
\begin{align}
  P_{e,\mathrm{ch}}
  &\lesssim 2^{-nT\,E_{r,\mathrm{ch}}(R/n)},\\
  P_{e,\mathrm{comp}}
  &\lesssim 2^{-mT\,E_{r,\mathrm{gate}}(R/m)},
\end{align}
for suitable channel- and primitive-specific exponents $E_{r,\mathrm{ch}}(\cdot)$ and $E_{r,\mathrm{gate}}(\cdot)$.
These exponent forms provide an alternative (and sometimes sharper, in the very-small-$\epsilon$ regime) refinement of the supply benchmark than the normal approximation.

\paragraph{Closed form for the BSC compute primitive}
For a BSC$(\varepsilon)$ compute primitive with uniform input, Gallager's $E_0$ function admits the closed form
\begin{equation}
  E_{0,\mathrm{BSC}}(\rho)
  =
  \rho
  -
  (1+\rho)\log_2\!\Big((1-\varepsilon)^{\frac{1}{1+\rho}}+\varepsilon^{\frac{1}{1+\rho}}\Big),
  \qquad 0\le \rho\le 1,
  \label{eq:E0_bsc}
\end{equation}
and the random-coding exponent is
$E_{r,\mathrm{BSC}}(\tilde R)=\max_{0\le \rho\le 1}\big(E_{0,\mathrm{BSC}}(\rho)-\rho \tilde R\big)$.
Consequently, storing (or transporting) $TR$ task bits across $mT$ BSC$(\varepsilon)$ materializations at per-sample rate $R$
admits the reliability benchmark
\begin{equation}
  P_{e,\mathrm{comp}}
  \ \lesssim\
  2^{-mT\,E_{r,\mathrm{BSC}}(R/m)}.
  \label{eq:comp_error_exponent_bsc}
\end{equation}

Dispersion and error-exponent refinements are complementary: the normal approximation is most useful for moderate error probabilities, whereas exponent bounds are more transparent in the ultra-reliable regime. Both may be fed into the throughput mapping in Appendix~\ref{sec:throughput} by replacing the first-order supplies $n\Cch$ and $m\Cgate$ with nonasymptotic achievable-message bounds.

\subsection{Design-Oriented Benchmark: Normal Approximation for Task Demand}
In the remote task setting, the indirect rate--distortion demand can be written as an \emph{ordinary} rate--distortion function of the observable process $Y$ under an induced (indirect) distortion measure.

\begin{lemma}[Indirect rate--distortion as a direct RD problem with induced distortion]
\label{lem:induced_distortion_equivalence}
Define the induced single-letter distortion
\begin{equation}
  \bar d(y,\hat x)\triangleq \mathbb{E}\!\big[d(X,\hat x)\mid Y=y\big].
  \label{eq:induced_distortion}
\end{equation}
Then the indirect (remote) rate--distortion function in~\eqref{eq:indirect_rd_def} admits the equivalent representation
\begin{equation}
  R_{X|Y}(D)
  =
  \inf_{P_{\hat X|Y}:\ \mathbb{E}[\bar d(Y,\hat X)]\le D}\ I(Y;\hat X),
  \label{eq:indirect_rd_induced}
\end{equation}
i.e., it is the (direct) rate--distortion function of source $Y$ with reproduction alphabet $\hat{\mathcal X}$ and distortion measure $\bar d$.
\end{lemma}

\begin{proof}
For any test channel $P_{\hat X|Y}$, by iterated expectation,
\[
  \mathbb{E}[d(X,\hat X)]
  =
  \mathbb{E}\big[\mathbb{E}[d(X,\hat X)\mid Y,\hat X]\big]
  =
  \mathbb{E}\big[\mathbb{E}[d(X,\hat X)\mid Y]\big]
  =
  \mathbb{E}[\bar d(Y,\hat X)],
\]
where the second equality uses that $\hat X$ is conditionally independent of $X$ given $Y$ under $P_{X,Y}P_{\hat X|Y}$.
Substituting this identity into~\eqref{eq:indirect_rd_def} yields~\eqref{eq:indirect_rd_induced}; see also classical treatments of indirect rate--distortion~\cite{BergerBook71,Witsenhausen80}.
\end{proof}

We now define a finite-blocklength demand quantity that aligns with the excess-distortion criterion \eqref{eq:excess_dist_prob_def}.
Let $R_{X|Y}^{(T)}(D,\epsilon)$ denote the minimum description rate (bits/sample) of a fixed-length encoder--decoder mapping $Y^T\mapsto \hat X^T$ such that
\begin{equation}
  \mathbb{P}\!\left(\frac{1}{T}\sum_{t=1}^T d(X_t,\hat X_t) > D\right)\le \epsilon.
  \label{eq:finiteT_remote_rate_def}
\end{equation}
By Lemma~\ref{lem:induced_distortion_equivalence}, this is precisely the fixed-length lossy compression problem for the memoryless source $Y^T$ under the separable distortion measure $\bar d$ in \eqref{eq:induced_distortion}.
Therefore, under the regularity assumptions in the finite-blocklength lossy-compression theory of Kostina and Verd{\'u}~\cite{KostinaVerdu12_LossyCompression},
$R_{X|Y}^{(T)}(D,\epsilon)$ admits the normal approximation
\begin{equation}
  R_{X|Y}^{(T)}(D,\epsilon)
  =
  R_{X|Y}(D)
  +
  \sqrt{\frac{V_{X|Y}(D)}{T}}\;Q^{-1}(\epsilon)
  +
  O\!\left(\frac{\log T}{T}\right),
  \label{eq:remote_normal_approx}
\end{equation}
where $V_{X|Y}(D)$ is the (induced-distortion) rate--dispersion function.

Concretely, if $P_{\hat X|Y}^\star$ attains the infimum in \eqref{eq:indirect_rd_induced} at distortion level $D$, let $P_{\hat X}^\star$ be the induced marginal of $\hat X$, and let $s^\star\ge 0$ be an associated Lagrange multiplier for the distortion constraint in \eqref{eq:indirect_rd_induced} (often denoted $\lambda^\star$ in rate--distortion theory; we reserve $\lambda$ for throughput). The induced-distortion $d$-tilted information is~\cite{KostinaVerdu12_LossyCompression,KostinaVerdu13_JSCC}
\begin{equation}
  \jmath_{X|Y}(y,D)
  \triangleq
  -\log_2 \mathbb{E}_{\hat X\sim P_{\hat X}^\star}\!\Big[\exp\big(s^\star(D-\bar d(y,\hat X))\big)\Big].
  \label{eq:remote_d_tilted}
\end{equation}
Here $\log_2$ is paired with the natural exponential; equivalently, one may replace $s^\star$ by $\tilde s^\star\triangleq s^\star/\ln 2$ and write the expectation with base-2 exponentials.

The corresponding task dispersion is
\begin{equation}
  V_{X|Y}(D)\triangleq \mathrm{Var}\big(\jmath_{X|Y}(Y,D)\big).
  \label{eq:task_dispersion_def}
\end{equation}
We use $V_{X|Y}(D)$ below as the concise dispersion summary of the task demand in \eqref{eq:remote_normal_approx}.

\subsection{A Dispersion-Aware Supply--Demand Benchmark}
We first state a clean finite-$T$ benchmark that composes (i) finite-blocklength remote lossy compression,
(ii) finite-blocklength channel coding, and (iii) finite-blocklength ``compute coding'' across the $L=mT$ noisy primitives.
We then specialize this benchmark via the normal approximations in \eqref{eq:normal_supply_channel}--\eqref{eq:remote_normal_approx}.

\begin{proposition}[A finite-$T$ benchmark via concatenation of finite-blocklength codes]
\label{prop:finiteT_concat_benchmark}
Consider Architecture~(A) over a block of $T$ task instances.
Fix a distortion threshold $D$ and error budgets $(\epsilon_{\mathrm{src}},\epsilon_{\mathrm{ch}},\epsilon_{\mathrm{comp}})$ such that
\begin{equation}
  \epsilon_{\mathrm{src}}+\epsilon_{\mathrm{ch}}+\epsilon_{\mathrm{comp}}\le \epsilon.
  \label{eq:epsilon_budget_union}
\end{equation}
If there exists a fixed-length encoder--decoder mapping $Y^T\mapsto \hat X^T$ whose block excess-distortion probability satisfies
\[
  \mathbb{P}\!\left(\frac{1}{T}\sum_{t=1}^T d(X_t,\hat X_t)>D\right)\le \epsilon_{\mathrm{src}},
\]
and whose description alphabet size $M$ satisfies
\[
  M \le M_{\mathrm{ch}}^\star(N_{\mathrm{ch}},\epsilon_{\mathrm{ch}})
  \qquad\text{and}\qquad
  M \le M_{\mathrm{comp}}^\star(L,\epsilon_{\mathrm{comp}}),
\]
then there exists a task-direct scheme whose block excess-distortion probability is at most $\epsilon$.
Equivalently, using the finite-blocklength remote rate in \eqref{eq:finiteT_remote_rate_def}, a sufficient condition is
\begin{equation}
  R_{X|Y}^{(T)}(D,\epsilon_{\mathrm{src}})
  \ \le\
  \frac{1}{T}\min\Big\{\log_2 M_{\mathrm{ch}}^\star(N_{\mathrm{ch}},\epsilon_{\mathrm{ch}}),\ \log_2 M_{\mathrm{comp}}^\star(L,\epsilon_{\mathrm{comp}})\Big\}.
  \label{eq:finiteT_supply_demand_exact}
\end{equation}
\end{proposition}

\begin{proof}
Let $W\in\{1,\dots,M\}$ denote the (fixed-length) remote lossy description index produced from $Y^T$, and let $\hat X^T(W)$ denote the associated reproduction sequence.
By assumption, the remote mapping $Y^T\mapsto W\mapsto \hat X^T(W)$ satisfies
\[\mathbb{P}\!\left(\frac{1}{T}\sum_{t=1}^T d(X_t,\hat X_t(W))>D\right)\le \epsilon_{\mathrm{src}}.\]

Since $M\le M_{\mathrm{ch}}^\star(N_{\mathrm{ch}},\epsilon_{\mathrm{ch}})$, there exists a physical-layer channel code over $N_{\mathrm{ch}}=nT$ uses that communicates $W$ with block error probability at most $\epsilon_{\mathrm{ch}}$.
Since $M\le M_{\mathrm{comp}}^\star(L,\epsilon_{\mathrm{comp}})$, there exists an internal redundancy mechanism (e.g., an error-correcting code for the BSC$(\varepsilon)$ materializations) that stores/communicates $W$ across the $L=mT$ compute uses with block error probability at most $\epsilon_{\mathrm{comp}}$.

Let $\mathsf{E}_{\mathrm{src}}$ be the excess-distortion event under the correct index $W$, let $\mathsf{E}_{\mathrm{ch}}$ be the channel-decoding error event, and let $\mathsf{E}_{\mathrm{comp}}$ be the internal compute-decoding error event.
If none of these events occurs, then the receiver reconstructs $\hat X^T(W)$ and meets the distortion threshold.
Therefore,
\[
\epsilon_D^{\mathrm{blk}}(D)
\le \mathbb{P}(\mathsf{E}_{\mathrm{src}}\cup \mathsf{E}_{\mathrm{ch}}\cup \mathsf{E}_{\mathrm{comp}})
\le \mathbb{P}(\mathsf{E}_{\mathrm{src}})+\mathbb{P}(\mathsf{E}_{\mathrm{ch}})+\mathbb{P}(\mathsf{E}_{\mathrm{comp}})
\le \epsilon_{\mathrm{src}}+\epsilon_{\mathrm{ch}}+\epsilon_{\mathrm{comp}}
\le \epsilon,
\]
where the second inequality is a union bound and the third uses the respective design budgets.
\end{proof}

Proposition~\ref{prop:finiteT_concat_benchmark} composes three nonasymptotic primitives under the idealization that any additional bookkeeping needed to implement compute-side redundancy or channel decoding is already absorbed into the budgets $(n,m)$. Those circuit-level overheads matter for concrete implementations, but they are orthogonal to the information-flow limits emphasized here.

\subsection{A Computable Worked Example: Plugging $(p_{\mathrm{ue}},p_{\mathrm{er}})$ Into a Block Tail and Throughput}
\label{subsec:finiteT_tail_worked_example}
Proposition~\ref{prop:finiteT_concat_benchmark} expresses computation reliability through a block-error budget $\epsilon_{\mathrm{comp}}$ for the internal redundancy mechanism.
In safety-critical settings, it is often important to refine this budget by distinguishing \emph{undetected errors} (UE) from \emph{detected erasures} (ER) and to propagate these events into the excess-distortion tail.
Appendix~\ref{sec:erasures} introduces an OK/ER/UE abstraction with probabilities $(p_{\mathrm{ue}},p_{\mathrm{er}})$ and provides a general block-tail bound (Corollary~\ref{cor:tail_sandwich}).

For our word-level MCU template and an $r$-replica duplication-and-compare detector, the pair $(p_{\mathrm{ue}},p_{\mathrm{er}})$ admits the closed-form expression \eqref{eq:pue_r_replica_mcu} as a function of the MCU parameters and the replication factor $r$.
When these closed forms are inserted into the block tail bound for a length-$\Lif$ digital interface message (Subsection~\ref{subsec:tail_throughput_worked_example}), one obtains an explicit end-to-end constraint such as \eqref{eq:tail_block_interface_message_rreplica} and a sizing rule for $r$ (e.g., \eqref{eq:r_for_tail_target_interface_message_exact} in the $D_0$-safe regime).
Finally, the replication detection overhead reduces the effective compute budget and induces a throughput penalty quantified in \eqref{eq:lambda_replica_overhead}.

\begin{remark}[Normal-approximation form (design guideline)]
Under the regularity conditions of the finite-blocklength normal approximations
\eqref{eq:normal_supply_channel}--\eqref{eq:normal_supply_compute} and \eqref{eq:remote_normal_approx},
the sufficient condition \eqref{eq:finiteT_supply_demand_exact} can be written in the (per-sample) normal-approximation form
\begin{equation}
  R_{X|Y}(D)
  + \sqrt{\frac{V_{X|Y}(D)}{T}}\;Q^{-1}(\epsilon_{\mathrm{src}})
  \;\lesssim\;
  \min\Big\{
    n\Cch - \sqrt{\frac{nV_{\mathrm{ch}}}{T}}\;Q^{-1}(\epsilon_{\mathrm{ch}}),\;
    m\Cgate - \sqrt{\frac{mV_{\mathrm{gate}}}{T}}\;Q^{-1}(\epsilon_{\mathrm{comp}})
  \Big\},
  \label{eq:normal_supply_demand}
\end{equation}
where $\lesssim$ hides terms of order $O\!\big(\frac{\log T}{T}\big)$.
The key point is that the first-order converse bound $\Rsup$ in \eqref{eq:Rsupply_def} acquires a $T^{-1/2}$ backoff governed by dispersion on \emph{both} the communication and computation sides, as well as by the intrinsic task dispersion.
\end{remark}

\subsection{Dispersion-Aware Refinements for Hard-Separation and Serial Pipelines}
The first-order serial loss is already captured by Theorem~\ref{thm:interface_tax}; here we only add the finite-$T$ dispersion backoffs associated with those same cuts. Using the same normal-approximation reasoning as above, a finite-$T$ analog of Theorem~\ref{thm:interface_tax} is:
\begin{align}
  &R_{X|Y}(D) + \sqrt{\frac{V_{X|Y}(D)}{T}}\;Q^{-1}(\epsilon_{\mathrm{src}})\nonumber\\
  &\lesssim
  \min\Big\{
    n\Cch - \sqrt{\frac{nV_{\mathrm{ch}}}{T}}\;Q^{-1}(\epsilon_{\mathrm{ch}}),\;
    m_{\mathrm{dec}}\Cgate - \sqrt{\frac{m_{\mathrm{dec}}V_{\mathrm{gate}}}{T}}\;Q^{-1}(\epsilon_{\mathrm{dec}}),\;
    m_{\mathrm{task}}\Cgate - \sqrt{\frac{m_{\mathrm{task}}V_{\mathrm{gate}}}{T}}\;Q^{-1}(\epsilon_{\mathrm{task}})
  \Big\},
  \label{eq:normal_interface_tax}
\end{align}
with $\epsilon_{\mathrm{src}}+\epsilon_{\mathrm{ch}}+\epsilon_{\mathrm{dec}}+\epsilon_{\mathrm{task}}\le \epsilon$.
Equation~\eqref{eq:normal_interface_tax} is simply the normal-approximation counterpart of Theorem~\ref{thm:interface_tax}; see Remark~\ref{rem:interface_tax_not_sep} for the interpretation of the first-order serial loss.

\begin{remark}[Dispersion-aware compute splitting under hard-separation]
The equal-split rule from Theorem~\ref{thm:interface_tax} is refined in \eqref{eq:normal_interface_tax} by equalizing the \emph{effective} stage supplies
\begin{align}
  \tilde S_{\mathrm{dec}}(m_{\mathrm{dec}})
  &\triangleq
  m_{\mathrm{dec}}\Cgate
  - \sqrt{\frac{m_{\mathrm{dec}}V_{\mathrm{gate}}}{T}}\;Q^{-1}(\epsilon_{\mathrm{dec}}),
  \label{eq:eff_supply_dec}\\
  \tilde S_{\mathrm{task}}(m_{\mathrm{task}})
  &\triangleq
  m_{\mathrm{task}}\Cgate
  - \sqrt{\frac{m_{\mathrm{task}}V_{\mathrm{gate}}}{T}}\;Q^{-1}(\epsilon_{\mathrm{task}}),
  \label{eq:eff_supply_task}
\end{align}
subject to $m_{\mathrm{dec}}+m_{\mathrm{task}}=m$ (and with $[x]_+\triangleq \max\{x,0\}$ applied if the normal approximation yields a negative value).
When one sets symmetric reliability targets (e.g., $\epsilon_{\mathrm{dec}}=\epsilon_{\mathrm{task}}$) and both stages use identical primitives, symmetry implies that $m_{\mathrm{dec}}=m_{\mathrm{task}}=m/2$ is optimal (up to $O(T^{-1})$ terms).
When the stages have unequal reliability targets, the more stringent stage (smaller $\epsilon$, hence larger $Q^{-1}(\epsilon)$) typically requires a larger share of the compute budget to equalize \eqref{eq:eff_supply_dec} and \eqref{eq:eff_supply_task}.
\end{remark}

More generally, for the serial chain of Theorem~\ref{thm:k_stage_tax}, one replaces each stage supply $m_k\Cgate$ by its normal-approximation counterpart, so the same finite-$T$ refinement applies stagewise across the pipeline.

\subsection{SSCC vs.\ JSCC Under Committed Interfaces}
\label{subsec:sscc_jscc_interface_tax}
In the classical setting with reliable receiver computation, the source--channel separation principle incurs no first-order loss, but finite-blocklength performance can still depend strongly on whether one uses JSCC or separate source--channel coding (SSCC), i.e., concatenates a finite-blocklength lossy source code with a finite-blocklength channel code.
In particular, Kostina and Verd{\'u}~\cite{KostinaVerdu13_JSCC} showed that SSCC can be strictly suboptimal in the finite-blocklength regime under the excess-distortion probability criterion, and that joint design can deliver noticeable gains at short blocklengths.

Our dual-noise results add an architectural first-order effect to that classical finite-blocklength picture. The relevant loss is the serial hard-separation penalty already quantified by Theorem~\ref{thm:interface_tax} and Theorem~\ref{thm:k_stage_tax}; equation~\eqref{eq:normal_interface_tax} simply overlays the usual dispersion backoffs on those serial cuts. See Remark~\ref{rem:interface_tax_not_sep} for why this is not a classical separation failure.

Recent \emph{deep} JSCC schemes learn analog mappings from source samples to channel inputs and train a task decoder end-to-end, thereby avoiding an explicit decoded interface between ``channel decoding'' and downstream inference (e.g., wireless image transmission~\cite{BourtsoulatzeTCCN19_DeepJSCC} and semantic nonlinear transform coding~\cite{DaiJSAC22_NTSCC}). In our terminology these are soft-interface receivers in the sense of Definition~\ref{def:hard_separation} and Proposition~\ref{prop:no_tax_soft_interface}. By contrast, digital learned pipelines with an explicit decoded interface fall under the serial hard-separation abstractions of Theorem~\ref{thm:interface_tax} and Corollary~\ref{cor:hard_sep_supply_demand}. The relevant comparison under dual noise is therefore not ``analog versus digital'' per se, but whether the architecture preserves downstream access to the channel output or commits to a vulnerable digital interface; the practical mitigations are exactly the ones cataloged in Subsection~\ref{subsec:hetero_interface_tax}.

In low-latency dual-noise regimes, these dispersion-aware benchmarks serve as design guides and should be compared against explicit baselines such as uncoded/analog or hybrid task-direct schemes. In classical finite-blocklength JSCC, performance can depend sharply on code choice, and the nonasymptotic fundamental limit need not be monotone in blocklength~\cite{KostinaVerdu13_JSCC}. In our dual-noise setting, receiver architecture additionally matters because committed interfaces change the first-order supply.

\subsection{Gaussian Remote Estimation: An Explicit Dispersion Benchmark}
For scalar Gaussian remote estimation (Subsection~\ref{subsec:scalar_gaussian_example}) under an AWGN physical channel and BSC$(\varepsilon)$ compute primitives, the dispersions in \eqref{eq:normal_supply_demand} admit closed forms:
$V_{\mathrm{ch}}=V_{\mathrm{AWGN}}(\rho)$ in \eqref{eq:awgn_dispersion} and $V_{\mathrm{gate}}=V_{\mathrm{BSC}}(\varepsilon)$ in \eqref{eq:bsc_dispersion}.
Moreover, the indirect rate--distortion problem reduces to direct lossy coding of the Gaussian sufficient statistic $\tilde X=\mathbb{E}[X|Y]$ (Subsection~\ref{subsec:scalar_gaussian_example}), so the relevant source dispersion equals the Gaussian MSE rate--dispersion constant
\begin{equation}
  V_{X|Y}(D)=\frac{1}{2}(\log_2 e)^2 \quad \text{bits$^2$/sample},\qquad \sigma_{x|y}^2 < D < \sigma_x^2,
  \label{eq:gaussian_source_dispersion}
\end{equation}
independent of $D$ (cf.~\cite{KostinaVerdu13_JSCC}).
Substituting these quantities into \eqref{eq:normal_supply_demand} yields an explicit low-latency design benchmark for $D$ as a function of $(T,n,m,\epsilon)$, which can be written in closed form by inverting the Gaussian indirect rate--distortion function.

\begin{remark}[Normal-approximation distortion benchmark (scalar Gaussian remote estimation; design guideline)]
\label{rem:gaussian_normal_approx}
Consider scalar Gaussian remote estimation with quadratic distortion under an AWGN physical channel of SNR $\rho$ and BSC$(\varepsilon)$ compute primitives.
Fix a task-blocklength $T$ and error budgets $(\epsilon_{\mathrm{src}},\epsilon_{\mathrm{ch}},\epsilon_{\mathrm{comp}})$.
Define the normal-approximation supplies (bits/sample)
\begin{align}
  R_{\mathrm{ch,NA}}(T)
  &\triangleq
  n\Cch - \sqrt{\frac{nV_{\mathrm{AWGN}}(\rho)}{T}}\;Q^{-1}(\epsilon_{\mathrm{ch}}),
  \label{eq:Rch_NA}\\
  R_{\mathrm{comp,NA}}(T)
  &\triangleq
  m\Cgate - \sqrt{\frac{mV_{\mathrm{BSC}}(\varepsilon)}{T}}\;Q^{-1}(\epsilon_{\mathrm{comp}}),
  \label{eq:Rcomp_NA}\\
  R_{\mathrm{sup,NA}}(T)
  &\triangleq
  \min\{R_{\mathrm{ch,NA}}(T),\,R_{\mathrm{comp,NA}}(T)\},
  \label{eq:Rsup_NA}
\end{align}
and the effective available task rate after accounting for source dispersion
\begin{equation}
  R_{\mathrm{eff,NA}}(T)
  \triangleq
  \Big[R_{\mathrm{sup,NA}}(T) - \sqrt{\frac{V_{X|Y}(D)}{T}}\;Q^{-1}(\epsilon_{\mathrm{src}})\Big]_+,
  \label{eq:Reff_NA}
\end{equation}
where $[x]_+\triangleq \max\{x,0\}$.
The normal-approximation benchmark \eqref{eq:normal_supply_demand} predicts that operating points $(D,\epsilon)$ are feasible when
\begin{equation}
  D
  \ \gtrsim\
  \sigma_{x|y}^2 + (\sigma_x^2-\sigma_{x|y}^2)\,2^{-2 R_{\mathrm{eff,NA}}(T)},
  \label{eq:gaussian_D_normal_approx}
\end{equation}
with total excess-distortion probability budget $\epsilon \approx \epsilon_{\mathrm{src}}+\epsilon_{\mathrm{ch}}+\epsilon_{\mathrm{comp}}$.
For hard-separation, replace $R_{\mathrm{sup,NA}}(T)$ by the three-way minimum in \eqref{eq:normal_interface_tax} (with the corresponding stage error budgets).
We emphasize that \eqref{eq:gaussian_D_normal_approx} is a \emph{design-oriented} prediction based on normal approximations and is not claimed to furnish a matched characterization in general.
\end{remark}

\begin{remark}[Gaussian normal approximation under word-level MCUs]
If the compute substrate is dominated by word-level additive MCUs,
then the benchmark in Remark~\ref{rem:gaussian_normal_approx} continues to apply after replacing the compute terms $(\Cgate,V_{\mathrm{BSC}}(\varepsilon))$ in \eqref{eq:Rcomp_NA}
by the effective pair $(\Cgateeff,V_{\mathrm{gate}}^{\mathrm{eff}})$ from Remark~\ref{rem:mcu_dispersion_plugin}, i.e.,
\[
  R_{\mathrm{comp,NA}}(T)
  = m\Cgateeff - \sqrt{\frac{mV_{\mathrm{gate}}^{\mathrm{eff}}}{T}}\;Q^{-1}(\epsilon_{\mathrm{comp}}).
\]
\end{remark}

Figure~\ref{fig:gaussian_finiteT} plots the resulting normal-approximation distortion benchmark and contrasts task-direct and hard-separation finite-blocklength behavior against compute-reliable baselines.

\begin{figure}[t]
\centering
\scalebox{0.80}{%
\begin{tikzpicture}
\pgfmathsetmacro{\sigxy}{0.5}
\pgfmathsetmacro{\sigdelta}{0.5}
\pgfmathsetmacro{\nEx}{1.0}
\pgfmathsetmacro{\mEx}{2.0}
\pgfmathsetmacro{\mStage}{1.0} 
\pgfmathsetmacro{\CchEx}{2.0}
\pgfmathsetmacro{\VchEx}{1.03662}
\pgfmathsetmacro{\CgateEx}{0.53100}
\pgfmathsetmacro{\VgateEx}{0.90436}
\pgfmathsetmacro{\VsrcEx}{1.04068}
\pgfmathsetmacro{\Qdir}{2.71305}  
\pgfmathsetmacro{\Qsep}{2.80703}  
\pgfmathsetmacro{\Qjscc}{2.32635} 
\pgfmathsetmacro{\Qsscc}{2.57583} 

\begin{axis}[
  name=main,
  width=0.90\linewidth,
  height=0.58\linewidth,
  xlabel={Blocklength $T$ (task instances)},
  ylabel={$D_{\mathrm{NA}}(T)$},
  xmin=20,
  xmax=2000,
  ymin=0.5,
  ymax=1.0,
  xtick={20,500,1000,1500,2000},
  ytick={0.5,0.6,0.7,0.8,0.9,1.0},
  grid=major,
  major grid style={draw=black!15},
  tick align=outside,
  tick style={black},
  axis line style={black},
  every axis plot/.style={line width=1.8pt},
  legend columns=2,
  legend cell align=left,
  legend style={
    draw=none,
    fill=none,
    font=\footnotesize,
    at={(0.5,-0.135)},
    anchor=north,
    /tikz/every even column/.append style={column sep=0.9em}
  },
  clip=false,
]

\addplot[
  color=black,
  solid,
  mark=o,
  mark size=2.0pt,
  mark repeat=26,
  mark options={solid, fill=white},
  domain=20:2000,
  samples=320,
]
{\sigxy + \sigdelta*pow(2,-2*max(0,
  min(\nEx*\CchEx - \Qdir*sqrt(\nEx*\VchEx/x),
      \mEx*\CgateEx - \Qdir*sqrt(\mEx*\VgateEx/x))
  - \Qdir*sqrt(\VsrcEx/x)))};
\addlegendentry{Task-direct (vulnerable compute)}

\addplot[
  color=black,
  densely dashed,
  mark=square*,
  mark size=1.9pt,
  mark repeat=28,
  mark options={solid, fill=black, draw=black},
  domain=20:2000,
  samples=320,
]
{\sigxy + \sigdelta*pow(2,-2*max(0,
  min(\nEx*\CchEx - \Qsep*sqrt(\nEx*\VchEx/x),
      \mStage*\CgateEx - \Qsep*sqrt(\mStage*\VgateEx/x))
  - \Qsep*sqrt(\VsrcEx/x)))};
\addlegendentry{Hard-separation (equal split)}

\addplot[
  color=black,
  dash pattern=on 6pt off 1.4pt on 1.2pt off 1.4pt,
  mark=triangle*,
  mark size=2.3pt,
  mark repeat=30,
  mark options={solid, fill=black, draw=black},
  domain=20:2000,
  samples=320,
]
{\sigxy + \sigdelta*pow(2,-2*max(0,
  \nEx*\CchEx - \Qjscc*sqrt((\nEx*\VchEx + \VsrcEx)/x)))};
\addlegendentry{Reliable JSCC baseline}

\addplot[
  color=black,
  densely dotted,
  mark=diamond*,
  mark size=2.0pt,
  mark repeat=32,
  mark options={solid, fill=white, draw=black},
  domain=20:2000,
  samples=320,
]
{\sigxy + \sigdelta*pow(2,-2*max(0,
  \nEx*\CchEx - \Qsscc*sqrt(\nEx*\VchEx/x)
  - \Qsscc*sqrt(\VsrcEx/x)))};
\addlegendentry{Reliable SSCC baseline}

\draw[black, densely dashed, line width=0.6pt] (axis cs:450,0.5295) rectangle (axis cs:2000,0.5515);

\end{axis}

\begin{axis}[
  at={(main.north east)},
  anchor=north east,
  xshift=-2.0mm,
  yshift=-2.0mm,
  width=0.38\linewidth,
  height=0.24\linewidth,
  xmin=450,
  xmax=2000,
  ymin=0.5285,
  ymax=0.5482,
  xtick={500,1000,1500,2000},
  ytick={0.530,0.536,0.542,0.548},
  tick label style={font=\scriptsize},
  grid=major,
  major grid style={draw=black!12},
  axis background/.style={fill=white},
  tick align=outside,
  every axis plot/.style={line width=1.35pt},
]
\addplot[
  color=black,
  solid,
  mark=o,
  mark size=1.5pt,
  mark repeat=24,
  mark options={solid, fill=white},
  domain=450:2000,
  samples=240,
]
{\sigxy + \sigdelta*pow(2,-2*max(0,
  min(\nEx*\CchEx - \Qdir*sqrt(\nEx*\VchEx/x),
      \mEx*\CgateEx - \Qdir*sqrt(\mEx*\VgateEx/x))
  - \Qdir*sqrt(\VsrcEx/x)))};

\addplot[
  color=black,
  densely dashed,
  mark=square*,
  mark size=1.45pt,
  mark repeat=24,
  mark options={solid, fill=black, draw=black},
  domain=450:2000,
  samples=240,
]
{\sigxy + \sigdelta*pow(2,-2*max(0,
  min(\nEx*\CchEx - \Qsep*sqrt(\nEx*\VchEx/x),
      \mStage*\CgateEx - \Qsep*sqrt(\mStage*\VgateEx/x))
  - \Qsep*sqrt(\VsrcEx/x)))};

\addplot[
  color=black,
  dash pattern=on 6pt off 1.4pt on 1.2pt off 1.4pt,
  mark=triangle*,
  mark size=1.8pt,
  mark repeat=24,
  mark options={solid, fill=black, draw=black},
  domain=450:2000,
  samples=240,
]
{\sigxy + \sigdelta*pow(2,-2*max(0,
  \nEx*\CchEx - \Qjscc*sqrt((\nEx*\VchEx + \VsrcEx)/x)))};

\addplot[
  color=black,
  densely dotted,
  mark=diamond*,
  mark size=1.55pt,
  mark repeat=24,
  mark options={solid, fill=white, draw=black},
  domain=450:2000,
  samples=240,
]
{\sigxy + \sigdelta*pow(2,-2*max(0,
  \nEx*\CchEx - \Qsscc*sqrt(\nEx*\VchEx/x)
  - \Qsscc*sqrt(\VsrcEx/x)))};

\node[font=\scriptsize, fill=white, inner sep=1pt, anchor=west]
  at (axis cs:1150,0.5440) {Reliable SSCC};
\node[font=\scriptsize, fill=white, inner sep=1pt, anchor=west]
  at (axis cs:1150,0.5315) {Reliable JSCC};
\end{axis}
\end{tikzpicture}%
}
\vspace{-3mm}
\caption{Finite-blocklength benchmark for scalar Gaussian remote estimation under the normal approximation (Remark~\ref{rem:gaussian_normal_approx}). The vulnerable-compute task-direct and equal-split hard-separation architectures are compared with compute-reliable JSCC and SSCC baselines under the same communication parameters. The inset enlarges the large-$T$ tail, where the two compute-reliable baselines nearly overlap but remain visually distinct. Example parameters: $\sigma_x^2=\sigma_v^2=1$ (so $\sigma_{x|y}^2=1/2$), AWGN SNR $\rho=15$ (so $\Cch=2$ and $V_{\mathrm{AWGN}}(\rho)\approx 1.0366$), BSC compute primitive with $\varepsilon=0.1$ (so $\Cgate\approx 0.5310$ and $V_{\mathrm{BSC}}(\varepsilon)\approx 0.9044$), and budgets $(n,m)=(1,2)$. We fix the total excess-distortion probability $\epsilon=0.01$ for all curves. For task-direct (A), we use the equal three-way split $\epsilon_{\mathrm{src}}=\epsilon_{\mathrm{ch}}=\epsilon_{\mathrm{comp}}=\epsilon/3$ (so $Q^{-1}(\epsilon/3)\approx 2.713$). For hard-separation (B), we use an equal compute split $m_{\mathrm{dec}}=m_{\mathrm{task}}=m/2$ and an equal four-way error split $\epsilon_{\mathrm{src}}=\epsilon_{\mathrm{ch}}=\epsilon_{\mathrm{dec}}=\epsilon_{\mathrm{task}}=\epsilon/4$ (so $Q^{-1}(\epsilon/4)\approx 2.807$). We also plot the compute-reliable ($\varepsilon=0$) JSCC/SSCC normal-approximation baselines from~\cite{KostinaVerdu13_JSCC}, using total excess-distortion probability $\epsilon$ for JSCC and an equal two-way split $(\epsilon/2,\epsilon/2)$ for SSCC. All curves approach their first-order limits as $T$ grows. Hard-separation includes the additional two-stage serial cut and typically a larger finite-$T$ penalty.}
\label{fig:gaussian_finiteT}
\end{figure}
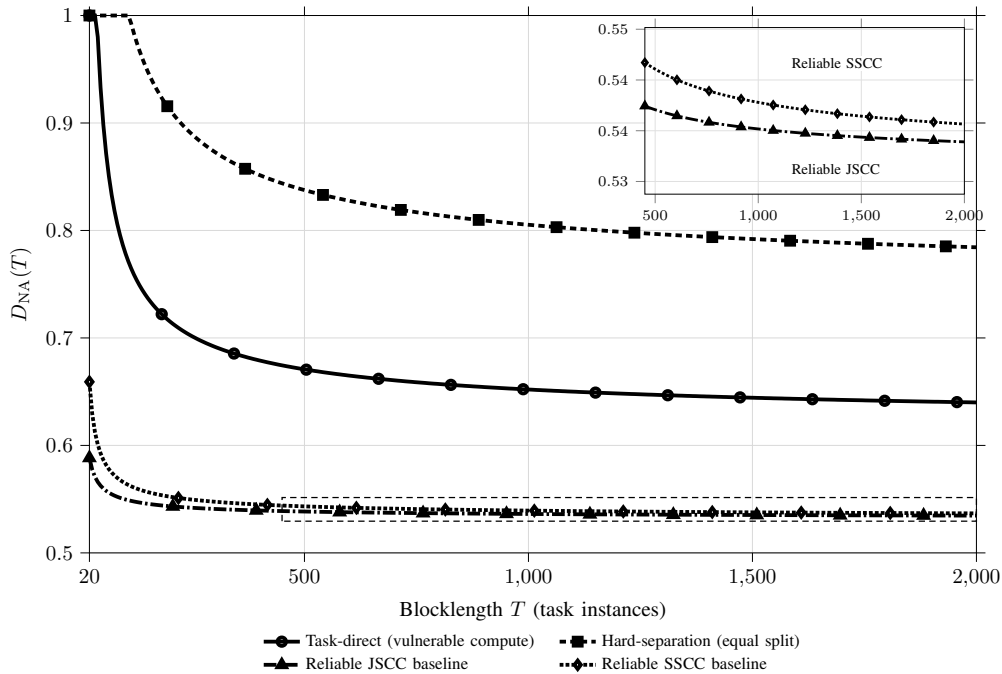

\section{Supplementary Gaussian Variants and Additional Examples}
\label{app:gaussian_supp}
\indent This appendix collects the scalar-Gaussian strict-gap witness, noisy-logic/word-level plug-ins, uncoded analog baselines, and vector Gaussian extensions omitted from the compact main-text examples section.

\subsection{Additional Scalar-Gaussian Benchmark Material}
\begin{corollary}[Explicit strict separation distortion gap (scalar Gaussian, compute-limited)]
Assume the achievability conditions of Theorem~\ref{thm:achievability_general} and that $\condcomp$.
Define
\begin{equation}
  D_{\mathrm{direct}}^\star \triangleq \sigma_{x|y}^2 + \big(\sigma_x^2-\sigma_{x|y}^2\big)\,2^{-2m\Cgate},
  \qquad
  D_{\mathrm{sep}}^\star \triangleq \sigma_{x|y}^2 + \big(\sigma_x^2-\sigma_{x|y}^2\big)\,2^{-m\Cgate}.
\end{equation}
Then any distortion level $D> D_{\mathrm{direct}}^\star$ is achievable by task-direct processing (Architecture~(A)),
whereas any hard-separation scheme (Architecture~(B)) satisfies $D\ge D_{\mathrm{sep}}^\star$ under the optimal split
$m_{\mathrm{dec}}=m_{\mathrm{task}}=\frac{m}{2}$.
Consequently, every distortion level $D\in( D_{\mathrm{direct}}^\star,\ D_{\mathrm{sep}}^\star )$ is achievable by task-direct processing but impossible under hard-separation.
\end{corollary}

\begin{proof}
In the compute-limited regime $\condcomp$, Theorem~\ref{thm:achievability_general} together with the scalar Gaussian identity
\[
  R_{X|Y}(D)=\frac12\log_2\!\left(\frac{\sigma_x^2-\sigma_{x|y}^2}{D-\sigma_{x|y}^2}\right)
\]
shows that every $D>D_{\mathrm{direct}}^\star$ is task-direct achievable, while Corollary~\ref{cor:gaussian_sep} yields the converse bound $D\ge D_{\mathrm{sep}}^\star$ under the optimal split. Since $D_{\mathrm{direct}}^\star < D_{\mathrm{sep}}^\star$ whenever $m\Cgate>0$, the stated gap follows.
\end{proof}

This strict-gap witness also admits a simple visual interpretation in normalized excess distortion. In the compute-limited regime, the quantity
\[
  \frac{D-\sigma_{x|y}^2}{\sigma_x^2-\sigma_{x|y}^2}
\]
decays as $2^{-2R}$ under task-direct processing and as $2^{-R}$ under symmetric hard-separation, where $R\triangleq m\Cgate$. Figure~\ref{fig:interface_tax_gaussian} plots these two decay laws, quantifying one compute-limited slice of the architecture boundaries summarized in Figure~\ref{fig:arch_phase_supply_plane}.

\begin{figure}[t]
\centering
\includefigure[width=0.7\linewidth]{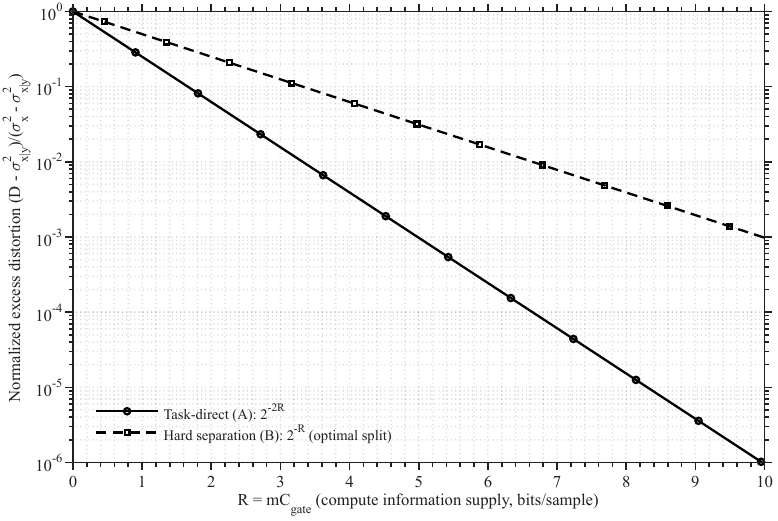}
\vspace{-3mm}
\caption{Square-root illustration in normalized excess distortion for scalar Gaussian remote estimation in the compute-limited regime. The task-direct curve decays as $2^{-2R}$, while the symmetric hard-separation curve decays as $2^{-R}$, where $R=m\Cgate$.}
\label{fig:interface_tax_gaussian}
\end{figure}

\subsection{Plug-In Variants: Word-Level Upsets and Noisy Receiver Logic}
\label{subsec:gaussian_plugins}
The closed-form expression \eqref{eq:gaussian_distortion_bound} extends immediately to other receiver-compute abstractions via the plug-in philosophy of this paper: replace the compute supply term by a valid per-instance information upper bound for the corresponding primitive model.

\begin{corollary}[Scalar Gaussian bound under word-level MCUs]
\label{cor:gaussian_mcu}
Under the scalar Gaussian model, suppose the compute substrate is a word-level additive upset primitive
$Z^{(w)}=U^{(w)}\oplus E$ with word size $w$ and error pattern $E\sim P_E$.
Assume the per-instance compute budget $m$ counts vulnerable \emph{bits}, so that $m$ bit materializations correspond to at most $m/w$ word uses.
Then any task-direct scheme must satisfy
\begin{equation}
  \boxed{
  D \ \ge\ \sigma_{x|y}^2 + \big(\sigma_x^2-\sigma_{x|y}^2\big)\,2^{-2\Rsup^{(\mathrm{MCU})}},
  \qquad
  \Rsup^{(\mathrm{MCU})}=\min\{n\Cch,\ m\Cgateeff\},
  }
  \label{eq:gaussian_distortion_bound_mcu}
\end{equation}
where $\Cgateeff=1-\frac{H(E)}{w}$ (cf.~Corollary~\ref{cor:mcu_plugin}).
\end{corollary}

\begin{proof}
By Corollary~\ref{cor:mcu_plugin}, the scalar Gaussian task-direct design must satisfy
$R_{X|Y}(D)\le \min\{n\Cch,m\Cgateeff\}$.
Substituting the Gaussian expression for $R_{X|Y}(D)$ and solving for $D$ yields \eqref{eq:gaussian_distortion_bound_mcu}.
\end{proof}

\begin{corollary}[Closed-form scalar Gaussian bound under noisy receiver logic]
Under the scalar Gaussian model, suppose the receiver computation is implemented by a $K_{\mathrm{fan}}$-fan-in noisy-gate circuit with gate flip probability $\delta$ and depth at least $d_{\mathrm{logic}}$ as in Theorem~\ref{thm:supply_demand_noisy_gate}.
Then any task-direct scheme must satisfy
\begin{equation}
  \boxed{
  D \ \ge\ \sigma_{x|y}^2 + \big(\sigma_x^2-\sigma_{x|y}^2\big)\,2^{-2\Rsup^{(\mathrm{logic})}},
  }
  \label{eq:gaussian_distortion_bound_noisy_gate}
\end{equation}
where
\begin{equation}
  \Rsup^{(\mathrm{logic})}
  \triangleq
  \min\Big\{\, n\Cch,\ m\,C_{\mathrm{gate}}(\delta),\ m\,\min\!\big\{1,\ \big(K_{\mathrm{fan}}(1-2\delta)^2\big)^{d_{\mathrm{logic}}}\big\}\Big\}.
\end{equation}
\end{corollary}

\begin{proof}
The proof is identical to Corollary~\ref{cor:gaussian_closed_form}, replacing the supply constraint from Theorem~\ref{thm:supply_demand_converse} by the depth-dependent supply in Theorem~\ref{thm:supply_demand_noisy_gate}.
\end{proof}
\subsection{To Code or Not to Code: A Symbol-by-Symbol Analog Baseline}
\label{subsec:to_code_or_not_gaussian}
For matched scalar Gaussian AWGN/MSE, zero-delay uncoded linear transmission is Shannon-optimal~\cite{Goblick65}. We include this subsection only as a baseline. It shows that a low-latency task-direct mapping can avoid a hard digital interface, while Subsection~\ref{subsec:vector_gaussian} gives the clearer high-dimensional interface-tax witness.

Specialize the physical channel to a real AWGN channel
\begin{equation}
  R = S + W,\qquad W\sim\mathcal{N}(0,\sigma_w^2),
\end{equation}
with average power constraint $\mathbb{E}[S^2]\le P$, so that $\Cch=\frac12\log_2(1+P/\sigma_w^2)$.
Let the encoder compute the Gaussian sufficient statistic
\begin{equation}
  \tilde X \triangleq \mathbb{E}[X|Y] = \frac{\sigma_x^2}{\sigma_x^2+\sigma_v^2}\,Y,
  \qquad
  \mathrm{Var}(\tilde X)=\sigma_x^2-\sigma_{x|y}^2,
\end{equation}
and transmit it \emph{uncoded} (symbol-by-symbol) via
\begin{equation}
  S = \sqrt{\frac{P}{\sigma_x^2-\sigma_{x|y}^2}}\;\tilde X.
\end{equation}
The receiver forms the MMSE estimate $\hat X=\mathbb{E}[X|R]$.
Since $X=\tilde X + (X-\tilde X)$ with $\tilde X\perp (X-\tilde X)$ in the jointly Gaussian model, the resulting MSE is
\begin{equation}
  D_{\mathrm{analog}}
  = \sigma_{x|y}^2
  + \mathrm{Var}(\tilde X\,|\,R)
  = \sigma_{x|y}^2 + (\sigma_x^2-\sigma_{x|y}^2)\frac{\sigma_w^2}{P+\sigma_w^2}
  = \sigma_{x|y}^2 + (\sigma_x^2-\sigma_{x|y}^2)2^{-2\Cch}.
  \label{eq:analog_gaussian_mse}
\end{equation}
Comparing \eqref{eq:analog_gaussian_mse} to \eqref{eq:gaussian_distortion_bound} shows that in the \emph{\regch} and for $n=1$, this uncoded task-direct scheme achieves the benchmark without coding across task instances ($T=1$).

As a coarse fixed-point proxy, if the receiver quantizes the linear estimate to $b$ output bits over a clipped range $[-A,A]$, then the added MSE is approximately
\begin{equation}
  \mathbb{E}\big[(\hat X-\hat X^{(b)})^2\big]
  \;\approx\;
  \frac{A^2}{3}\,2^{-2b},
\end{equation}
so achieving an additional penalty of about $\delta_q$ requires
$b \gtrsim \frac{1}{2}\log_2\!\big(\frac{A^2}{3\delta_q}\big)$ reliable output bits. Under our compute substrate model this corresponds roughly to $m\Cgate \gtrsim b$ in the first-order regime (or its finite-$T$ refinement in \eqref{eq:normal_supply_demand}), which is why uncoded/analog or otherwise shallow estimators can be attractive in low-latency compute-limited regimes.

\subsection{Vector Gaussian Remote Estimation (Extension)}
\label{subsec:vector_gaussian}
A more revealing Gaussian witness arises in vector regimes, where simple symbol-by-symbol linear mappings need not be optimal and digital compression is often beneficial.

Consider the vector Gaussian remote estimation model with $X\in\mathbb{R}^p$
\begin{equation}
  X \sim \mathcal{N}(0,\Sigma_X),\qquad Y = X + V,\quad V\sim \mathcal{N}(0,\Sigma_V),
\end{equation}
with $X$ and $V$ independent, and quadratic distortion $d(x,\hat x)=\|x-\hat x\|_2^2$.
Let $\Sigma_{X|Y}$ denote the conditional covariance of $X$ given $Y$ and let
\begin{equation}
  \Sigma_{\tilde X}\triangleq \mathrm{Cov}\big(\mathbb{E}[X|Y]\big)=\Sigma_X-\Sigma_{X|Y}.
\end{equation}
Let $\{\lambda_i\}_{i=1}^p$ be the eigenvalues of $\Sigma_{\tilde X}$.

The indirect rate--distortion function satisfies the classical decomposition~\cite{BergerBook71,Witsenhausen80}
\begin{equation}
  R_{X|Y}(D) = R_{\tilde X}\big(D-\mathrm{tr}(\Sigma_{X|Y})\big),
\end{equation}
where $R_{\tilde X}(\cdot)$ is the (direct) Gaussian rate--distortion function of $\tilde X=\mathbb{E}[X|Y]$.
Consequently, the distortion--rate function is given by water-filling:
\begin{equation}
  D(R) = \mathrm{tr}(\Sigma_{X|Y}) + \sum_{i=1}^p \min\{\nu,\lambda_i\},
  \qquad
  R = \frac{1}{2}\sum_{i=1}^p \Big[\log_2\!\Big(\frac{\lambda_i}{\nu}\Big)\Big]_+,
  \label{eq:vector_gaussian_waterfill}
\end{equation}
where $\nu\ge 0$ is chosen to meet the rate constraint.

\begin{corollary}[Vector Gaussian distortion bound under dual noise]
\label{cor:vector_gaussian_dual_noise}
For the vector Gaussian model, any task-direct scheme (Architecture~(A)) must satisfy
\begin{equation}
  D \ \ge\ D(\Rsup)
\end{equation}
where $D(R)$ is defined by \eqref{eq:vector_gaussian_waterfill} and $\Rsup$ is defined by \eqref{eq:Rsupply_def}.
Under hard-separation (Architecture~(B)), the same bound holds with $\Rsup$ replaced by $\min\{n\Cch,\frac{m}{2}\Cgate\}$.
\end{corollary}

\begin{proof}
Combine Theorem~\ref{thm:supply_demand_converse} (or Theorem~\ref{thm:interface_tax} for hard-separation) with the Gaussian indirect rate--distortion characterization in \eqref{eq:vector_gaussian_waterfill}.
\end{proof}

\begin{remark}[Diagonal/isotropic closed forms, the symmetric hard-separation specialization, and a one-line numerical instantiation]
For the additive observation model $Y=X+V$ above, the conditional covariance admits the closed form
\begin{equation}
  \Sigma_{X|Y} = \big(\Sigma_X^{-1}+\Sigma_V^{-1}\big)^{-1},
  \label{eq:vector_gaussian_cond_cov}
\end{equation}
(when $\Sigma_X$ and $\Sigma_V$ are invertible), and the sufficient statistic is linear:
$\tilde X=\mathbb{E}[X|Y]=\Sigma_X(\Sigma_X+\Sigma_V)^{-1}Y$, with covariance
\begin{equation}
  \Sigma_{\tilde X}=\Sigma_X-\Sigma_{X|Y}=\Sigma_X(\Sigma_X+\Sigma_V)^{-1}\Sigma_X.
  \label{eq:vector_gaussian_suff_cov}
\end{equation}
If $\Sigma_X$ and $\Sigma_V$ commute (e.g., are simultaneously diagonalizable), then writing their eigenvalues as
$\sigma_{x,i}^2$ and $\sigma_{v,i}^2$ yields
$\lambda_i=\frac{\sigma_{x,i}^4}{\sigma_{x,i}^2+\sigma_{v,i}^2}$ and
$\mathrm{tr}(\Sigma_{X|Y})=\sum_i\frac{\sigma_{x,i}^2\sigma_{v,i}^2}{\sigma_{x,i}^2+\sigma_{v,i}^2}$, making \eqref{eq:vector_gaussian_waterfill} fully explicit.

In the isotropic case $\Sigma_X=\sigma_x^2 I_p$ and $\Sigma_V=\sigma_v^2 I_p$, all modes are active and \eqref{eq:vector_gaussian_waterfill} reduces to
\begin{equation}
  D(R)
  = p\,\frac{\sigma_x^2\sigma_v^2}{\sigma_x^2+\sigma_v^2}
    + p\,\frac{\sigma_x^4}{\sigma_x^2+\sigma_v^2}\,2^{-2R/p}.
  \label{eq:vector_gaussian_isotropic_DR}
\end{equation}
Therefore, in a compute-limited regime where $\Rsup=m\Cgate$ dominates and hard-separation forces
$\Rsuphard=\frac{1}{2}\Rsup$,
the excess distortion above the MMSE floor obeys
\begin{equation}
  D\!\Big(\tfrac{1}{2}\Rsup\Big)-\mathrm{tr}(\Sigma_{X|Y})
  \;=\;
  \big(D(\Rsup)-\mathrm{tr}(\Sigma_{X|Y})\big)\,2^{\Rsup/p},
  \label{eq:vector_gaussian_isotropic_tax_ratio}
\end{equation}
making the first-order supply loss $\Rsup\mapsto\frac{1}{2}\Rsup$ quantitatively visible even in vector regimes where digital compression is usually advantageous.

As a minimal isotropic instantiation, take $(p,\sigma_x^2,\sigma_v^2,m\Cgate)=(16,1,1/4,48)$, so $\Rsup=48$ bits per length-$p$ vector (3 bits/dimension) whereas hard-separation yields $\Rsuphard=24$ bits. The MMSE floor is $16\cdot \frac{1\cdot (1/4)}{1+1/4}=3.2$, and \eqref{eq:vector_gaussian_isotropic_DR} gives $D(48)=3.4$ versus $D(24)=4.8$; equivalently, the excess distortion jumps from $0.2$ to $1.6$, matching the factor $2^{\Rsup/p}=2^{48/16}=2^3$ predicted by \eqref{eq:vector_gaussian_isotropic_tax_ratio}.
\end{remark}

\begin{remark}[Worked diagonal examples deferred to the appendix]
Appendix~\ref{app:vector_gaussian_worked_example} gives a hand-checkable $p=2$ example and a higher-dimensional plotted gap (Fig.~\ref{fig:vector_gap_uncoded_hard}), making explicit both the compute-limited MSE gap from $m\Cgate\mapsto \tfrac{1}{2}m\Cgate$ and the delay in channel-limited saturation under symmetric hard-separation.
\end{remark}

Corollary~\ref{cor:vector_gaussian_dual_noise} also serves as a tractable local linear-Gaussian proxy for high-dimensional learned inference: around an operating point, Jacobian- or Gauss--Newton-type linearizations often yield approximately linear-Gaussian error models, and \eqref{eq:vector_gaussian_waterfill} then predicts how the available supply $\min\{n\Cch,m\Cgate\}$ should be spent across principal modes. In per-second notation (Appendix~\ref{sec:throughput}), a task-arrival rate $\lambda$ maps this to $R=\Rsup(\lambda)=\frac{1}{\lambda}\min\{\mathcal{B}\Cch,\,\mathcal{G}\Cgate\}$ bits per length-$p$ vector, with hard-separation halving the compute term to $\tfrac{1}{2}\mathcal{G}\Cgate$.


\section{Supplementary Safer Interfaces: Undetected Errors Versus Erasures}
\label{sec:erasures}

\indent We discuss undetected-error versus erasure handling and show how erasure-aware designs can improve tail reliability under compute budgets.

\paragraph{From bit flips to $\mathsf{OK}/\mathsf{ER}/\mathsf{UE}$}
Under either i.i.d.\ single-bit flips or correlated word-level upsets, a receiver implementation typically includes lightweight invariants (e.g., CRC/parity, range checks for floating-point fields, consistency tests, or redundant execution) that can detect a subset of corruptions.
Detected corruptions can be surfaced as explicit failures (erasures), converting what would be silent data corruption into $\mathsf{ER}$ events, as in software systems such as Radshield~\cite{WangRadshieldASPLOS26} and SAVE~\cite{ZhengSAVEATC25}.
The remaining undetected corruptions correspond to $\mathsf{UE}$ events and can induce rare but catastrophic outliers, especially when faults affect control flow or high-significance floating-point fields (e.g.,~\cite{LiOneBitFlipUSENIX24}), motivating the erasure-aware fallback and clipping principles developed below.

\subsection{A Minimal Decoder Outcome Model}
A noisy decoder may output either a valid (possibly correct) payload, an undetected wrong payload, or declare a failure.
Let $\mathsf{OK}$, $\mathsf{UE}$, and $\mathsf{ER}$ denote the correct-output, undetected-error, and detected-erasure events, respectively, with probabilities
\begin{equation}
  \mathbb{P}[\mathsf{UE}] = p_{\mathrm{ue}},\qquad
  \mathbb{P}[\mathsf{ER}] = p_{\mathrm{er}},\qquad
  \mathbb{P}[\mathsf{OK}] = 1-p_{\mathrm{ue}}-p_{\mathrm{er}}.
\end{equation}

\subsection{Task-Level Impact: Three-Outcome Decomposition and Erasure-Aware Fallback}
\begin{proposition}[MSE decomposition under $\mathsf{OK}/\mathsf{ER}/\mathsf{UE}$]
\label{prop:erasure_fallback_mse}
Suppose the receiver outputs $\hat X^{(\mathsf{OK})}$ on $\mathsf{OK}$,
outputs a fallback estimate $\hat X^{(\mathsf{FB})}$ on $\mathsf{ER}$,
and outputs some (possibly corrupted) estimate $\hat X^{(\mathsf{UE})}$ on $\mathsf{UE}$.
Then the mean-squared error satisfies
\begin{equation}
  D
  = p_{\mathrm{ok}}\,D_{\mathsf{OK}} + p_{\mathrm{er}}\,D_{\mathsf{FB}} + p_{\mathrm{ue}}\,D_{\mathsf{UE}},
  \label{eq:mse_three_event}
\end{equation}
where $p_{\mathrm{ok}}\triangleq 1-p_{\mathrm{ue}}-p_{\mathrm{er}}$ and
\begin{align}
  D_{\mathsf{OK}} &\triangleq \mathbb{E}\!\big[(X-\hat X^{(\mathsf{OK})})^2 \mid \mathsf{OK}\big],\\
  D_{\mathsf{FB}} &\triangleq \mathbb{E}\!\big[(X-\hat X^{(\mathsf{FB})})^2 \mid \mathsf{ER}\big],\\
  D_{\mathsf{UE}} &\triangleq \mathbb{E}\!\big[(X-\hat X^{(\mathsf{UE})})^2 \mid \mathsf{UE}\big].
\end{align}
In particular, if $p_{\mathrm{ue}}=0$ (no undetected errors), then
$D=(1-p_{\mathrm{er}})D_{\mathsf{OK}}+p_{\mathrm{er}}D_{\mathsf{FB}}$.
\end{proposition}

\begin{proof}
This follows immediately from the law of total expectation over the partition $\{\mathsf{OK},\mathsf{ER},\mathsf{UE}\}$.
\end{proof}

\paragraph*{Short consequence (block Markov bound)}
For nonnegative distortion and any threshold $D_0>0$,
\begin{equation}
  \boxed{
  \epsilon_D^{\mathrm{blk}}(D_0)
  \ \le\
  \frac{1}{D_0}\cdot \frac{1}{T}\sum_{t=1}^T \mathbb{E}\!\big[d(X_t,\hat X_t)\big].
  }
  \label{eq:block_markov_template}
\end{equation}
In the stationary setting this simplifies to $\epsilon_D^{\mathrm{blk}}(D_0)\le D/D_0$, where $D=\mathbb{E}[d(X,\hat X)]$ is the average distortion. Under the $\mathsf{OK}/\mathsf{ER}/\mathsf{UE}$ model of Proposition~\ref{prop:erasure_fallback_mse}, the same bound becomes
\begin{equation}
  \epsilon_D^{\mathrm{blk}}(D_0)
  \ \le\
  \frac{(1-p_{\mathrm{ue}}-p_{\mathrm{er}})\,D_{\mathsf{OK}} + p_{\mathrm{er}}\,D_{\mathsf{FB}} + p_{\mathrm{ue}}\,D_{\mathsf{UE}}}{D_0},
  \label{eq:block_markov_three_event}
\end{equation}
which turns any computable upper bound on $D_{\mathsf{UE}}$ (e.g., Proposition~\ref{prop:clipping_mse}) into a directly computable block-tail benchmark.

\begin{proposition}[Checksum/hashing can provably convert $\mathsf{UE}$ into $\mathsf{ER}$ using $h$ protected bits]
\label{prop:hash_detection}
Let $U\in\{0,1\}^L$ denote any task-critical intermediate digital state (e.g., an interface message $\hat B$, a control-plane state, or a cached tensor block) that may be corrupted into $U'$ by the noisy substrate.
Fix an integer $h\ge 1$ and choose a hash/check family $\{f_k:\{0,1\}^L\to\{0,1\}^h\}_{k\in\mathcal{K}}$ that is \emph{2-universal}, i.e., for any $u\neq u'$,
\[
  \Pr_{k\sim \mathrm{Unif}(\mathcal{K})}\big[f_k(u)=f_k(u')\big]\ \le\ 2^{-h}.
\]
Assume the hash key $k$ and tag $t=f_k(U)$ are stored (or computed) on a sufficiently protected/reliable region (e.g., a reliable memory island),
and that the corruption mechanism producing $U'$ is independent of $k$.\footnote{This ``independent key'' assumption is natural for random hardware faults; it is the standard setting for universal-hash error detection.}
Upon retrieving $U'$, declare a detected failure ($\mathsf{ER}$) unless $f_k(U')=t$.
Then the resulting undetected error probability satisfies
\begin{equation}
  p_{\mathrm{ue}}
  \ \le\
  2^{-h}.
  \label{eq:pue_hash_bound}
\end{equation}
Consequently, to enforce a block-level excess-distortion target via Corollary~\ref{cor:tail_sandwich},
it is sufficient to ensure
\[
  T\Big((1-p_{\mathrm{ue}}-p_{\mathrm{er}})\delta_{\mathsf{OK}}+p_{\mathrm{er}}\delta_{\mathsf{FB}}+p_{\mathrm{ue}}\Big)\le \epsilon.
\]
In the additional safe nominal / safe fallback regime discussed below, this reduces to the simpler design rule $p_{\mathrm{ue}}\lesssim \epsilon/T$.
Using \eqref{eq:pue_hash_bound}, it therefore suffices to allocate
\begin{equation}
  h\ \ge\ \left\lceil \log_2\!\frac{T}{\epsilon}\right\rceil
  \label{eq:h_for_target_pue}
\end{equation}
protected bits for hashing/checking.
\end{proposition}

\begin{proof}
Condition on any realization of the corrupted state $U'=u'$.
If $u'=U$, then no error occurs.
If $u'\neq U$, then by 2-universality,
$\Pr_k[f_k(u')=f_k(U)]\le 2^{-h}$.
Averaging over $U'$ yields \eqref{eq:pue_hash_bound}.
\end{proof}

\begin{proposition}[Duplication-and-compare detection under additive faults]
\label{prop:dup_compare}
Let $U\in\{0,1\}^L$ denote any task-critical intermediate digital state (e.g., an interface message $\hat B$, a control-plane state, or a cached tensor block).
Consider a duplication-and-compare detector that stores/retrieves two replicas through additive fault patterns:
\begin{equation}
  Z^{(1)} = U\oplus E_1,\qquad Z^{(2)}=U\oplus E_2,
  \label{eq:dup_compare_model}
\end{equation}
where $E_1,E_2\in\{0,1\}^L$ are random error patterns.
The detector declares a detected failure ($\mathsf{ER}$) unless $Z^{(1)}=Z^{(2)}$, in which case it accepts $\hat U\triangleq Z^{(1)}$.

Assume the comparison and control logic of the detector are reliable (or protected via a reliable island).
Then the undetected-error event satisfies
\[
  \mathsf{UE}
  \;=\;
  \{Z^{(1)}=Z^{(2)}\neq U\}
  \;=\;
  \{E_1=E_2\neq 0^L\},
\]
and therefore
\begin{equation}
  p_{\mathrm{ue}} \ =\ \mathbb{P}(E_1=E_2,\ E_1\neq 0^L).
  \label{eq:pue_dup_general}
\end{equation}

In particular, if $E_1$ and $E_2$ are independent with a common pmf $P_E$, then
\begin{equation}
  p_{\mathrm{ue}}
  = \sum_{e\neq 0^L} P_E(e)^2,
  \qquad
  p_{\mathrm{ok}} = P_E(0^L)^2,
  \qquad
  p_{\mathrm{er}} = 1-\sum_{e}P_E(e)^2.
  \label{eq:pue_dup_iid}
\end{equation}
\emph{Extension to $r$ replicas.}
If the detector stores/retrieves $r\ge 2$ replicas through additive faults
$Z^{(j)}=U\oplus E_j$ for $j=1,\dots,r$ and declares $\mathsf{ER}$ unless \emph{all} retrieved replicas agree,
then the undetected-error event is $\{E_1=\cdots=E_r\neq 0^L\}$ and hence
\begin{equation}
  p_{\mathrm{ue}}
  = \mathbb{P}(E_1=\cdots=E_r,\ E_1\neq 0^L).
\end{equation}
Under i.i.d.\ faults $E_1,\dots,E_r\stackrel{\mathrm{i.i.d.}}{\sim}P_E$, this reduces to
\begin{equation}
  p_{\mathrm{ue}} = \sum_{e\neq 0^L} P_E(e)^r,
  \qquad
  p_{\mathrm{ok}} = P_E(0^L)^r,
  \qquad
  p_{\mathrm{er}} = 1-\sum_{e}P_E(e)^r.
  \label{eq:pue_r_replica}
\end{equation}

More generally, under a common-mode mixture model with parameter $\theta\in[0,1]$ in which
with probability $\theta$ the two replicas suffer the same fault ($E_1=E_2=E$ with $E\sim P_E$) and with probability $1-\theta$ they suffer independent faults $E_1,E_2\stackrel{\mathrm{i.i.d.}}{\sim}P_E$, one has
\begin{align}
  p_{\mathrm{ue}}
  &= \theta\big(1-P_E(0^L)\big) + (1-\theta)\sum_{e\neq 0^L} P_E(e)^2,
  \label{eq:pue_dup_common}\\
  p_{\mathrm{ok}}
  &= \theta P_E(0^L) + (1-\theta)P_E(0^L)^2,\nonumber\\
  p_{\mathrm{er}}
  &= 1-p_{\mathrm{ok}}-p_{\mathrm{ue}}.\nonumber
\end{align}
\end{proposition}

\begin{proof}
The test accepts iff $Z^{(1)}=Z^{(2)}$.
Under the additive model \eqref{eq:dup_compare_model}, $Z^{(1)}=Z^{(2)}$ iff $E_1=E_2$.
Moreover, acceptance is incorrect iff $Z^{(1)}\neq U$, i.e., $E_1\neq 0^L$.
This proves \eqref{eq:pue_dup_general}.
If $E_1,E_2$ are i.i.d.\ with pmf $P_E$, then
$\mathbb{P}(E_1=E_2=e)=P_E(e)^2$, yielding \eqref{eq:pue_dup_iid}.
The $r$-replica extension follows by the same argument: the all-agree test accepts iff $E_1=\cdots=E_r$, and acceptance is incorrect iff this common error is nonzero. The common-mode mixture formulas follow by conditioning on the mixture component.
\end{proof}

\paragraph{Specialization to the three-direction multiplicity MCU template}
Specializing Proposition~\ref{prop:dup_compare} to a word-level additive MCU law $E\sim P_E$ on $\{0,1\}^w$ with $P_E(0^w)=1-\alpha$ and, conditional on $(K=k,Q=q)$, uniform placement over $N_{k,q}$ nonzero patterns, gives the iid two-replica formulas
\begin{equation}
  p_{\mathrm{ue}}
  = \alpha^2 \sum_{k,q} \frac{P_{K,Q}(k,q)^2}{N_{k,q}},
  \qquad
  p_{\mathrm{ok}}=(1-\alpha)^2,
  \qquad
  p_{\mathrm{er}} = 1-(1-\alpha)^2-\alpha^2 \sum_{k,q}\frac{P_{K,Q}(k,q)^2}{N_{k,q}}.
  \label{eq:pue_dup_mcu}
\end{equation}
For $r\ge 2$ independent replicas,
\begin{equation}
  p_{\mathrm{ue}}
  = \alpha^r \sum_{k,q} \frac{P_{K,Q}(k,q)^r}{N_{k,q}^{\,r-1}},
  \qquad
  p_{\mathrm{ok}}=(1-\alpha)^r,
  \qquad
  p_{\mathrm{er}} = 1-p_{\mathrm{ok}}-p_{\mathrm{ue}},
  \label{eq:pue_r_replica_mcu}
\end{equation}
and under the common-mode mixture model of \eqref{eq:pue_dup_common},
\begin{equation}
  p_{\mathrm{ue}}
  = \theta\alpha + (1-\theta)\alpha^2 \sum_{k,q} \frac{P_{K,Q}(k,q)^2}{N_{k,q}}.
  \label{eq:pue_dup_mcu_common}
\end{equation}
Each display follows by substituting the MCU pmf into \eqref{eq:pue_dup_iid}, \eqref{eq:pue_r_replica}, and \eqref{eq:pue_dup_common}.

\begin{discussion}{Replication-based detection and fault domains}
Proposition~\ref{prop:dup_compare} quantifies a key systems lesson (e.g., Radshield~\cite{WangRadshieldASPLOS26}):
duplication-and-compare converts many silent corruptions into erasures, but its residual undetected-error probability depends critically on \emph{fault correlation}.
Under independent replicas, $p_{\mathrm{ue}}$ equals the \emph{collision probability} of the error-pattern law, $\sum_{e\neq 0}P_E(e)^2$, which can be much smaller than $p_{\mathrm{fail}}^2$ when the upset patterns are diverse.
Under common-mode upsets (shared fault domains), $p_{\mathrm{ue}}$ can instead be on the order of the single-replica failure probability $p_{\mathrm{fail}}=1-P_E(0)$; the mixture parameter $\theta$ in \eqref{eq:pue_dup_common} captures this inflation.
Thus, placing replicas in distinct fault domains (different cache slices/banks/SMs, etc.) aims to drive $\theta$ down.
Finally, checksum/hashing (Proposition~\ref{prop:hash_detection}) provides a complementary, designable route to control $p_{\mathrm{ue}}$ by allocating $h$ protected bits for tags/keys.
\end{discussion}

\begin{proposition}[Excess-distortion probability decomposition under $\mathsf{OK}/\mathsf{ER}/\mathsf{UE}$]
\label{prop:excess_distortion_decomp}
Fix any distortion threshold $D_0>0$ and recall the \emph{per-instance} excess-distortion probability in~\eqref{eq:excess_dist_prob_one},
\begin{equation}
  \epsilon_D^{(1)}(D_0)\triangleq \mathbb{P}\big(d(X,\hat X)>D_0\big).
\end{equation}
Under the same three-outcome model as Proposition~\ref{prop:erasure_fallback_mse},
\begin{equation}
  \epsilon_D^{(1)}(D_0)
  = p_{\mathrm{ok}}\,\epsilon^{(1)}_{\mathsf{OK}}(D_0)
  + p_{\mathrm{er}}\,\epsilon^{(1)}_{\mathsf{FB}}(D_0)
  + p_{\mathrm{ue}}\,\epsilon^{(1)}_{\mathsf{UE}}(D_0),
  \label{eq:tail_three_event}
\end{equation}
where
$\epsilon^{(1)}_{\mathsf{OK}}(D_0)\triangleq \mathbb{P}(d(X,\hat X^{(\mathsf{OK})})>D_0\mid \mathsf{OK})$
and similarly for $\epsilon^{(1)}_{\mathsf{FB}}(D_0)$ and $\epsilon^{(1)}_{\mathsf{UE}}(D_0)$.
\end{proposition}

\begin{proof}
This follows from the law of total probability over the partition $\{\mathsf{OK},\mathsf{ER},\mathsf{UE}\}$.
\end{proof}

\subsection{A Computable Tail Benchmark Template}
\label{subsec:tail_benchmark_template}

The decomposition in Proposition~\ref{prop:excess_distortion_decomp} is exact but still involves the conditional tail probabilities
$\epsilon^{(1)}_{\mathsf{OK}}(D_0)$, $\epsilon^{(1)}_{\mathsf{FB}}(D_0)$, and $\epsilon^{(1)}_{\mathsf{UE}}(D_0)$.
For system design, it is often useful to convert these into \emph{computable} upper/lower envelopes in terms of
the event probabilities $(p_{\mathrm{ue}},p_{\mathrm{er}})$ and coarse performance summaries of the $\mathsf{OK}$ and $\mathsf{ER}$ branches.

\begin{corollary}[A tail-probability sandwich in terms of $(p_{\mathrm{ue}},p_{\mathrm{er}})$]
\label{cor:tail_sandwich}
Fix any threshold $D_0>0$.
Suppose there exist numbers $\delta_{\mathsf{OK}},\delta_{\mathsf{FB}}\in[0,1]$ and $\beta_{\mathsf{UE}}\in[0,1]$ such that
\begin{equation}
  \epsilon^{(1)}_{\mathsf{OK}}(D_0)\le \delta_{\mathsf{OK}},
  \qquad
  \epsilon^{(1)}_{\mathsf{FB}}(D_0)\le \delta_{\mathsf{FB}},
  \qquad
  \epsilon^{(1)}_{\mathsf{UE}}(D_0)\ge \beta_{\mathsf{UE}}.
  \label{eq:tail_sandwich_assump}
\end{equation}
Then the per-instance excess-distortion probability satisfies
\begin{equation}
  \boxed{
  \beta_{\mathsf{UE}}\,p_{\mathrm{ue}}
  \ \le\
  \epsilon_D^{(1)}(D_0)
  \ \le\
  (1-p_{\mathrm{ue}}-p_{\mathrm{er}})\,\delta_{\mathsf{OK}}
  + p_{\mathrm{er}}\,\delta_{\mathsf{FB}}
  + p_{\mathrm{ue}}.
  }
  \label{eq:tail_sandwich}
\end{equation}
Consequently, the block excess-distortion probability obeys the computable upper bound
\begin{equation}
  \epsilon_D^{\mathrm{blk}}(D_0)\ \le\ T\,\epsilon_D^{(1)}(D_0)
  \ \le\
  T\Big((1-p_{\mathrm{ue}}-p_{\mathrm{er}})\,\delta_{\mathsf{OK}}
  + p_{\mathrm{er}}\,\delta_{\mathsf{FB}}
  + p_{\mathrm{ue}}\Big).
  \label{eq:tail_block_sandwich}
\end{equation}
\end{corollary}

\begin{proof}
The bounds follow by substituting \eqref{eq:tail_sandwich_assump} into the exact decomposition \eqref{eq:tail_three_event}
and using the union bound \eqref{eq:block_to_single_tail_bound}.
\end{proof}

\begin{corollary}[Worked example: an end-to-end computable excess-distortion tail under MCUs and $r$-replica detection]
\label{cor:tail_mcu_rreplica}
Fix any threshold $D_0>0$ and blocklength $T$.
Assume the hypotheses of Corollary~\ref{cor:tail_sandwich} hold with parameters $(\delta_{\mathsf{OK}},\delta_{\mathsf{FB}},\beta_{\mathsf{UE}})$.
Suppose further that a task-critical intermediate \emph{word} $U$ is protected by an $r$-replica duplication-and-compare detector operating at the word level under the MCU specialization summarized in \eqref{eq:pue_r_replica_mcu}, so that $(p_{\mathrm{ue}},p_{\mathrm{er}},p_{\mathrm{ok}})$ are given by that display.
Then the block excess-distortion probability admits the explicit upper bound
\begin{align}
  \epsilon_D^{\mathrm{blk}}(D_0)
  &\le
  T\Big(p_{\mathrm{ok}}\,\delta_{\mathsf{OK}}+p_{\mathrm{er}}\,\delta_{\mathsf{FB}}+p_{\mathrm{ue}}\Big)\nonumber\\
  &=
  T\Big(
    (1-\alpha)^r\delta_{\mathsf{OK}}
    + \big[1-(1-\alpha)^r-p_{\mathrm{ue}}\big]\delta_{\mathsf{FB}}
    + p_{\mathrm{ue}}
  \Big)\nonumber\\
  &=
  T\Big(
    (1-\alpha)^r\delta_{\mathsf{OK}}
    + \big[1-(1-\alpha)^r\big]\delta_{\mathsf{FB}}
    + (1-\delta_{\mathsf{FB}})\,\alpha^r \sum_{k,q}\frac{P_{K,Q}(k,q)^r}{N_{k,q}^{\,r-1}}
  \Big),
  \label{eq:tail_block_mcu_rreplica}
\end{align}
where the last line substitutes $p_{\mathrm{ue}}$ from \eqref{eq:pue_r_replica_mcu}.
In particular, if the nominal and fallback branches are $D_0$-safe in the sense that
$\epsilon^{(1)}_{\mathsf{OK}}(D_0)=\epsilon^{(1)}_{\mathsf{FB}}(D_0)=0$ (hence $\delta_{\mathsf{OK}}=\delta_{\mathsf{FB}}=0$),
then
\begin{equation}
  \epsilon_D^{\mathrm{blk}}(D_0)
  \ \le\
  T\,\alpha^r \sum_{k,q}\frac{P_{K,Q}(k,q)^r}{N_{k,q}^{\,r-1}}.
  \label{eq:tail_block_mcu_rreplica_safe}
\end{equation}
Consequently, to enforce a target block tail budget $\epsilon_D^{\mathrm{blk}}(D_0)\le \epsilon$, it suffices to choose $r$ such that
\begin{equation}
  \alpha^r \sum_{k,q}\frac{P_{K,Q}(k,q)^r}{N_{k,q}^{\,r-1}}
  \ \le\ \frac{\epsilon}{T}.
  \label{eq:r_for_tail_target}
\end{equation}
\end{corollary}

\begin{proof}
Combine Corollary~\ref{cor:tail_sandwich} with the $r$-replica MCU detection probabilities in \eqref{eq:pue_r_replica_mcu}.
\end{proof}

\paragraph{Markov-type computable upper bound (nonnegative distortion)}
Assume the distortion is nonnegative, $d(X,\hat X)\ge 0$.
Define the conditional mean distortions on the $\mathsf{OK}$ and erasure-fallback branches:
\begin{equation}
  \bar d_{\mathsf{OK}}\triangleq \mathbb{E}\!\big[d(X,\hat X^{(\mathsf{OK})})\mid \mathsf{OK}\big],
  \qquad
  \bar d_{\mathsf{FB}}\triangleq \mathbb{E}\!\big[d(X,\hat X^{(\mathsf{FB})})\mid \mathsf{ER}\big].
  \label{eq:conditional_mean_dist_ok_fb}
\end{equation}
Then for any threshold $D_0>0$,
\begin{equation}
  \epsilon_D^{(1)}(D_0)
  \ \le\
  p_{\mathrm{ue}}
  + \frac{(1-p_{\mathrm{ue}}-p_{\mathrm{er}})\,\bar d_{\mathsf{OK}} + p_{\mathrm{er}}\,\bar d_{\mathsf{FB}}}{D_0},
  \label{eq:tail_markov_template}
\end{equation}
and, if one also has a computable upper bound $\bar d_{\mathsf{UE}}^{\mathrm{ub}}$ on the conditional mean distortion under $\mathsf{UE}$,
\begin{equation}
  \epsilon_D^{(1)}(D_0)
  \ \le\
  \frac{(1-p_{\mathrm{ue}}-p_{\mathrm{er}})\,\bar d_{\mathsf{OK}} + p_{\mathrm{er}}\,\bar d_{\mathsf{FB}} + p_{\mathrm{ue}}\,\bar d_{\mathsf{UE}}^{\mathrm{ub}}}{D_0}.
  \label{eq:tail_markov_all_branches}
\end{equation}
Both displays follow by applying Markov's inequality to $d(X,\hat X)\mathbf{1}\{\mathsf{OK}\cup\mathsf{ER}\}$, or to $d(X,\hat X)$ directly in the second case. For the MSE distortion, Proposition~\ref{prop:clipping_mse} supplies $\bar d_{\mathsf{UE}}^{\mathrm{ub}}$ via clipping.

\begin{discussion}{Detection primarily acts by reducing $p_{\mathrm{ue}}$}
For MSE-type tasks, \eqref{eq:mse_three_event} shows that converting silent corruptions into erasures helps whenever the fallback branch is substantially safer than the undetected branch. For tail metrics, \eqref{eq:tail_three_event} and \eqref{eq:tail_markov_template} show that when $\epsilon^{(1)}_{\mathsf{UE}}(D_0)\approx 1$ over the safety-relevant range of $D_0$, the excess-distortion tail is often dominated by $p_{\mathrm{ue}}$.
In the common ``safe nominal / safe fallback'' regime where $\epsilon^{(1)}_{\mathsf{OK}}(D_0)\approx \epsilon^{(1)}_{\mathsf{FB}}(D_0)\approx 0$, one has $\epsilon_D^{(1)}(D_0)\approx p_{\mathrm{ue}}$ and therefore $\epsilon_D^{\mathrm{blk}}(D_0)\lesssim T\,p_{\mathrm{ue}}$. A useful design rule is thus $p_{\mathrm{ue}}\lesssim \epsilon/T$ for a target block tail budget $\epsilon$.
\end{discussion}

\begin{proposition}[Clipping bounds the conditional MSE on undetected-error events]
\label{prop:clipping_mse}
Assume $\mathbb{E}[X^2]<\infty$.
Let $\tilde X$ denote the (possibly corrupted) numerical output produced by the receiver, and define a \emph{clipped} output
\begin{equation}
  \hat X \triangleq \mathrm{clip}(\tilde X;A)\ \in[-A,A],
\end{equation}
where $\mathrm{clip}(\cdot;A)$ saturates its argument to $[-A,A]$.
Then, on the event $\mathsf{UE}$,
\begin{equation}
  D_{\mathsf{UE}}
  = \mathbb{E}\!\big[(X-\hat X)^2\mid \mathsf{UE}\big]
  \le 2\,\mathbb{E}[X^2\mid \mathsf{UE}] + 2A^2.
  \label{eq:clipping_bound_conditional}
\end{equation}
Consequently, the overall MSE satisfies
\begin{equation}
  D
  \le p_{\mathrm{ok}}\,D_{\mathsf{OK}} + p_{\mathrm{er}}\,D_{\mathsf{FB}}
  + p_{\mathrm{ue}}\big(2\,\mathbb{E}[X^2\mid \mathsf{UE}] + 2A^2\big).
  \label{eq:clipping_bound}
\end{equation}
\end{proposition}

\begin{proof}
On $\mathsf{UE}$, we only use the fact that $|\hat X|\le A$, so
$(X-\hat X)^2\le (|X|+A)^2\le 2X^2+2A^2$.
Taking conditional expectations given $\mathsf{UE}$ yields \eqref{eq:clipping_bound_conditional}, and substituting this into
\eqref{eq:mse_three_event} yields \eqref{eq:clipping_bound}.
\end{proof}

\begin{remark}[De-conditioning the UE second moment]
The conditional second moment in \eqref{eq:clipping_bound} can be removed without any independence assumption, since
$p_{\mathrm{ue}}\,\mathbb{E}[X^2\mid \mathsf{UE}]
= \mathbb{E}\!\big[X^2\mathbf{1}\{\mathsf{UE}\}\big]
\le \mathbb{E}[X^2]$.
More generally, for any $a>0$,
$\mathbb{E}[X^2\mathbf{1}\{\mathsf{UE}\}] \le \mathbb{E}[X^2\mathbf{1}\{|X|>a\}] + a^2 p_{\mathrm{ue}}$.
\end{remark}

\begin{corollary}[Clipping controls the excess-distortion tail on $\mathsf{UE}$ (MSE case)]
Under the squared-error distortion $d(x,\hat x)=(x-\hat x)^2$ and the assumptions of Proposition~\ref{prop:clipping_mse}, for any threshold $D_0>0$,
\begin{equation}
  \epsilon^{(1)}_{\mathsf{UE}}(D_0)
  =
  \mathbb{P}\!\big((X-\hat X)^2>D_0\mid \mathsf{UE}\big)
  \le \frac{D_{\mathsf{UE}}}{D_0}
  \le \frac{2\,\mathbb{E}[X^2\mid \mathsf{UE}] + 2A^2}{D_0},
  \label{eq:clipping_tail_bound}
\end{equation}
where the first inequality is Markov's inequality.
\end{corollary}

\begin{remark}[A computable special case (modelling suggestion)]
If $\mathsf{UE}$ is (approximately) independent of $X$ (e.g., corruption arises from post-processing faults largely independent
of the source realization), then $\mathbb{E}[X^2\mid \mathsf{UE}]=\mathbb{E}[X^2]$ and the $\mathsf{UE}$ contribution in
\eqref{eq:clipping_bound} becomes $p_{\mathrm{ue}}(2\mathbb{E}[X^2]+2A^2)$.
If $X$ is almost surely bounded, $|X|\le B$, then $D_{\mathsf{UE}}\le (B+A)^2$ and the $\mathsf{UE}$ term is at most $p_{\mathrm{ue}}(B+A)^2$.
\end{remark}

Proposition~\ref{prop:clipping_mse} shows that clipping acts as a damage limiter: even if undetected errors cannot be fully eliminated, bounding the numerical range of the output (or key intermediates) prevents rare fault-induced outliers, such as sign/exponent corruption in floating-point representations, from dominating the MSE. In addition, even under hard-separation, a practical hybrid rule is to fall back to a low-complexity task-direct estimate from $R_t^n$ whenever the decoder declares failure, thereby preserving some task-relevant information without requiring the full digital interface to be trusted.


\section{Supplementary Throughput Limits and Redundancy Allocation}
\label{sec:throughput}

\indent This section translates per-second budgets into the per-instance parameters of the earlier theory.
Let $\mathcal{B}$ denote the available channel uses per second and $\mathcal{G}$ the available noisy primitive uses per second.
At throughput $\lambda$ samples per second, the per-sample resources are $n(\lambda)=\mathcal{B}/\lambda$ and $m(\lambda)=\mathcal{G}/\lambda$, and the per-second supplies (bits/sec) are
$S_{\mathrm{ch}}^{(\mathrm{sec})}\triangleq \mathcal{B}\Cch$ and $S_{\mathrm{comp}}^{(\mathrm{sec})}\triangleq \mathcal{G}\Cgate$.

\subsection{Closed-Form Maximum Throughput for a Target Distortion (Gaussian Benchmark)}
For scalar Gaussian remote estimation, Corollary~\ref{cor:gaussian_closed_form} (and Corollary~\ref{cor:gaussian_mcu} under word-level MCUs) gives
\begin{equation}
  D(\lambda)\ \ge\ \sigma_{x|y}^2 + \big(\sigma_x^2-\sigma_{x|y}^2\big)\,
  2^{-\frac{2}{\lambda}\min\{S_{\mathrm{ch}}^{(\mathrm{sec})},\,S_{\mathrm{comp}}^{(\mathrm{sec})}\}}.
  \label{eq:D_lambda}
\end{equation}

\begin{corollary}[Closed-form $\lambda_{\max}(D)$ for estimation]
\label{cor:lmax_est}
Fix a target distortion $D$ with $\sigma_{x|y}^2 < D < \sigma_x^2$.
Any task-direct design achieving $D$ must satisfy
\begin{equation}
  \boxed{
  \lambda\ \le\ \lambda_{\max}(D)
  \triangleq
  \frac{2\,\min\{S_{\mathrm{ch}}^{(\mathrm{sec})},\,S_{\mathrm{comp}}^{(\mathrm{sec})}\}}
  {\log_2\!\left(\frac{\sigma_x^2-\sigma_{x|y}^2}{D-\sigma_{x|y}^2}\right)}.
  }
\end{equation}
\end{corollary}

Under hard-separation, replace $S_{\mathrm{comp}}^{(\mathrm{sec})}$ by $S_{\mathrm{comp}}^{(\mathrm{sec})}/2$ in Corollary~\ref{cor:lmax_est} due to Theorem~\ref{thm:interface_tax}, yielding the familiar factor-of-two penalty in the symmetric two-stage compute-limited regime.

\subsection{A Tail-Aware Safe-Throughput Worked Example: Protecting a Digital Interface With $r$ Replicas Under MCUs}
\label{subsec:tail_throughput_worked_example}
The preceding throughput and allocation results focus on mean distortion benchmarks.
In many safety-critical settings, however, one must also control \emph{tail} events (Appendix~\ref{sec:erasures}).
We now connect the MCU+$r$-replica detection formulas for $(p_{\mathrm{ue}},p_{\mathrm{er}})$ to an \emph{end-to-end computable} tail benchmark and make explicit the (inevitable) throughput--reliability tradeoff induced by replication overhead.

\paragraph{Protecting a length-$\Lif$ interface message by word-level replication}
Suppose each task instance must materialize a task-critical digital interface message (or cached state) $B\in\{0,1\}^{\Lif}$ bits.
Partition $B$ into $M_{\mathrm{if}}\triangleq \Lif/w$ words of $w$ bits (assume $w\mid \Lif$ for simplicity).
Protect each word by an $r$-replica ``all-agree'' duplication-and-compare detector (Proposition~\ref{prop:dup_compare}), where each replica experiences an independent word-level MCU under the specialization summarized in \eqref{eq:pue_r_replica_mcu}.

Let $(p_{\mathrm{ok}}^{(w)},p_{\mathrm{ue}}^{(w)},p_{\mathrm{er}}^{(w)})$ denote the \emph{per-word} outcome probabilities.
Under the same word-level MCU specialization,
\begin{equation}
  p_{\mathrm{ok}}^{(w)}=(1-\alpha)^r,
  \qquad
  p_{\mathrm{ue}}^{(w)}=\alpha^r\sum_{k,q}\frac{P_{K,Q}(k,q)^r}{N_{k,q}^{\,r-1}},
  \qquad
  p_{\mathrm{er}}^{(w)}=1-p_{\mathrm{ok}}^{(w)}-p_{\mathrm{ue}}^{(w)}.
  \label{eq:per_word_ok_ue_er_mcu_rrep}
\end{equation}

Assume the $M_{\mathrm{if}}$ words experience independent MCU patterns (across words) and that the interface is accepted iff \emph{all} words are accepted by their local all-agree tests.
Then the \emph{message-level} outcome probabilities $(p_{\mathrm{ok}},p_{\mathrm{ue}},p_{\mathrm{er}})$ are explicit:
\begin{align}
  p_{\mathrm{ok}}
  &= \big(p_{\mathrm{ok}}^{(w)}\big)^{M_{\mathrm{if}}},
  \label{eq:pok_message_from_word}\\
  p_{\mathrm{ue}}
  &= \big(p_{\mathrm{ok}}^{(w)}+p_{\mathrm{ue}}^{(w)}\big)^{M_{\mathrm{if}}} - \big(p_{\mathrm{ok}}^{(w)}\big)^{M_{\mathrm{if}}},
  \label{eq:pue_message_from_word}\\
  p_{\mathrm{er}}
  &= 1-\big(p_{\mathrm{ok}}^{(w)}+p_{\mathrm{ue}}^{(w)}\big)^{M_{\mathrm{if}}}.
  \label{eq:per_message_er_from_word}
\end{align}

\paragraph{A computable block tail bound}
Plugging \eqref{eq:pok_message_from_word}--\eqref{eq:per_message_er_from_word} into Corollary~\ref{cor:tail_sandwich} yields an explicit upper bound on the block excess-distortion probability $\epsilon_D^{\mathrm{blk}}(D_0)$.
In the particularly transparent $D_0$-safe regime (cf.~Corollary~\ref{cor:tail_mcu_rreplica}), where the nominal and fallback branches satisfy $\epsilon^{(1)}_{\mathsf{OK}}(D_0)=\epsilon^{(1)}_{\mathsf{FB}}(D_0)=0$,
\begin{equation}
  \epsilon_D^{\mathrm{blk}}(D_0)
  \ \le\
  T\,p_{\mathrm{ue}}
  =
  T\Big[\big(p_{\mathrm{ok}}^{(w)}+p_{\mathrm{ue}}^{(w)}\big)^{M_{\mathrm{if}}} - \big(p_{\mathrm{ok}}^{(w)}\big)^{M_{\mathrm{if}}}\Big],
  \label{eq:tail_block_interface_message_rreplica}
\end{equation}
where $p_{\mathrm{ok}}^{(w)}$ and $p_{\mathrm{ue}}^{(w)}$ are given explicitly by \eqref{eq:per_word_ok_ue_er_mcu_rrep}.
Thus the tail constraint $\epsilon_D^{\mathrm{blk}}(D_0)\le \epsilon$ is enforceable by choosing $r$ such that
\begin{equation}
  \big(p_{\mathrm{ok}}^{(w)}+p_{\mathrm{ue}}^{(w)}\big)^{M_{\mathrm{if}}} - \big(p_{\mathrm{ok}}^{(w)}\big)^{M_{\mathrm{if}}}
  \ \le\ \frac{\epsilon}{T}.
  \label{eq:r_for_tail_target_interface_message_exact}
\end{equation}
If one prefers a simpler conservative sizing rule, the union bound gives $p_{\mathrm{ue}}\le M_{\mathrm{if}}p_{\mathrm{ue}}^{(w)}$, reducing \eqref{eq:r_for_tail_target_interface_message_exact} to the per-word condition in \eqref{eq:r_for_tail_target} with $\epsilon$ replaced by $\epsilon/M_{\mathrm{if}}$.

\paragraph{Replication overhead induces an explicit throughput penalty}
Storing $r$ replicas of the length-$\Lif$ interface costs $r\Lif$ vulnerable bit materializations per task instance.
Relative to storing the interface once, the additional overhead is $(r-1)\Lif$ materializations per instance.
Under a per-second compute budget $\mathcal{G}$ and throughput $\lambda$, a conservative accounting is therefore
\begin{equation}
  m_{\mathrm{eff}}(\lambda,r)
  \triangleq
  \frac{\mathcal{G}}{\lambda} - (r-1)\Lif
\end{equation}
available for the remaining receiver computation (decoding and/or task inference).
Applying the first-order supply--demand benchmark with $m$ replaced by $m_{\mathrm{eff}}$ yields the design constraint
\begin{equation}
  R_{X|Y}(D)
  \ \le\
  \min\left\{\frac{\mathcal{B}}{\lambda}\Cch,\ \Big(\frac{\mathcal{G}}{\lambda}-(r-1)\Lif\Big)\Cgate\right\}.
  \label{eq:supply_demand_with_replica_overhead}
\end{equation}
In the compute-limited regime, \eqref{eq:supply_demand_with_replica_overhead} is equivalent to the explicit throughput bound
\begin{equation}
  \lambda
  \ \le\
  \frac{\mathcal{G}}{(r-1)\Lif + R_{X|Y}(D)/\Cgate},
  \label{eq:lambda_replica_overhead}
\end{equation}
which shows explicitly how ``more replicas'' (larger $r$) improve tail reliability (smaller $p_{\mathrm{ue}}$) but reduce the maximum sustainable throughput.
For Gaussian remote estimation, substitute the closed-form $R_{X|Y}(D)$ from Corollary~\ref{cor:gaussian_closed_form}.


\section{Supplementary Discussion and Extensions}
\label{sec:discussion}

\subsection{Compute-Supply Plug-Ins Beyond the BSC}
\label{subsec:compute_supply_plugins}
The BSC$(\varepsilon)$ bit-flip primitive in Definition~\ref{def:gpu_bitflip_model} provides a tractable benchmark that makes the compute bottleneck explicit, but the converse chain only uses a plug-in rule: each vulnerable materialization contributes at most some bounded conditional mutual information $C_i$, so Lemmas~\ref{lem:compute_info_bound} and~\ref{lem:compute_info_bound_hetero} give $I(R^{nT};\hat X^T)\le \sum_i C_i$. Consequently, all supply--demand and interface-tax converses remain valid after replacing $m\Cgate$ by the effective compute supply $S_{\mathrm{comp}}\triangleq \sum_i C_i$ (or a conservative upper bound), covering word-level additive-noise primitives and heterogeneous reliable-island designs as examples. A more conservative fully noisy-logic extension, in which the receiver logic itself is unreliable, is developed in Subsection~\ref{subsec:noisy_gate_connections}.

\subsection{Toward Learned Inference on Noisy Accelerators}
\label{subsec:nn_discussion}
Beyond squared-error estimation, the same supply--demand viewpoint extends to decision and learned-inference tasks once the task demand and effective compute supply are replaced by the appropriate problem-specific quantities. Proposition~\ref{prop:fano_classification} in Appendix~\ref{app:section3_supp} records a simple Fano-style classification benchmark, while Appendix~\ref{app:learned_inference_supp} records supplementary derivations on model size/precision, MSE decomposition, and reliable parameter storage for learned inference on noisy accelerators. The same structural lesson remains: hard digital interfaces introduce additional mandatory bottlenecks, whereas heterogeneous protection or reliable islands enlarge the effective compute supply.

\section{Supplementary Plug-Ins, Task Variants, and Benchmark Coding Ingredients}
\label{app:section3_supp}
\indent This appendix collects a few supporting plug-ins and task variants for the baseline converse, outside the paper's main architecture theorem chain.

\subsection{Heterogeneous Compute-Supply Plug-In}
\label{subsec:hetero_compute_plugin}
\begin{lemma}[Compute information bound with heterogeneous primitives]
\label{lem:compute_info_bound_hetero}
Suppose the receiver interacts with $J$ classes of unreliable primitives.
Class $j\in\{1,\dots,J\}$ is modeled as a DMC $W^{(j)}_{Z|U}$ with Shannon capacity $\Cgate^{(j)}$ (bits/use).
Over a block of length $T$, suppose the receiver uses at most $m_jT$ primitives of class $j$ (in any adaptive order), where $m_j\ge 0$.
Then the final output $\hat X^T$ satisfies
\begin{equation}
  \frac{1}{T}I(R^{nT};\hat X^T)
  \le \sum_{j=1}^J m_j\,\Cgate^{(j)}.
  \label{eq:compute_info_bound_hetero}
\end{equation}
\end{lemma}

\begin{proof}
Index the total number of primitive uses over the block by $i\in\{1,\dots,M\}$, where $M\le (\sum_{j=1}^J m_j)T$, and let $j(i)$ denote the class used at the $i$-th primitive use.
Let $Z_i$ denote the primitive output at use $i$, and $Z^M$ the collection of all outputs.
By data processing, $I(R^{nT};\hat X^T)\le I(R^{nT};Z^M)$.
Applying the chain rule,
\begin{equation}
  I(R^{nT};Z^M)=\sum_{i=1}^M I(R^{nT};Z_i\mid Z^{i-1}).
\end{equation}
Conditioned on $Z^{i-1}$, the intended primitive input bit/word $U_i$ is a function of $(R^{nT},Z^{i-1})$, and the primitive implements the DMC $W^{(j(i))}_{Z|U}$.
Thus, conditioned on $Z^{i-1}$, we have the Markov chain $R^{nT}\to U_i\to Z_i$, so
\begin{equation}
  I(R^{nT};Z_i\mid Z^{i-1})\le I(U_i;Z_i\mid Z^{i-1})\le \Cgate^{(j(i))}.
\end{equation}
Summing over $i$ and grouping terms by class yields
$I(R^{nT};Z^M)\le \sum_{j=1}^J (m_jT)\,\Cgate^{(j)}$, which proves \eqref{eq:compute_info_bound_hetero}.
\end{proof}

Lemma~\ref{lem:compute_info_bound_hetero} shows that the heterogeneous compute bottleneck is captured by the scalar
\begin{equation}
  S_{\mathrm{comp}} \triangleq \sum_{j=1}^J m_j\,\Cgate^{(j)}\quad \text{bits/sample},
\end{equation}
which reduces to $m\Cgate$ in the homogeneous BSC$(\varepsilon)$ model. All supply--demand converses based on the memoryless compute-primitive abstraction therefore remain valid with $m\Cgate$ replaced by $S_{\mathrm{comp}}$, and the same plug-in also covers reliable islands or unequal-protection strategies in which task-critical bits are preferentially assigned to higher-capacity primitives~\cite{ZhengSAVEATC25}.

\subsection{Word-Level Compute Plug-In}
\label{subsec:word_level_compute_plugin}
\begin{corollary}[Supply--demand and hard-separation converses under word-level MCU primitives]
\label{cor:mcu_plugin}
Consider a word-level additive compute primitive
$Z^{(w)}=U^{(w)}\oplus E$
with $E\sim P_E$ on $\{0,1\}^w$.
Assume the receiver compute budget counts vulnerable \emph{bits}, so that at most $m$ bits are materialized per task instance (equivalently at most $m/w$ word uses).
Then the conclusions of Theorem~\ref{thm:supply_demand_converse} and Theorem~\ref{thm:interface_tax} remain valid after replacing $\Cgate$ by the per-bit effective capacity
\begin{equation}
  \Cgateeff \triangleq \frac{1}{w}\max_{P_{U^{(w)}}}I(U^{(w)};Z^{(w)})
  \;=\;
  1-\frac{H(E)}{w}.
  \label{eq:Ceff_mcu}
\end{equation}
In particular, any task-direct scheme achieving distortion $D$ must satisfy
\begin{equation}
  R_{X|Y}(D)\ \le\ \min\{n\Cch,\ m\Cgateeff\},
\end{equation}
and any hard-separation scheme (Section~\ref{sec:interface_tax_section}) must satisfy
\begin{equation}
  \frac{1}{T}I(Y^T;\hat X^T)\ \le\ \min\{n\Cch,\ m_{\mathrm{dec}}\Cgateeff,\ m_{\mathrm{task}}\Cgateeff\}.
\end{equation}

\end{corollary}

\begin{proof}
Treat each $w$-bit word materialization as one DMC use with capacity $C_{\mathrm{gate}}^{(w)}=w-H(E)$.
(For additive noise, the uniform input is capacity-achieving, so $\max_{P_{U^{(w)}}}I(U^{(w)};Z^{(w)})=w-H(E)$.)
Since one word consumes $w$ units of the bit-budget, $m$ bit materializations correspond to at most $m/w$ word uses, yielding a total compute supply
$\frac{m}{w}C_{\mathrm{gate}}^{(w)}=m\Cgateeff$.
Applying Lemma~\ref{lem:compute_info_bound} (or Lemma~\ref{lem:compute_info_bound_hetero}) with this per-use capacity bound and then repeating the cut-set steps in the proofs of Theorems~\ref{thm:supply_demand_converse} and~\ref{thm:interface_tax} yields the stated replacement.
\end{proof}

\subsection{Practical Computation-Noise Primitives and Calibration}
\label{subsec:practical_noise_models}
For the converse chain, a practical primitive enters only through a per-use information-supply bound $C_i$, so the homogeneous term $m\Cgate$ is replaced by the plug-in supply $\sum_i C_i$ (equivalently $\sum_j m_j\Cgate^{(j)}$ under Lemma~\ref{lem:compute_info_bound_hetero}); correlated or state-dependent fault mechanisms are therefore covered whenever each vulnerable materialization admits such an upper bound, and word-level MCUs reduce to the effective per-bit quantity $\Cgateeff$ in Corollary~\ref{cor:mcu_plugin}. As a calibration convention, $(m,\Cgate)$ are operating-point abstractions: $\Cgate$ summarizes the effective information supply of each vulnerable commit--retrieve event after low-level protection, while $m$ is obtained by counting the vulnerable bits or words written and later used per task instance, including intermediate activations and any per-inference weight re-materializations not placed on a reliable island.
For the BSC baseline, $\varepsilon$ may be estimated by measurement, fault injection, or a soft-error-rate mapping; for example, a Poisson upset model with per-bit rate $\lambda_{\mathrm{SEU}}$ and dwell time $\tau$ gives $\varepsilon\approx 1-e^{-\lambda_{\mathrm{SEU}}\tau}\approx \lambda_{\mathrm{SEU}}\tau$ when $\lambda_{\mathrm{SEU}}\tau\ll 1$.
Table~\ref{tab:bm_accounting} summarizes representative accounting conventions for common architecture-level artifacts.

\begin{table}[t]
\centering
\caption{Representative accounting conventions for common artifact classes. The table is schematic; the surrounding text and Appendix~\ref{app:learned_inference_supp} provide architecture-specific refinements.}
\label{tab:bm_accounting}
\small
\begin{tabular}{p{0.30\linewidth}p{0.08\linewidth}p{0.53\linewidth}}
\toprule
Artifact class & Charge & Interpretation \\
\midrule
Protected/raw-output side information retained to the final task output & $b$ & Reliable bypass or other protected side path; counts toward the explicit protected side-information budget. \\
Vulnerable intermediate activations or decoded interface state that is later reused & $m$ & Each vulnerable commit--retrieve event contributes to the per-instance materialization budget. \\
Per-inference weight reloads through vulnerable memory & $m$ & Re-materialized parameters are charged like other vulnerable writes that are later used downstream. \\
Reliably stored / island-protected parameters or state & not in $m$ & No vulnerable per-instance rematerialization is charged once these artifacts are kept on protected storage. \\
Cached soft statistics and non-negligible support overhead for a soft path & $m_{\mathrm{int}}$ & Reserved interface-support budget covering stored soft state and any routing / buffering / forwarding support that should not be treated as free. \\
\bottomrule
\end{tabular}
\end{table}

\subsection{A Classification Benchmark via Fano's Inequality}
\label{subsec:classification_benchmark}

\begin{proposition}[A Fano-type classification lower bound under dual noise]
\label{prop:fano_classification}
Let $J\in\{1,\dots,M\}$ be a label with $H(J)=\log_2 M$ (e.g., $J$ uniform), and let the receiver output $\hat J$.
For the task-direct architecture (A), the classification error probability $P_e\triangleq \mathbb{P}[\hat J\neq J]$ satisfies
\begin{equation}
  P_e \ \ge\ \max\!\left\{0,\ 1-\frac{\Rsup+1}{\log_2 M}\right\}.
  \label{eq:fano_bound_dual_noise}
\end{equation}
Under hard-separation (B), the same bound holds with $\Rsup$ replaced by $\min\{n\Cch,\ \frac{m}{2}\Cgate\}$.
\end{proposition}

\begin{proof}
By Fano's inequality~\cite{CoverThomas}, $P_e \ge 1-\frac{I(J;\hat J)+1}{\log_2 M}$.
Since $J\to Y\to \hat J$ forms a Markov chain, $I(J;\hat J)\le I(Y;\hat J)$.
For architecture (A), data processing and Lemma~\ref{lem:compute_info_bound} imply
$I(Y;\hat J)\le \Rsup$ by the same cut-set argument as in Theorem~\ref{thm:supply_demand_converse}.
For architecture (B), replace this by Theorem~\ref{thm:interface_tax} (or Corollary~\ref{cor:hard_sep_supply_demand}).
Substituting into Fano's inequality and clipping at zero yields \eqref{eq:fano_bound_dual_noise}.
\end{proof}

\begin{corollary}[Depth-induced classification lower bound under a $K$-stage serial pipeline]
Let $J\in\{1,\dots,M\}$ be a discrete label with $H(J)=\log_2 M$, and let $\hat J$ denote the receiver decision.
Suppose the receiver can be abstracted as a $K$-stage serial pipeline as in Theorem~\ref{thm:k_stage_tax}, with per-instance stage budgets $(m_1,\dots,m_K)$ and total budget $\sum_{k=1}^K m_k\le m$.
Then the classification error probability $P_e\triangleq \mathbb{P}[\hat J\neq J]$ satisfies
\begin{equation}
  P_e \ \ge\ \max\!\left\{0,\ 1 - \frac{\min\{\,n\Cch,\ \min_{k=1,\dots,K} m_k\Cgate\,\}+1}{\log_2 M}\right\}.
  \label{eq:depth_classification_bound}
\end{equation}
In particular, under equal splitting $m_k=m/K$,
\begin{equation}
  P_e \ \ge\ \max\!\left\{0,\ 1 - \frac{\min\{\,n\Cch,\ \frac{m}{K}\Cgate\,\}+1}{\log_2 M}\right\}.
\end{equation}
\end{corollary}

\begin{proof}
By the Markov chain $J\to Y \to \hat J$, data processing gives $I(J;\hat J)\le I(Y;\hat J)$.
Specializing Theorem~\ref{thm:k_stage_tax} to a single instance ($T=1$) with final output $\hat J$ yields
$I(Y;\hat J)\le \min\{\,n\Cch,\ \min_k m_k\Cgate\,\}$.
Applying the same Fano bound as in Proposition~\ref{prop:fano_classification} gives
$P_e \ge 1 - \frac{I(J;\hat J)+1}{\log_2 M}$, and clipping at zero gives the claim.
\end{proof}

\subsection{A Shannon-Style Internal Coding Ingredient for the Stronger Protected-Support Closure}
\label{subsec:benchmark_internal_coding}
\begin{proposition}[Capacity-achieving internal redundancy for a memoryless compute primitive]
\label{prop:compute_code_bsc}
Consider a task block of length $T$ and let $L\triangleq mT$ denote the total number of unreliable primitive uses available to the receiver in the block.
Suppose each primitive use is a (possibly nonbinary) discrete memoryless channel (DMC) $W_{Z|U}$ with Shannon capacity $\Cgate$,
and that the $L$ primitive uses are independent across time given their inputs (as in Definition~\ref{def:gpu_bitflip_model} for the BSC$(\varepsilon)$ case).
Then, for any \emph{per-sample} rate $R<m\Cgate$, there exists a sequence of internal codes that enables reliable storage/communication of
$\lfloor TR\rfloor$ bits across the $L=mT$ primitive uses with block error probability tending to zero as $T\to\infty$.
For the BSC$(\varepsilon)$ bit-flip primitive, $\Cgate=1-\hbin{\varepsilon}$.
\end{proposition}
\begin{proof}
Let the internal message set have size $M=2^{\lfloor TR\rfloor}$.
Using the $L=mT$ primitive uses, the induced \emph{per-primitive} code rate equals
$\frac{\log_2 M}{L}\approx \frac{TR}{mT}=\frac{R}{m}$ bits/primitive.
The condition $R<m\Cgate$ is equivalent to $\frac{R}{m}<\Cgate$, so Shannon's channel coding theorem for the DMC $W_{Z|U}$ yields a sequence of length-$L$ codes
with vanishing block error probability as $L\to\infty$ (equivalently, as $T\to\infty$); see~\cite{Shannon48,CoverThomas}.
The primitive-interactivity allowed by Definition~\ref{def:gpu_bitflip_model} corresponds to channel coding with feedback, which does not increase capacity for memoryless channels and is unnecessary for this first-order statement.
\end{proof}

\begin{remark}[Instantiation of internal redundancy mechanisms]
Proposition~\ref{prop:compute_code_bsc} is the first-order existence statement used in the achievability theorems under Assumption~\ref{assump:achievability_benchmark}. Appendix~\ref{app:internal_redundancy_instantiation} records one explicit LDPC/stable-memory instantiation, but the main converses do not depend on any particular internal redundancy mechanism.
\end{remark}

Some physical implementation models couple reliability to \emph{energy} rather than to an abstract primitive budget. This does not conflict with our benchmark abstraction: we treat $m$ (and the primitive noise level entering $\Cgate$) as operating-point parameters that already absorb any energy/latency/area constraints. If a physical model yields a relationship such as $m=m(E_{\mathrm{comp}})$ or $\varepsilon=\varepsilon(E_{\mathrm{comp}})$, then substituting $m(E_{\mathrm{comp}})\Cgate(E_{\mathrm{comp}})$ into our converses gives the corresponding energy--task tradeoffs.

\section{Worked Example: Vector Gaussian Uncoded Baseline Versus Hard-Separation}
\label{app:vector_gaussian_worked_example}
\indent This appendix records explicit diagonal examples that contrast a simple task-direct \emph{uncoded} linear baseline with the water-filling indirect rate--distortion converse under task-direct and hard-separation organizations. We first give a fully closed-form $p=2$ example that can be checked by hand, and then a higher-dimensional ($p=8$) illustration with a plotted gap (Fig.~\ref{fig:vector_gap_uncoded_hard}).

\paragraph{A simple uncoded linear baseline}
Assume $\Sigma_X$ and $\Sigma_V$ are diagonal so that the sufficient statistic $\tilde X=\mathbb{E}[X|Y]$ has independent coordinates with variances $\{\lambda_i\}_{i=1}^p$ (Subsection~\ref{subsec:vector_gaussian}).
Consider a bandwidth-matched $p$-use real AWGN channel with independent noise $W_i\sim\mathcal{N}(0,\sigma_w^2)$ and per-use power constraint $P$, so each use has capacity $\Cch=\frac12\log_2(1+P/\sigma_w^2)$.
A coordinatewise uncoded linear (analog) strategy transmits $\tilde X$ with matched scaling $S_i=\sqrt{P/\lambda_i}\,\tilde X_i$ and the receiver forms the coordinatewise MMSE estimate from $Z_i=S_i+W_i$.
In the jointly Gaussian model, $X=\tilde X+(X-\tilde X)$ with $(X-\tilde X)\perp \tilde X$, yielding the end-to-end MSE
\begin{equation}
  D_{\mathrm{uncoded}}
  \;=\;
  \mathrm{tr}(\Sigma_{X|Y}) + \frac{\sigma_w^2}{P+\sigma_w^2}\,\mathrm{tr}(\Sigma_{\tilde X})
  \;=\;
  \mathrm{tr}(\Sigma_{X|Y}) + \mathrm{tr}(\Sigma_{\tilde X})\,2^{-2\Cch}.
  \label{eq:vector_uncoded_mse}
\end{equation}
We do not claim this simple symbol-by-symbol \emph{linear} baseline is optimal among zero-delay analog mappings; its purpose is to provide a transparent achievable point for comparison with the hard-separation converse.

\subsection{A Fully Closed-Form $p=2$ Diagonal Example}
Take $p=2$ with diagonal $\Sigma_X=\mathrm{diag}(\sigma_{x,1}^2,\sigma_{x,2}^2)$ and $\Sigma_V=\mathrm{diag}(\sigma_{v,1}^2,\sigma_{v,2}^2)$.
Let $\lambda_1\ge \lambda_2$ denote the eigenvalues of $\Sigma_{\tilde X}$, so $\lambda_i=\frac{\sigma_{x,i}^4}{\sigma_{x,i}^2+\sigma_{v,i}^2}$ for $i=1,2$, and define the activation threshold $R_0 \triangleq \frac12\log_2\!\big(\frac{\lambda_1}{\lambda_2}\big)$.
Then the water-filling characterization \eqref{eq:vector_gaussian_waterfill} reduces to the explicit piecewise form
\[
  D(R)=\mathrm{tr}(\Sigma_{X|Y})+\begin{cases}
    \lambda_2 + \lambda_1\,2^{-2R}, & 0\le R\le R_0,\\[2pt]
    2\sqrt{\lambda_1\lambda_2}\,2^{-R}, & R\ge R_0,
  \end{cases}
\]
The first branch allocates rate only to the stronger mode; the second activates both modes.

As a concrete instantiation, take $(\sigma_{x,1}^2,\sigma_{x,2}^2)=(4,1)$ and $(\sigma_{v,1}^2,\sigma_{v,2}^2)=(1,1)$, so
$\lambda_1=\frac{16}{5}$, $\lambda_2=\frac12$, and $\mathrm{tr}(\Sigma_{X|Y})=\frac{4}{5}+\frac12=1.3$, with $R_0\approx 1.34$ bits/vector.
Let the physical channel be bandwidth-matched ($n=p=2$) with per-use capacity $\Cch=1$ bit/use, hence $n\Cch=2$ bits/vector, and take a compute supply $R_{\mathrm{comp}}=m\Cgate=2$ bits/vector.
Then
\[
  R_{\mathrm{eff}}=\min\{R_{\mathrm{comp}},n\Cch\}=2,
  \qquad
  R_{\mathrm{eff}}^{(\mathrm{hard})}=\min\{\tfrac12 R_{\mathrm{comp}},n\Cch\}=1,
\]
so Corollary~\ref{cor:vector_gaussian_dual_noise} yields the no-tax converse $D\ge D(2)\approx 1.93$ and the hard-separation converse $D\ge D(1)=2.6$.
Meanwhile, the uncoded linear baseline \eqref{eq:vector_uncoded_mse} gives
\[
  D_{\mathrm{uncoded}}
  = 1.3 + (\lambda_1+\lambda_2)\,2^{-2\Cch}
  = 1.3 + 3.7\cdot 2^{-2}
  \approx 2.23.
\]
Thus $D_{\mathrm{uncoded}}<D(1)$ yields a concrete MSE gap that \emph{cannot} be closed by any hard-separation pipeline under the same $(n\Cch,m\Cgate)$ supplies, while still leaving room relative to the no-tax benchmark $D(2)$.
Finally, this example also shifts the transition to channel-limited behavior: the no-tax bound saturates once $R_{\mathrm{comp}}\ge n\Cch=2$, whereas under hard-separation saturation requires $R_{\mathrm{comp}}\ge 2n\Cch=4$.

\subsection{An Illustrative Diagonal $p=8$ Example and a Plotted Gap}
To visualize the same mechanism in a higher-dimensional diagonal spectrum, consider the $p=8$ example shown in Fig.~\ref{fig:vector_gap_uncoded_hard}, with $\Sigma_X=\mathrm{diag}(8,4,2,1,\tfrac12,\tfrac14,\tfrac18,\tfrac1{16})$, $\Sigma_V=I$, and $\Cch=2$ bits/use, so with $n=p$ we have $n\Cch=16$ bits per vector.
Let $R_{\mathrm{comp}}=m\Cgate$ denote the \emph{compute} supply (bits/vector).
Then the effective task-direct information supply is $R_{\mathrm{eff}}=\min\{R_{\mathrm{comp}},n\Cch\}$, whereas under hard-separation it is
$R_{\mathrm{eff}}^{(\mathrm{hard})}=\min\{\tfrac12 R_{\mathrm{comp}},n\Cch\}$.
Accordingly, Corollary~\ref{cor:vector_gaussian_dual_noise} yields the task-direct converse $D\ge D(R_{\mathrm{eff}})$, whereas hard-separation yields $D\ge D(R_{\mathrm{eff}}^{(\mathrm{hard})})$.
For these parameters, \eqref{eq:vector_uncoded_mse} gives $D_{\mathrm{uncoded}}\approx 4.33$.
At $R_{\mathrm{comp}}=8$ bits/vector (compute-limited for both), the hard-separation converse gives $D\ge D(4)\approx 5.77$, leaving an explicit $\approx 1.43$ MSE gap in a vector setting ($p>1$).
Moreover, once $R_{\mathrm{comp}}$ exceeds the channel limit $n\Cch=16$, the no-tax converse saturates at $D(n\Cch)$ (channel-limited), whereas hard-separation continues to improve until $R_{\mathrm{comp}}=2n\Cch$, visibly shifting the compute-limited/channel-limited transition.

\begin{figure}[!tbp]
\centering
\includefigure[width=0.66\linewidth]{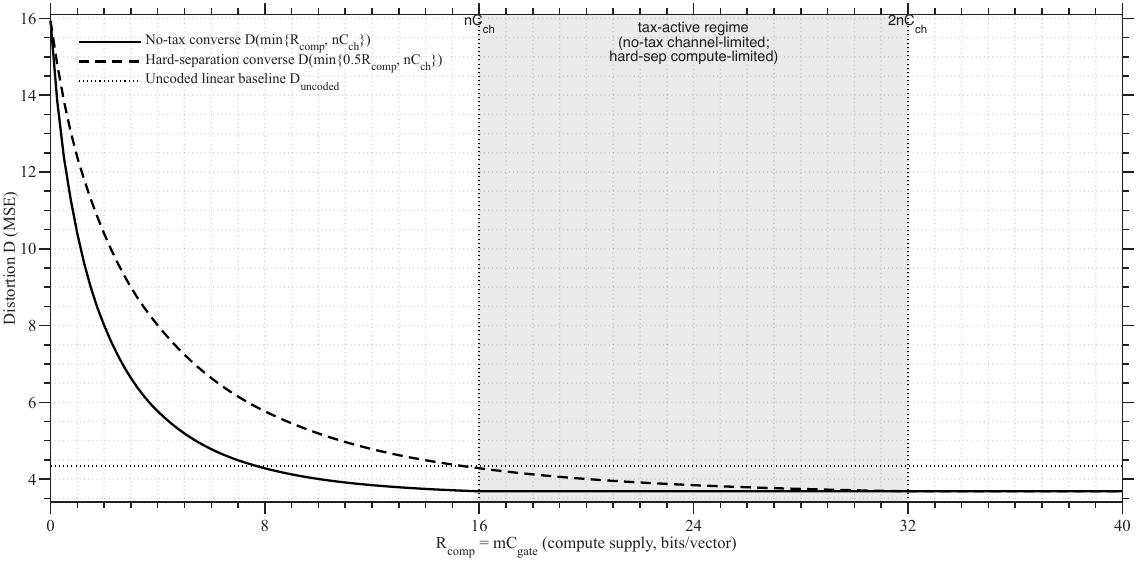}
\vspace{-3mm}
\caption{Diagonal vector Gaussian example for the $p=8$ spectrum described in the text, comparing the explicit uncoded linear baseline \eqref{eq:vector_uncoded_mse} with the water-filling indirect-RD converses under task-direct and hard-separation information supplies. Vertical markers at $R_{\mathrm{comp}}=n\Cch$ and $R_{\mathrm{comp}}=2n\Cch$ indicate the channel-limited transitions for the no-tax and hard-separation curves, respectively. The shaded region highlights the regime in which the no-tax system is already channel-limited while hard-separation remains compute-limited.}
\label{fig:vector_gap_uncoded_hard}
\end{figure}

\FloatBarrier
\section{Instantiation of Internal Redundancy Mechanisms on Noisy Hardware}
\label{app:internal_redundancy_instantiation}
\indent This appendix gives one concrete instantiation of the internal coding primitive assumed under Assumption~\ref{assump:achievability_benchmark} and formalized in Proposition~\ref{prop:compute_code_bsc}, using Varshney's region-to-use/stable-memory framework as an explicit example.

\begin{definition}[Region-to-use decoder (computable robustness region)]
Fix an iterative decoding/correcting network architecture $\mathsf{Dec}$ (e.g., a message-passing LDPC decoder, or an unrolled neural decoder) whose update rules are implemented on noisy components with parameter vector $\theta$ (e.g., register flips, gate flips, or word-level upsets).
For a target reliability level $\eta\in(0,1/2)$, the \emph{region-to-use decoder} $\mathcal{R}_\eta(\mathsf{Dec})$ is the set of component-noise parameters $\theta$ for which the associated asymptotic density-evolution/state-evolution recursion has an attracting fixed point whose bit error probability is at most $\eta$ when initialized from the operating point of interest.
Equivalently, inside $\mathcal{R}_\eta(\mathsf{Dec})$ the noisy decoder acts as a \emph{contraction} on error probability down to (at most) $\eta$.
\end{definition}

\begin{proposition}[Region-to-use $\Rightarrow \eta$-reliability $\Rightarrow$ stable memory and positive storage rate (Varshney)]
\label{prop:region_to_use_stable_memory}
Consider the noisy-register memory architecture with a noisy Gallager-A LDPC correcting network studied by Varshney~\cite{VarshneyTIT11}.
If the component-noise parameters lie in the region-to-use decoder (or its hypograph) for some $\eta$ that is below the decoding threshold of a noiseless ``silver decoder,''
then the correcting network achieves $\eta$-reliability and yields a \emph{stable} sequence of memories in the sense of~\cite[Defs.~10--11]{VarshneyTIT11}.
Moreover, for a Gallager-A correcting network based on an $(\lambda,\rho)$ LDPC ensemble with edge-perspective degree distributions $\lambda,\rho$, Varshney's Theorem~5 gives an explicit lower bound on the resulting storage capacity:
\begin{equation}
  \mathcal{C}_{\mathrm{store}}
  \ \ge\
  \frac{1-\lambda'(1)/\rho'(1)}{\lambda'(1)\rho'(1)-1},
  \label{eq:varshney_storage_capacity_lb}
\end{equation}
which reduces to $(1-(d_v-1)/(d_c-1))/((d_v-1)(d_c-1)-1)$ for a regular $(d_v,d_c)$ code.
In particular, $\mathcal{C}_{\mathrm{store}}>0$ whenever $1-\lambda'(1)/\rho'(1)>0$ (equivalently, $d_v<d_c$ for regular ensembles), implying that a positive rate of reliably retrievable bits can be maintained despite noisy storage and noisy iterative correction.
\end{proposition}

\begin{proof}
The existence of a computable region-to-use decoder and the implication ``region-to-use $\Rightarrow$ $\eta$-reliability'' follow from the density-evolution analysis in~\cite{VarshneyTIT11}.
Stability and the storage-capacity lower bound \eqref{eq:varshney_storage_capacity_lb} are stated explicitly in~\cite[Theorem~5]{VarshneyTIT11}.
\end{proof}

\begin{remark}[A concrete internal-redundancy example: unrolled (neural) correcting networks]
Proposition~\ref{prop:compute_code_bsc} establishes the existence of capacity-achieving internal codes under the memoryless-primitive abstraction.
An explicit \emph{computational} realization, closer to GPU implementations of channel decoding and learned inference, is to instantiate the internal redundancy by a \emph{noisy correcting network} that is itself an iterative message-passing computation graph.

A standard instance is an LDPC correcting network based on one round of a message-passing decoder (e.g., Gallager-A / belief propagation) applied periodically to a bank of noisy registers.
Such correcting networks admit a sparse deep-neural-network interpretation: unrolling $J$ iterations of message passing yields a depth-$J$ feedforward graph with local nonlinear updates, and learned variants simply replace fixed update rules by trainable weights.
Varshney~\cite{VarshneyTIT11} analyzes LDPC ensembles under \emph{faulty} iterative decoding and exhibits a computable region-to-use decoder; inside this region (and its hypograph) the noisy correcting network is $\eta$-reliable and induces a stable memory with positive storage capacity (Proposition~\ref{prop:region_to_use_stable_memory}); see also~\cite{HuangLiDolecek15}. We emphasize that the $\eta$-reliability/stability guarantees are currently proved for specific message-passing (graph-based) architectures; extending comparably sharp ``region-to-use'' characterizations to fully general black-box neural decoders remains an interesting direction.
In our framework, this provides a concrete example of an ``internal redundancy mechanism'' that protects task-critical digital state across repeated vulnerable materializations.

Viewed in the language of Section~\ref{sec:interface_tax_section}, a \emph{neural channel decoder} followed by a \emph{neural task network} is a hard-separation pipeline, so Theorem~\ref{thm:interface_tax} applies; an end-to-end network that retains channel-output access is instead task-direct or soft-interface, so Proposition~\ref{prop:no_tax_soft_interface} is the relevant baseline.
\end{remark}

\section{Supplementary Results for Learned Inference on Noisy Accelerators}
\label{app:learned_inference_supp}
\indent This appendix records a coarse proxy from network size/precision and data-movement overhead to the bit-materialization budget $m$, together with an MSE decomposition and a simple model-storage feasibility bound; under heterogeneous primitives or reliable islands, the relevant compute supply is the plug-in quantity from Lemma~\ref{lem:compute_info_bound_hetero}.
For a feedforward network with $P$ parameters evaluated at $b$-bit precision, let $b_{\mathrm{route}}$ denote a coarse bit-equivalent count of additional transport/routing/retiming/buffering events associated with making raw or soft side information available downstream, not already absorbed into the first two terms below.
A crude per-inference proxy is
\begin{equation}
  m \ \approx\  \kappa_{\mathrm{w}}\,P\,b \;+\; \kappa_{\mathrm{a}}\,b_{\mathrm{act}} \;+\; \kappa_{\mathrm{r}}\,b_{\mathrm{route}},
  \label{eq:m_mapping_nn}
\end{equation}
with $\kappa_{\mathrm{w}},\kappa_{\mathrm{a}},\kappa_{\mathrm{r}}$ absorbing reuse/caching, multicast/fan-out, and other implementation details.
In the language of Proposition~\ref{prop:soft_interface_tradeoff}, the last term provides a coarse implementation-level way to absorb non-negligible support overhead for forwarding raw or high-precision channel outputs into the reserved interface-support budget $m_{\mathrm{int}}$. The shorthand ``small interface cost'' therefore corresponds to a low-storage idealization rather than to physically free forwarding.

\paragraph{A decomposition: information-limited versus approximation-limited MSE}
\begin{proposition}[Inference MSE decomposition given an intermediate representation]
\label{prop:mse_decomp_representation}
Let $X$ be the task variable and let $U_{\mathrm{rep}}$ be \emph{any} intermediate representation available to the inference module (e.g., $U_{\mathrm{rep}}$ could be a decoded bitstream $\hat B$, a soft-output vector, or a learned latent).
For any (measurable) estimator $\hat X=g(U_{\mathrm{rep}})$ under squared loss,
\begin{equation}
  \mathbb{E}\big[(X-\hat X)^2\big]
  \;=\;
  \mathrm{mmse}(X\mid U_{\mathrm{rep}})
  \;+\;
  \mathbb{E}\Big[\big(\mathbb{E}[X\mid U_{\mathrm{rep}}]-g(U_{\mathrm{rep}})\big)^2\Big],
  \label{eq:mse_decomp_rep}
\end{equation}
where $\mathrm{mmse}(X\mid U_{\mathrm{rep}})\triangleq \mathbb{E}\big[(X-\mathbb{E}[X\mid U_{\mathrm{rep}}])^2\big]$.
\end{proposition}

\begin{proof}
Write $X-g(U_{\mathrm{rep}}) = (X-\mathbb{E}[X\mid U_{\mathrm{rep}}]) + (\mathbb{E}[X\mid U_{\mathrm{rep}}]-g(U_{\mathrm{rep}}))$ and expand the square.
The cross term has zero expectation because $\mathbb{E}[X-\mathbb{E}[X\mid U_{\mathrm{rep}}]\mid U_{\mathrm{rep}}]=0$.
\end{proof}

Proposition~\ref{prop:mse_decomp_representation} separates an information-limited term $\mathrm{mmse}(X\mid U_{\mathrm{rep}})$ from an approximation/implementation term. Combined with approximation-theoretic bit-complexity bounds~\cite{ElbrachterTIT21}, it yields the coarse co-design rule
\begin{equation}
  \text{target accuracy }\delta_{\mathrm{acc}}
  \quad\Longrightarrow\quad
  m\Cgate \ \gtrsim\  \mathrm{bits}(\delta_{\mathrm{acc}}),
  \label{eq:nn_accuracy_to_mCgate}
\end{equation}
where $\mathrm{bits}(\delta_{\mathrm{acc}})$ denotes the implementation-dependent number of vulnerable bit materializations needed to evaluate a network that achieves error $\delta_{\mathrm{acc}}$.

\paragraph{A reliability-limited model size bound (parameter integrity)}
Deep approximation theory characterizes the number of \emph{bits} needed to describe an accurate inference rule.
On a faulty accelerator, those bits must themselves be stored and retrieved reliably.
The following simple bound makes this requirement explicit.

\begin{proposition}[A reliability-limited model size bound on a faulty memory substrate]
\label{prop:model_storage_bound}
Let $\Theta\in\{1,\dots,2^{B_\Theta}\}$ be a uniformly distributed model index (e.g., encoding network topology and quantized weights), so $H(\Theta)=B_\Theta$ bits.
Suppose the model is stored as $M$ physical bits $U^M(\Theta)$ and later retrieved through independent bit flips
$Z_i = U_i\oplus N_i$, $N_i\stackrel{\text{i.i.d.}}{\sim}\mathrm{Bern}(\varepsilon)$.
Let $\hat\Theta=\phi(Z^M)$ denote the reconstructed model index.
Then the probability of model corruption $P_{\mathrm{err}}\triangleq \mathbb{P}[\hat\Theta\neq \Theta]$ satisfies
\begin{equation}
  P_{\mathrm{err}} \ \ge\ \max\!\left\{0,\ 1-\frac{M\Cgate+1}{B_\Theta}\right\}.
  \label{eq:model_storage_fano}
\end{equation}
In particular, achieving $P_{\mathrm{err}}\to 0$ as $B_\Theta\to\infty$ requires
$\frac{B_\Theta}{M}\le \Cgate+o(1)$; equivalently, any asymptotic storage rate strictly above $\Cgate$ is impossible.
\end{proposition}

\begin{proof}
Since $\Theta\to U^M\to Z^M\to \hat\Theta$, data processing gives
$I(\Theta;\hat\Theta)\le I(U^M;Z^M)$.
Since $Z^M$ is the output of $M$ independent BSC$(\varepsilon)$ uses with inputs $U^M$, we have
$I(U^M;Z^M)\le M\Cgate$.
Applying Fano's inequality gives
$P_{\mathrm{err}}\ge 1-\frac{I(\Theta;\hat\Theta)+1}{H(\Theta)}\ge 1-\frac{M\Cgate+1}{B_\Theta}$, and clipping at zero proves \eqref{eq:model_storage_fano}.
\end{proof}

Combining Proposition~\ref{prop:model_storage_bound} with approximation-theoretic bit-complexity bounds yields the same coarse feasibility requirement as in \eqref{eq:nn_accuracy_to_mCgate}.

\begin{remark}[Reliable-island plug-in]
Reliable islands are a direct plug-in specialization of Lemma~\ref{lem:compute_info_bound_hetero}: with a reliable fraction $\rho_{\mathrm{rel}}$ and an unreliable fraction $1-\rho_{\mathrm{rel}}$, replace $m\Cgate$ by
\begin{equation}
  S_{\mathrm{comp}}(\rho_{\mathrm{rel}})
  \triangleq
  \rho_{\mathrm{rel}}\,m + (1-\rho_{\mathrm{rel}})\,m\,\Cgate
  \quad
  (\text{or }\rho_{\mathrm{rel}}\,m + (1-\rho_{\mathrm{rel}})\,m\,\Cgateeff)
  \label{eq:Scomp_rho_rel}
\end{equation}
and then apply Theorem~\ref{thm:supply_demand_converse} for task-direct architectures or the same equal-splitting argument as in Theorem~\ref{thm:interface_tax} for hard-separation, yielding the corresponding task-direct and hard-separation plug-in bounds.
\end{remark}

\section{Auxiliary Lemma: Clipping Closes Mean-MSE Achievability}
\label{app:clipping_achievability_lemma}
\indent This appendix provides a short self-contained inequality used to close the achievability proof of Theorem~\ref{thm:achievability_general} under squared-error (unbounded) distortion.

\begin{lemma}[Clipping closes the mean-MSE achievability under unbounded distortion]
\label{lem:clip_closes_mse_achievability}
Assume squared-error distortion $d(x,\hat x)=(x-\hat x)^2$ and $\mathbb{E}[X^2]<\infty$ under the standing i.i.d.\ task model.
Fix a blocklength $T$ and an error event $\mathsf{E}_T$.
Let $\hat X_{\mathrm{RD}}^T$ denote any ``nominal'' reconstruction and let $\tilde X^T$ denote a (possibly corrupted) preliminary output satisfying
$\tilde X^T=\hat X_{\mathrm{RD}}^T$ on $\mathsf{E}_T^c$.
For $A>0$, define the coordinatewise clipped output $\hat X_t\triangleq \mathrm{clip}(\tilde X_t;A)\in[-A,A]$.
Then
\begin{equation}
\frac{1}{T}\sum_{t=1}^T \mathbb{E}\big[(X_t-\hat X_t)^2\big]
\le
\frac{1}{T}\sum_{t=1}^T \mathbb{E}\big[(X_t-\hat X_{\mathrm{RD},t})^2\big]
+ 8\,\mathbb{E}\!\big[X^2\mathbf{1}\{|X|>A\}\big]
+ 4A^2\,\mathbb{P}(\mathsf{E}_T).
\label{eq:clip_closes_mse_bound}
\end{equation}
In particular, if $\frac{1}{T}\sum_{t=1}^T \mathbb{E}[(X_t-\hat X_{\mathrm{RD},t})^2]\le D+\xi_T$ with $\xi_T\to 0$
and one can choose $A=A_T\to\infty$ such that $A_T^2\mathbb{P}(\mathsf{E}_T)\to 0$, then the clipped scheme achieves mean MSE $D$.
\end{lemma}

\begin{proof}
Decompose
\[
\frac{1}{T}\sum_{t=1}^T \mathbb{E}\big[(X_t-\hat X_t)^2\big]
=
\frac{1}{T}\sum_{t=1}^T \mathbb{E}\big[(X_t-\hat X_t)^2\mathbf{1}\{\mathsf{E}_T^c\}\big]
+
\frac{1}{T}\sum_{t=1}^T \mathbb{E}\big[(X_t-\hat X_t)^2\mathbf{1}\{\mathsf{E}_T\}\big].
\]

\emph{Success term.}
On $\mathsf{E}_T^c$, $\tilde X_t=\hat X_{\mathrm{RD},t}$.
Moreover, if $|X_t|\le A$, clipping cannot increase squared error:
$(X_t-\mathrm{clip}(u;A))^2\le (X_t-u)^2$ for any $u\in\mathbb{R}$.
Therefore,
\[
(X_t-\hat X_t)^2\mathbf{1}\{\mathsf{E}_T^c\}
\le
(X_t-\hat X_{\mathrm{RD},t})^2 + (|X_t|+A)^2\mathbf{1}\{|X_t|>A\}
\le
(X_t-\hat X_{\mathrm{RD},t})^2 + 4X_t^2\mathbf{1}\{|X_t|>A\},
\]
where the last step uses $|X_t|>A\Rightarrow (|X_t|+A)^2\le 4X_t^2$.

\emph{Error term.}
On $\mathsf{E}_T$, we only use $|\hat X_t|\le A$, so
\[
(X_t-\hat X_t)^2
\le (|X_t|+A)^2
\le 4A^2\,\mathbf{1}\{|X_t|\le A\}+4X_t^2\,\mathbf{1}\{|X_t|>A\}.
\]
Multiplying by $\mathbf{1}\{\mathsf{E}_T\}$ and taking expectations yields
$\mathbb{E}[(X_t-\hat X_t)^2\mathbf{1}\{\mathsf{E}_T\}]
\le 4A^2\mathbb{P}(\mathsf{E}_T)+4\mathbb{E}[X_t^2\mathbf{1}\{|X_t|>A\}]$.

Combining the two terms, averaging over $t$, and using the standing i.i.d.\ assumption gives \eqref{eq:clip_closes_mse_bound}.
Finally, $\mathbb{E}[X^2]<\infty$ implies $\mathbb{E}[X^2\mathbf{1}\{|X|>A_T\}]\to 0$ as $A_T\to\infty$, proving the stated consequence.
\end{proof}

\bibliographystyle{IEEEtran}
\bibliography{ref}

\end{document}